\documentclass[11pt]{article}
\usepackage[letterpaper, margin=0.8in]{geometry}

\usepackage{graphicx}
\usepackage[sort,nocompress]{cite}
\usepackage{amsfonts}
\usepackage{amsmath,amssymb}
\usepackage{url}
\usepackage{fontawesome}
\usepackage{wasysym}
\usepackage{amsthm}
\usepackage{enumitem}
\usepackage{wrapfig}

\usepackage[T1]{fontenc}

\usepackage{subcaption}

\usepackage{xpatch}
\usepackage[linesnumbered,ruled,vlined]{algorithm2e}
\xpretocmd{\algorithm}{\hsize=\linewidth}{}{}

\newcommand{\later}[1]{{}}
\newcommand{\old}[1]{{}}
\long\def\ignore#1{}

\newcommand{\curly}{\mathrel{\leadsto}}

\newtheorem{theorem}{Theorem}
\newtheorem{lemma}{Lemma}

\newtheorem{corollary}{Corollary}

\providecommand{\keywords}[1]{\textbf{{Keywords ---}} #1}

\begin{document}
\title{Engineering an algorithm for constructing low-stretch geometric graphs with near-greedy average-degrees\thanks{Research was supported by the NSF Award CCF-1947887.}}

\author{%
	FNU Shariful\\ 
	\small School of Computing\\
	\footnotesize University of North Florida, USA\\
	\footnotesize\tt \faicon{envelope}\texttt{ n01479223@unf.edu}\\
	\and
	Justin Weathers\\   
	\footnotesize School of Computing\\
	\footnotesize University of North Florida, USA\\
	\footnotesize\tt \faicon{envelope}\texttt{ n01501509@unf.edu}
	\and
	Anirban Ghosh\\      
	\footnotesize School of Computing\\
	\footnotesize University of North Florida, USA\\
	\footnotesize\tt \faicon{envelope}\texttt{ anirban.ghosh@unf.edu}\\
	\and
	Giri Narasimhan\\      
	\footnotesize  Knight Foundation School of Computing and Information Sciences\\
	\footnotesize Florida International University, USA\\
	\footnotesize\tt \faicon{envelope}{ giri@fiu.edu}
		\and
}

\date{}
\maketitle

\vspace{-35pt}

	\begin{abstract}
{\footnotesize We design and engineer \textsc{Fast-Sparse-Spanner}, a simple and practical (fast and memory-efficient) algorithm  for constructing
sparse low stretch-factor geometric graphs on large pointsets in the plane. To our knowledge, this is the
first practical algorithm to construct fast low stretch-factor graphs on large pointsets with average-degrees (hence, the number of edges) competitive with that of 
greedy-spanners, the  sparsest known class of Euclidean geometric spanners.

		To evaluate our implementation in terms of computation
		speed, memory usage,
		and quality of output, we performed extensive experiments with 
		synthetic and real-world pointsets, and by comparing it to our
		closest competitor \textsc{Bucketing}, the fastest known greedy-spanner algorithm for pointsets in the plane,
		devised by Alewijnse et al. (Algorithmica, 2017). 
	We always found that \textsc{Fast-Sparse-Spanner} generated near-greedy $t$-spanners while being fast and memory-efficient. 
	Our experiment with constructing a $1.1$-spanner on a large
	synthetic pointset with $128K$ points uniformly distributed within a
	square shows more than a $41$-fold speedup with roughly a third
	of the memory usage of that of \textsc{Bucketing}, but with only a 3$\%$
	increase in the average-degree of the resulting graph. When ran on a pointset with a million points drawn from the same distribution, we observed a $130$-fold speedup, with roughly a fourth
	of the memory usage of that of \textsc{Bucketing}, and just a $6\%$ increase in the average-degree. 
		In terms of diameter, the graphs generated by \textsc{Fast-Sparse-Spanner} beat greedy-spanners in most cases (have substantially lower diameter) while maintaining near-greedy average-degree. Further, our algorithm can be easily parallelized to take advantage of parallel  environments.
		
		As a byproduct of our research, we design and engineer \textsc{Fast-Stretch-Factor}, a practical parallelizable algorithm that can
		measure the stretch-factor of any graph generated by \textsc{Fast-Sparse-Spanner}.
		Our experiments show 
		that it is much faster than the naive Dijkstra-based stretch-factor measurement algorithm. 

We share the implementations via \textsf{GitHub} for broader uses and future research.

\textbf{GitHub repository.} \url{https://github.com/ghoshanirban/FSS}

			\keywords{geometric graph, sparse graph, geometric spanner, stretch-factor, algorithm engineering, experimental algorithmics}
		}
	\end{abstract}
	
	\vspace{-20pt}
	
	\section{Introduction}\label{sec:intro}
	
	Let $G$ be the complete Euclidean graph on a given set $P$ of $n$ points embedded in the plane. A
	\emph{geometric $t$-spanner}, or simply a \emph{$t$-spanner}, on $P$ is a geometric graph $H := (P, E)$, a subgraph of $G$ such that for every pair of
	points $p,q \in P$, the Euclidean length of a shortest path  between them in $H$
	is at most $t$ times their Euclidean distance $|pq|$, for some $t \geq 1$. The parameter $t$ is referred to as the \emph{stretch-factor} of $H$. 
	Thus, the subgraph $H$ approximately preserves pairwise shortest path distances for all point pairs in $P$. 
	The complete graph $G$ itself is a $1$-spanner
	with $\Theta(n^2)$ edges. Clearly, for large values of $n$, complete graphs are unsuitable for practical purposes because of their sheer size. So, the main focus in the research of geometric spanners is to construct subgraphs having $o(n^2)$ edges and guarantee various structural properties. Refer to \cite{narasimhan2007geometric,bose2013plane} for an introduction to geometric spanners and~\cite{gao2001geometric,czumaj2003fault,marble2013asymptotically,alzoubi2003geometric,russel2005exploring,rao1998approximating} for their uses. In this work, our objective is to engineer a practical algorithm that can construct low stretch-factor geometric spanners with low average-degrees\footnote{The \emph{average-degree} of a graph $H:=(V,E)$ is defined as $2|E|/|V|$.}.  
	
	Despite intensive theoretical studies on geometric spanners, how to construct sparse geometric spanners fast in practice on large pointsets with low average-degrees and low stretch-factors (less than $2$) remains unknown. Such spanners are used in real-world applications, whereas spanners with large stretch-factors are mainly of theoretical interest in computational geometry. In theoretical studies of spanners, bounds on the number of edges are typically expressed using asymptotic notations. As a result, owing to large hiding constants, many well-known algorithms are found to produce poor-quality spanners in practice, although novel asymptotic bounds have been obtained to prove them efficient. Farshi and Gudmundsson~\cite{farshi2010experimental} were the first to perform a meticulous experimental study on geometric spanners. They found that the popular spanner algorithms suffer from at least one of the following limitations: (i) slow on large pointsets, (ii) memory-intensive, and (iii) places a high number of edges. Their study revealed a serious gap between the theory and practice of geometric spanners. 
	
	Xia showed that $L_2$-Delaunay triangulations are $1.998$-spanners~\cite{xia2013stretch} with at most $3n$ edges. Practical algorithms for constructing such triangulations exist; see~\cite{toth2017handbook}.  
	There is a separate family of fast spanner construction algorithms that can create sparse bounded-degree plane spanners by using Delaunay triangulation (not necessarily using the $L_2$-norm) as the starting point; see~\cite{anderson2022bounded, anderson2021interactive, bose2005constructing, dumitrescu2016lower, dumitrescu2016lattice,kanj2017degree}. Unfortunately, in those algorithms, $t$ cannot be set arbitrarily.

		Experiments and theoretical
	studies~\cite{le2022truly,le2022greedy,sigurd2004construction,soares1994approximating,smid2009weak,farshi2010experimental,evans2023path,eppstein2021edge,borradaile2019greedy,alewijnse2015computing,filtser2020greedy}
	have shown that greedy-spanners (a class of spanners whose
	construction uses a greedy strategy), originally proposed by
	Alth{\"o}fer et al.~\cite{althofer1993sparse} and Chandra et
	al.~\cite{chandra5new},  are unbeatable, especially when one desires to minimize
	the number of edges  (alternatively, average-degree), for any desired stretch-factor. See
	Algorithm~\ref{alg:greedy} for a description of the
	original greedy algorithm, popularly known as {\sc Path-Greedy} 
	and Fig.~\ref{fig:greedy} for greedy-spanner samples for various values of $t$. {\sc Path-Greedy}  is a generalization of the folkore Kruskal’s minimum spanning tree algorithm and produces sparse spanners. A demonstration applet for the algorithm can be found in~\cite{farshi2016visualization}.
	
		\begin{algorithm}[H]
		Sort and store the $\binom{n}{2}$ pairs of distinct points in non-decreasing order of their distances  in  a list $L$\;
		Let $H$ be an empty graph on $P$\; 
		\For{\upshape  each edge $\{p_i,p_j\} \in L$}{
			\If{\upshape the length of a shortest path in $H$ between $p_i,p_j$ is greater than $t \cdot |p_ip_j|$}{
				Place the edge $\{p_i,p_j\}$ in $H$\;
			}
		}
		\Return $H$;
		\caption{\textsc{Path-Greedy}$(P,t> 1)$}
		\label{alg:greedy}
	\end{algorithm}
	
	\begin{figure}[h]
		
		\centering
		\includegraphics[scale=0.35]{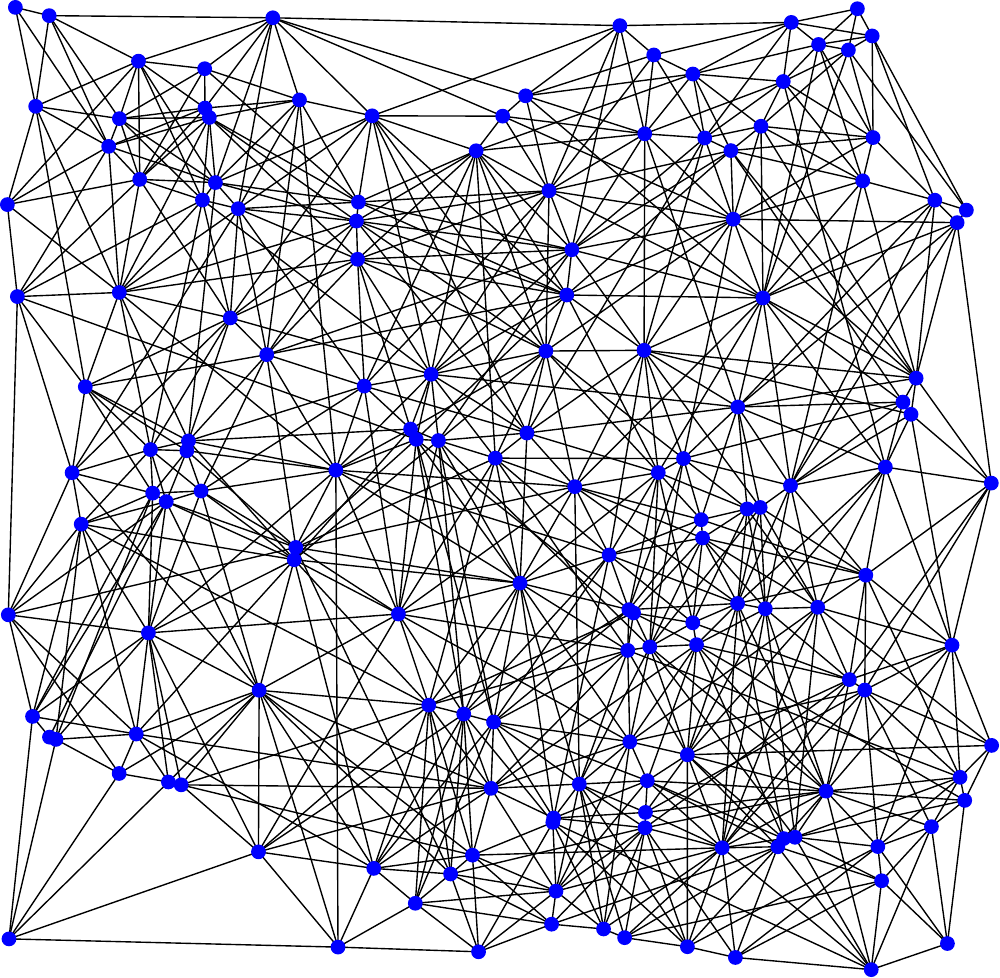} \hspace{5mm}
		\includegraphics[scale=0.35]{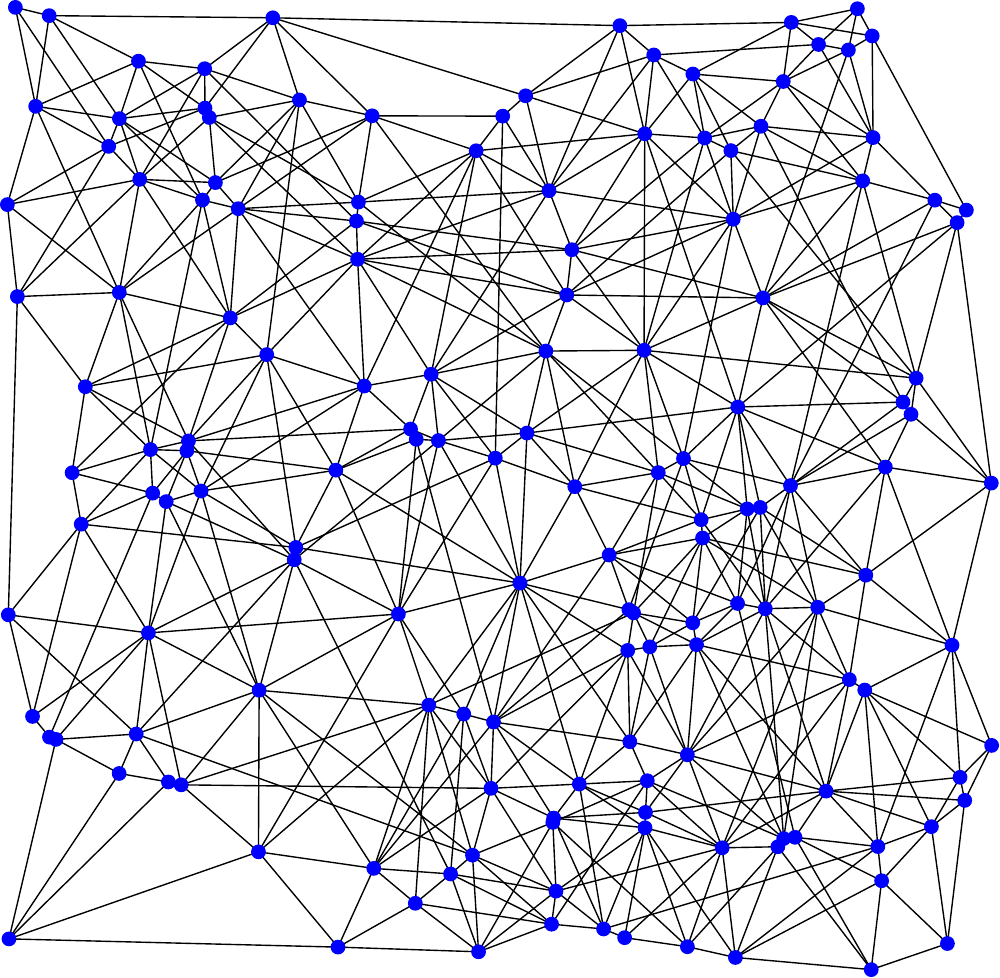} \hspace{5mm}
		\includegraphics[scale=0.35]{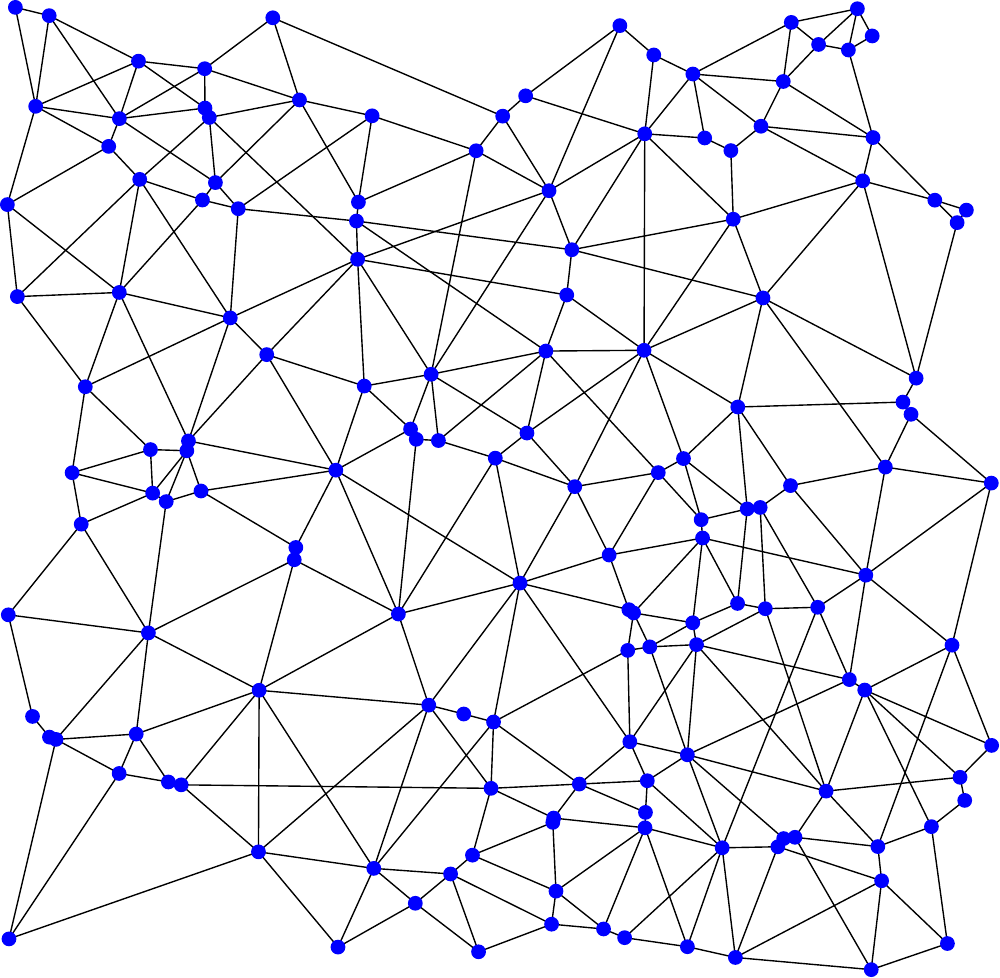} \hspace{5mm}
		\includegraphics[scale=0.35]{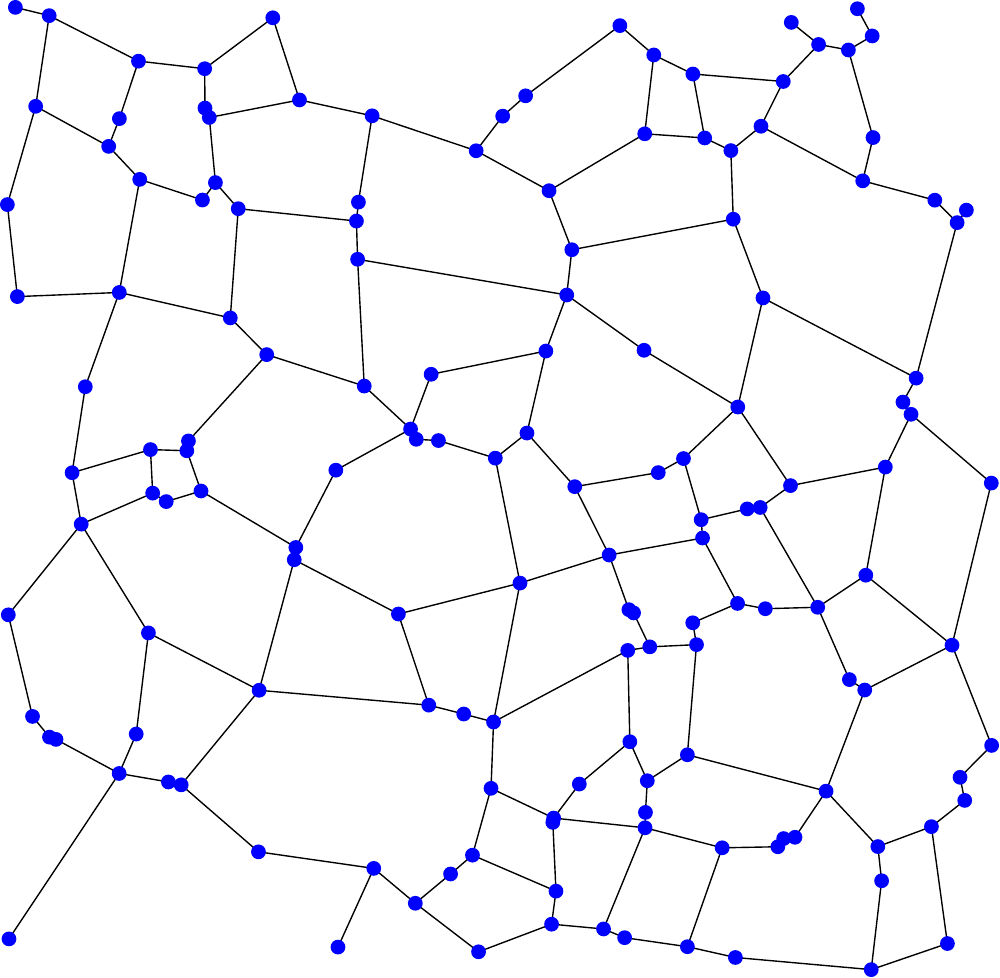}

		\caption{Four greedy-spanners on a set of $150$ points in the plane are shown, having stretch-factors of $1.05, 1.1, 1.25, 2$, respectively. In this case, the complete graph $G$ has $\binom{150}{2}  = 11,475$ edges.}
		\label{fig:greedy}
	\end{figure}

	The first attempt to construct near-greedy sparse spanners fast
	was made by Das and Narasimhan back in $1997$~\cite{das1997fast}. Their algorithm \textsc{Approximate-Greedy} runs in $O(n\log^2 n)$-time and is thus substantially faster than  \textsc{Path-Greedy}, that runs in $O(n^3\log n)$ time. Unfortunately, it was observed in~\cite{farshi2010experimental} that, in practice, the average-degrees of the graphs produced by \textsc{Approximate-Greedy} are far from that of the actual greedy-spanners. A slightly faster $O(n\log n)$-time algorithm was designed in~\cite{gudmundsson2002fast} by Gudmundsson, Levcopoulos, and Narasimhan to approximate greedy-spanners. Since the algorithm is quite involved, we did not implement it for comparison purposes. We observe that the algorithm is designed along the lines of \textsc{Approximate-Greedy}, and therefore, we believe that it will  unlikely generate  spanners sparser than the ones generated by   \textsc{Approximate-Greedy}. 

	Farshi and Gudmundsson~\cite{farshi2010experimental} presented a simple modification of the original   \textsc{Path-Greedy} algorithm to
	make it faster in practice by reducing the number of single-source
	shortest path computations using an additional matrix; see
	Algorithm~\ref{alg:fg-greedy}. It is sometimes referred to as
	\textsc{FG-Greedy}. 	But in theory, both the \textsc{Path-Greedy} and \textsc{FG-Greedy} algorithms run in $O(n^3\log n)$ time.  Further, due to the use of $\binom{n}{2}$ edges in the
	{\sc Path-Greedy} and {\sc FG-Greedy} algorithms (see Step 1), 
	$\Theta(n^2)$ extra space is needed for their execution. Consequently,
	they are unusable for large pointsets.

		\begin{algorithm}[H]
		Sort and store the $\binom{n}{2}$ pairs of distinct points in non-decreasing order of their distances  in  a list $L$\;
		Let $H$ be an empty graph on $P$\; 
		\For{\upshape $(p_i,p_j) \in P \times P$} {
			$\texttt{weight}(p_i,p_j) = \infty$;
		}
		
		\For{\upshape each edge $\{p_i,p_j\} \in L$}{
			\If{\upshape \texttt{weight}$(p_i,p_j) > t \cdot |p_ip_j|$}{
				Compute single-source shortest path with source $p_i$ in $H$\; 
				
				\For{\upshape $q \in P$} {
					Update \texttt{weight}$(p_i,q)$ and \texttt{weight}$(q,p_i)$ to the weight of the shortest path found between $p_i$ and $q$ in the previous step\;
				}

				\If{\upshape\texttt{weight}$(p_i,p_j) > t \cdot |p_ip_j|$} {
					Place the edge $\{p_i,p_j\}$ in $H$\;
				}

			}
		}
		\Return $H$;
		\caption{\textsc{FG-Greedy}$(P,t> 1)$}
		\label{alg:fg-greedy}
	\end{algorithm}

		In the same paper~\cite{farshi2010experimental}, Farshi and Gudmundsson experimentally
	found that for low values of $t$
	(for example $t \le 1.25$), 
	greedy-spanners always tend to be considerably sparser than other
	popular kinds of spanners such as the
	$\Theta$-\textsc{Graphs}~\cite{bose2004ordered,
		clarkson1987approximation,keil1988approximating},
	WSPD-spanners~\cite{callahan1993faster},
	\textsc{Sink}-spanners~\cite{arya1999dynamic}, \textsc{Skip-list}-spanners~\cite{arya1999dynamic}, and \textsc{Approximate-Greedy}-spanners~\cite{das1997fast}.

	  	Bose et al.~\cite{bose2010computing} devised a faster algorithm for
	constructing exact greedy-spanners in $O(n^2\log n)$ time, but
	unfortunately, like \textsc{Path-Greedy} and \textsc{FG-Greedy}, their
	algorithm uses $\Theta(n^2)$ space as well. They also showed that
	\textsc{FG-Greedy} 
	runs in $\Theta(n^3\log n)$ time since previously it was suspected
	that it runs in $o(n^3\log n)$ time.  
			
				Alewijnse et al.~\cite{alewijnse2015computing} presented an
	$O(n^2\log^2 n)$-time algorithm for constructing greedy-spanners.
	Although the algorithm is asymptotically slower than the algorithm
	designed by Bose et al.~\cite{bose2010computing}, it uses linear
	space. It was a significant breakthrough in this line of research. A
	faster algorithm, named \textsc{Bucketing}, was presented by Alewijnse et al.
	in \cite{alewijnse2017distribution} that runs in $O(n\log^2 n \log^2
	\log n)$ expected time. They observe that greedy-spanners are made up
	of `short' and `long' edges.
	Short edges are identified using a bucketing scheme, while long edges
	are computed using a WSPD (well-separated pair decomposition).
	They also experimentally showed that their algorithm is considerably
	faster than the one proposed in~\cite{alewijnse2015computing} and uses
	a reasonable amount of memory, making it the best greedy-spanner
	algorithm so far. For this reason, we have excluded the algorithm proposed in~\cite{alewijnse2015computing} in our experiments.
	Although being capable of producing arguably the best quality class of geometric spanners, we found that \textsc{Bucketing}
	becomes very compute-intensive when $t$ approaches $1$ (such $t$-values are more useful in practice), 
	making it slow on large pointsets.  
	For instance, it takes around $100$ minutes on a fast modern-day
	computer equipped with an \textsf{i9-12900K} processor, to compute a $1.1$-spanner on a $128K$-element pointset uniformly
	distributed inside a 
	square. For a million points drawn from the same distribution, it takes around five days on the same machine to run to completion, making it impractical for large pointsets.

	Recently, Abu-Affash et al.~\cite{abu2022delta} devised a new
	algorithm named $\delta$-\textsc{Greedy} that can produce greedy-like spanners
	in $O(n^2\log n)$ time. When $\delta$, a parameter used in the algorithm, is set to $t$, it
	outputs spanners identical to that of \textsc{Path-Greedy}. 
	Although experimentally shown to be
	speedy in practice,
	it remains impractical 
	for large values of $n$ 
	since it uses $\Theta(n^2)$ extra space. In fact, Step 1 in their algorithm is precisely the same as the one in \textsc{Path-Greedy} and \textsc{FG-Greedy}, contributing to quadratic space complexity. We found that $\delta$-\textsc{Greedy}, \textsc{Path-Greedy},  \textsc{FG-Greedy}, and the $O(n^2\log n)$-time algorithm presented by Bose et al.~\cite{bose2010computing} are impractical and run out of memory (on a machine equipped with $32$ GB of main memory) when $n$ is large. For instance, when $n=100K$, to store the complete graph edges in a list, we need $2\cdot \binom{100K}{2} \cdot 4 \text{ bytes} = 40~\text{GB}$ of main memory, assuming every edge is represented using a pair of $4$-byte integers. For this reason, we have excluded the above four algorithms from our experiments.
	\subsection{Our contributions}
	\begin{enumerate}\itemsep0pt
	
			\item  Motivated by the
	lack of a fast near-greedy-spanner construction algorithm for handling large pointsets,
	we have designed and engineered a fast and memory-efficient  construction
	algorithm that we have named \textsc{Fast-Sparse-Spanner} (see
	Section~\ref{sec:fss}). The algorithm is simple and intuitive
        and has been designed by observing how sparse spanners look in
        practice as opposed to the traditional approaches where
        obtaining theoretical bounds is the main focus. Of course,
        existing theory on spanners have been leveraged to design the
        algorithm.
        While, theoretically, the algorithm is not guaranteed to
        produce $t$-spanners, we conjecture in Section \ref{sec:argue}
        that for uniform distributions, its output is a $t$-spanner
        with high probability. 
In our experiments with more than $25K$ trials using synthetic and real-world pointsets of varied sizes, we
	found that it always produced fast, near-greedy sparse 
	$t$-spanners, even for $t$ as low as $1.05$. This inspired us to name our algorithm \textsc{Fast-Sparse-Spanner}. In real-world applications, sometimes fast and memory-efficient constructions of sparse geometric graphs is a necessity even if their actual stretch-factors are close the desired stretch-factors. Thus, our algorithm is still useful even if it misses stretch-factors in some cases (in our experiments, we never found so). We believe that if it ever misses the target stretch-factor, the actual stretch-factor of the generated geometric graph will not be far away from the target. \textsc{Fast-Sparse-Spanner} can  easily  leverage modern-day parallel environments (multi-core CPUs, for instance) for scalability.
	
	Using a preliminary experiment by comparing it with $\Theta$-\textsc{Graphs}, WSPD-spanners, \textsc{Sink}-spanners, \textsc{Skip-list}-spanners, \textsc{Approximate-Greedy}-spanners, and the spanners produced by \textsc{Bucketing}, we show that \textsc{Bucketing} is our closest competitor when average-degree is the main concern. Thereafter, we  compare
	\textsc{Fast-Sparse-Spanner} with \textsc{Bucketing} and show real-world efficacy of our algorithm by running it on large synthetic and real-world pointsets.
	\textsc{Fast-Sparse-Spanner} can construct near-greedy sized geometric graphs and is considerably 
	faster  than \textsc{Bucketing}, especially for low stretch-factors. For instance, on our machines, it could construct a $1.1$-spanner on a $128K$-element uniformly distributed pointset inside a square within $2.5$ minutes and used $65\%$ less extra memory compared to \textsc{Bucketing} while placing only $3\%$ extra edges. For a million points drawn from the same distribution, it finished its execution in around $55$ minutes using around $1.3$ GB of main memory, making it remarkably faster than \textsc{Bucketing} which took around 5 days and used $5.8$ GB of main memory. In this case, we observed just a $6\%$ increase in the number of edges compared to the graph produced \textsc{Bucketing}. 
	
	In most cases, we found that the spanners generated by our algorithm have substantially lower diameters than the greedy-spanners while maintaining near-greedy average-degrees. For instance,  $1.1$-spanners generated by \textsc{Fast-Sparse-Spanner} on $128K$-element uniformly distributed pointsets have diameters around $25$, whereas greedy-spanners have around $200$. 

	We compare the spanners generated by our algorithm to the ones produced by \textsc{Bucketing} in terms of  
	average-degree and diameter\footnote{The \textit{diameter} of a graph $G$ is the length (number of edges) of the longest shortest path among all vertex pairs in $G$.} (see Section~\ref{sec:exp}).   For broader uses and future research in this direction, we share our engineered \textsf{C}\texttt{++} code via \textsf{GitHub}.  

	\item  To our knowledge, there are no practical algorithms for measuring the
	stretch-factors of large spanner networks. However, an attempt has been made by Narasimhan and Smid to approximate stretch-factors of geometric spanners in~\cite{narasimhan2000approximating}. Chen et al.~\cite{cheng2012approximating} devised algorithms for approximating average stretch-factors of geometric spanners. Stretch-factor measurement algorithms for special types of geometric graphs such as paths, cycles,  trees, and plane graphs, can be found in~\cite{agarwal2008computing,wulff2010computing,anderson2022bounded}. Klein et al.~\cite{klein2009dilation} devised algorithms for reporting all pairs of vertices
	whose stretch factor is at least some given value $t$, when the input graph is a geometric path, tree, or cycle.
	Distance oracles can be used to estimate stretch-factors of geometric graphs. In this regard, Gudmundsson et al.~\cite{gudmundsson2008approximate} presented distance oracles for geometric spanners. For a discussion on shortest path queries for static weighted graphs including distance oracles, refer to the survey~\cite{sommer2014shortest} by Sommer.
	Surprisingly, the only known easy-to-engineer stretch-factor measurement  algorithm that works for any type of Euclidean graph, runs in $O(n^2\log n + n|E|)$ time by running the folklore Dijkstra's algorithm from every node~\cite{narasimhan2000approximating}, making it very slow for large spanners. When the input graph is sparse, this approach runs in $O(n^2\log n)$ time. Use of the Floyd-Warshall algorithm is even more impractical since it uses $\Theta(n^2)$ extra space. Sub-cubic all-pairs shortest path (APSP) algorithms exist, but those are involved and are seemingly challenging to engineer; see the paper~\cite{chan2010more} by Chan and the references therein. A $O(n^{2.922})$-time stretch-factor measurement algorithm follows directly from the $O(n^{2.922})$-time APSP algorithm presented by Chan in the same paper for Euclidean graphs.

	In this work, by exploiting the construction of the spanners generated by \textsc{Fast-Sparse-Spanner}, we have designed a new practical algorithm named \textsc{Fast-Stretch-Factor} that can compute stretch-factor of any spanner generated by \textsc{Fast-Sparse-Spanner} (see Section~\ref{sec:sf}). Further, it is easily parallelizable. Our experiments found that it is substantially faster than the naive Dijkstra-based algorithm (see Section~\ref{sec:exp}). For instance, for a spanner constructed on a $1M$-element pointset uniformly distributed  inside a square and $t=1.1$, \textsc{Fast-Stretch-Factor} ran to completion in around $47$ minutes, whereas the Dijkstra-based algorithm took around $22$ hours when $4$ threads were used in both the cases. 

\end{enumerate}

	\subsection{Preliminaries and notations}
	
	Let $P$ and $Q$ be two finite pointsets and $s$ be a positive real number. We say that $P$ and $Q$ are \emph{well-separated} with respect to $s$, if there exist two congruent disjoint disks $D_P$ and $D_Q$, such that $D_P$ contains the  bounding-box of $P$, $D_Q$ contains the  bounding-box of $Q$, and the distance between $D_P$ and $D_Q$ is at least $s$ times the common radius of $D_P$ and $D_Q$. The quantity $s$ is referred to as the \emph{separation ratio} of the decomposition. 
	
	Using the above idea of well-separability, one can define a well-separated decomposition (WSPD) \cite{callahan1995decomposition}  of a pointset in the following way. Let $P$ be a set of $n$ points and $s$ be a positive real number. A WSPD for $P$ with respect to $s$ is a collection of pairs of non-empty subsets of $P$, $$\{A_1,B_1\},\{A_2,B_2\},\ldots,\{A_m,B_m\}$$ for some integer $m$ (referred to as the size of the {WSPD}) such that 
	for each $i$ with $1 \leq i \leq m$, $A_i$ and $B_i$ are well-separated with respect to $s$, and
	for any two distinct points $p,q \in P$, there is exactly one index $i$ with $1 \leq  i \leq m$, such that  $p \in A_i, q \in B_i$, or $p \in B_i, q \in A_i$. 
	
	Given a pointset $P$ and $t>1$, a WSPD $t$-spanner~\cite{callahan1993faster,narasimhan2007geometric,ghosh2022visualizing} on $P$ is constructed in the following way. Start with an empty graph $H$ on $P$. Let $s = 4(t+1)/(t-1)$. Construct a {WSPD} of $P$ with separation ratio $s$.
	For every pair $\{A_i,B_i\}$ of the decomposition, include the edge $\{a_i,b_i\}$ in $H$, where $a_i$ is an arbitrary point in $A_i$ and $b_i$ is an arbitrary point in $B_i$. 
	
	Given a pointset $P$ and a positive integer $k$, a \emph{region quad-tree}~\cite{finkel1974quad} (or simply, a \emph{quad-tree}) on $P$ is a $4$-ary tree in which every node has either four children or none and is constructed by partitioning the bounding-box $B$ of $P$ into four equal-sized quadrants and creating four corresponding children attached to the root of the tree. 
	For each child with more than $k$ points from $P$ inside its quadrant, recursively create a quad-tree rooted at the child by dividing its quadrant into four equal-sized sub-quadrants. With every leaf,  store the points in $P$ located inside it.
	
	A leaf of a quad-tree is said to be \emph{empty} if it does not contain any point in $P$.
	
	Let $H$ be an Euclidean graph on $P$ and  $u,v \in P$. A path $\rho(u,v)$ between $u,v$ in $H$ is a $t$-path for $u,v$ if the Euclidean length of the path $\rho(u,v)$ is at most $t$ times $|uv|$. It can be observed that $H$ is a $t$-spanner if and only if at least one $t$-path exists in $H$ for every pair of points in $P$.

	\section{A  new practical hybrid  algorithm }
	\label{sec:fss}
	
	It is well known 
	that greedy algorithms are found to be slow in
	practice on large pointsets, especially for low stretch-factors  but are
	best for generating extremely sparse spanners. 	This motivated us to  construct local greedy-spanners on small disjoint chunks of $P$ to reduce the overall runtime. 
	Then, we carefully and efficiently ``stitch'' (merge) 
	the local greedy-spanners 
	into one final graph using long and short edges. 
	In practice, we found that the number of edges needed to merge them is
	not too many. This gives us a practical geometric graph construction algorithm
	that is competitive with the greedy-spanner algorithms with respect to 
	the number of edges placed, while being fast and memory efficient. Our algorithm, 
	\textsc{Fast-Sparse-Spanner}  uses \textsc{FG-Greedy}, WSPD-spanner,
	and a layering strategy to incrementally compute $H$, 
	a subgraph of $G$ (the complete Euclidean graph on $P$). Refer to Algorithm~\ref{alg:fss}.
	
	\begin{algorithm}[H]
		Let $H$ be an empty graph on $P$\;
		
		Create a quad-tree $T$ on $P$ such that the size of every leaf is at most $k$\;

		$S = \emptyset$;
		
		\For{\upshape \textbf{each} non-empty leaf $\sigma \in T$}{
			$P_\sigma = P \cap \sigma$\;
			Construct a greedy-spanner $H_\sigma$ on $P_\sigma$ using \textsc{FG-Greedy$(P_\sigma,t)$} (Algorithm~\ref{alg:fg-greedy})\; 
			$H = H \cup H_\sigma$\;
			$S = S \cup \ell_{\sigma}$, where $\ell_\sigma$ is the leader point of $\sigma$\;
		}
		
		Create a WSPD $t'$-spanner $W$ on $S$ and add 
		the edges of $W$ in $H$\;

		Let $G_T$ be the dual graph on the leaves of $T$, where the edge $\{\sigma_i,\sigma_j\} \in G_T$ if and only if the leaves $\sigma_i,\sigma_j$ are neighbors\;

		\For{\upshape every non-empty leaf $\sigma \in T$ }{
			
			\For{\upshape every non-empty neighbor $\sigma'$ of $\sigma$ in $G_T$}{
				\If{\upshape$\sigma,\sigma'$ are both non-empty and have  not been merged yet}{
					$\textsc{Greedy-Merge}(\sigma,\sigma',t,H)$ (Algorithm~\ref{alg:greedymerge})\;
				}
			}
			
		}
		\For{\upshape $i=2$ to $h$}{ 		
			\For{\upshape \textbf{each} non-empty leaf $\sigma \in T$}{
				
				\For{\upshape \textbf{each} non-empty leaf $\sigma'$  that is $i$ hops away from $\sigma$ in $G_T$}{  
					\If{\upshape$\sigma,\sigma'$ are both non-empty and have not been merged yet}{
						$\textsc{Greedy-Merge-Light}(\sigma,\sigma',t,H)$ (Algorithm~\ref{alg:greedymergelight})\;
					}
				}
			}
		}
		
		\Return $H$;
		\caption{\textsc{Fast-Sparse-Spanner}$(P,t > 1,t'\geq t,k\in \mathbf{Z}^+,h \in \mathbf{Z}^+)$}
		\label{alg:fss}
	\end{algorithm}
	
	In Section~\ref{sec:exp}, we present precise values for the three parameters $t',k,h$ used in our experiments. 
	In what follows, we provide details on the steps. 
	
	\begin{itemize}\itemsep0pt
	\item \textit{Step 1.}  Create a region quad-tree $T$ on $P$ where the size of every leaf is at most $k$. Later, in our experiments, we fix a suitable value for $k$. The region quad-tree aids in partitioning $P$ into almost same-sized disjoint chunks. 
		
		\item \textit{Step 2.}  For every non-empty leaf $\sigma$ of $T$, we create a local greedy $t$-spanner $H_\sigma$ on $P_\sigma$ using \textsc{FG-Greedy} (Algorithm~\ref{alg:fg-greedy}), where $P_\sigma = P \cap \sigma$. All the edges of the local greedy-spanners are placed in $H$.   Further, for every leaf $\sigma \in T$, we find its \textit{leader} point $\ell_\sigma$, a point in $P_\sigma$ that is closest to the center of the bounding-box of $P_\sigma$. In the cases of ties, choose arbitrarily.
		
		\item \textit{Step 3.}  Let $S := \cup_{\sigma \in T}
		\ell_\sigma$. Create a WSPD-spanner $W$ on $S$
		with stretch-factor $t'\geq t$ using the $O(n\log n)$ time
		algorithm presented in~\cite{narasimhan2007geometric}; an 
		efficient implementation can also be found in~\cite{nz-geomst-01}.  Include the edges
		of $W$ in $H$. 
		This creates a strong network on the leaders of the
		non-empty leaves. Consequently, the non-empty leaves of $T$
		are now connected using the edges of $W$. See Fig.~\ref{fig:qt} (left) for an illustration.
		
		\begin{figure}[ht]
			\centering
			\includegraphics[scale=0.35]{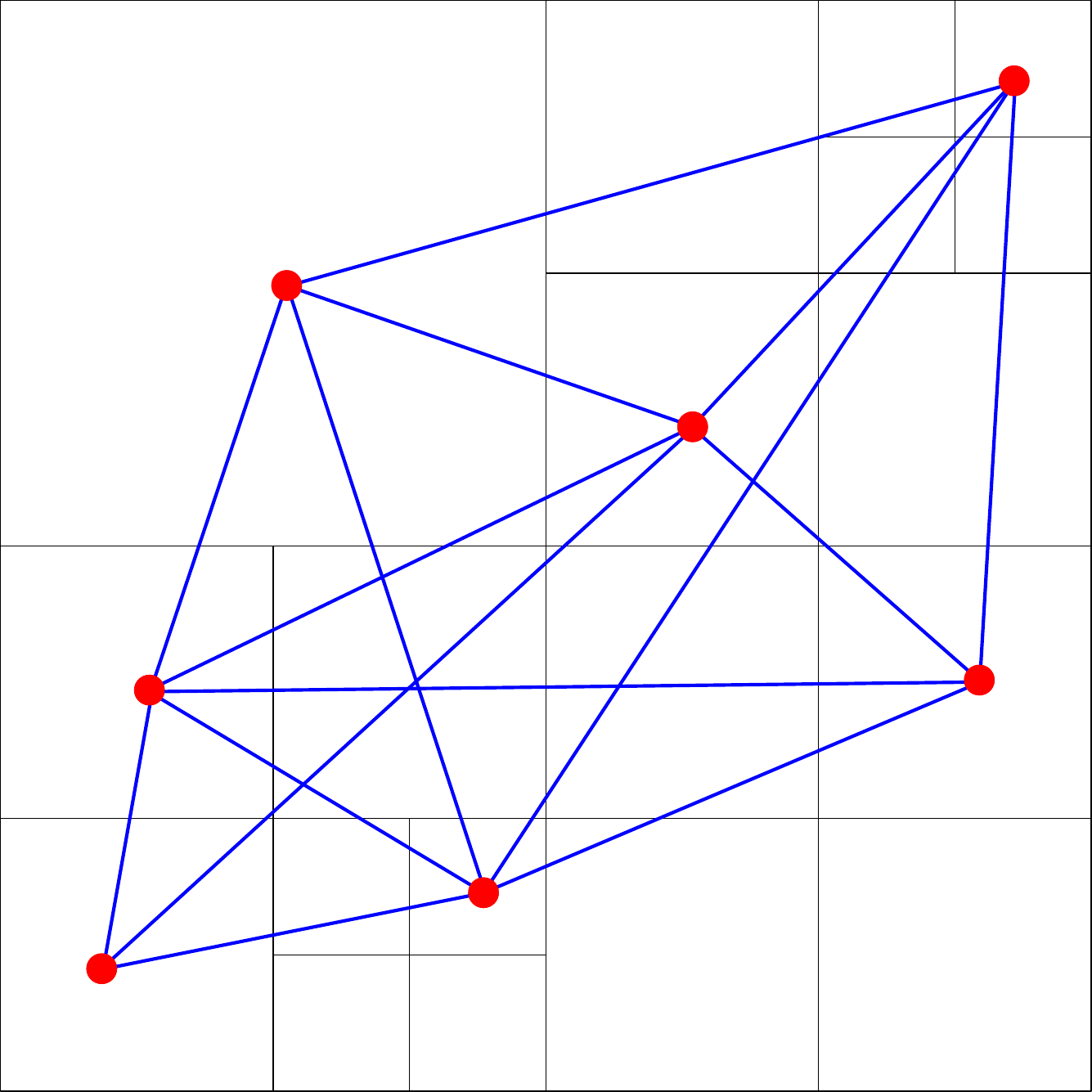}		\hspace{20mm}	\includegraphics[scale=0.35]{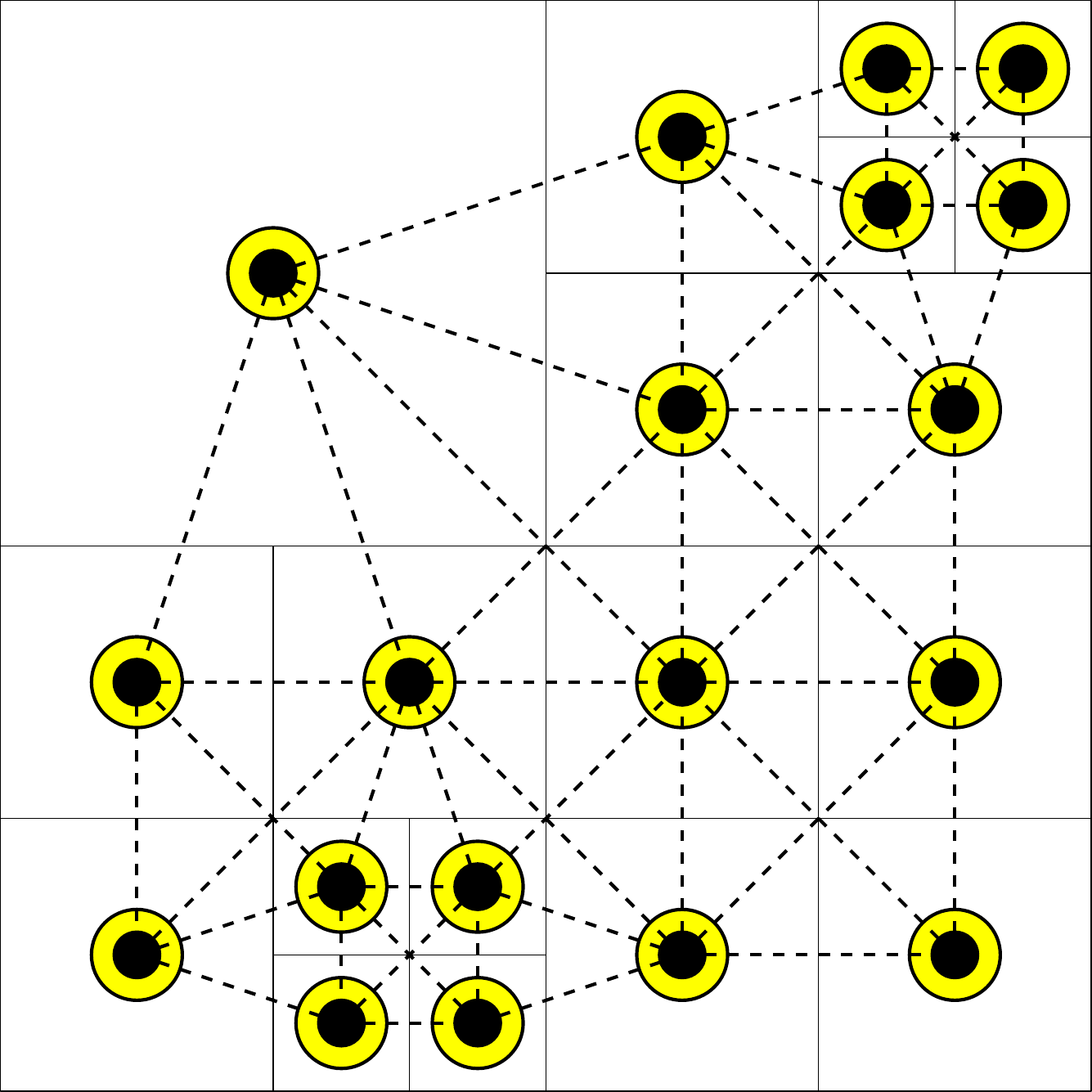}
			\caption{Left: A WSPD-spanner on the leaders of the non-empty leaves. Right: The dual graph $G_T$ on the leaves of the quad-tree.}
			\label{fig:qt}
		\end{figure}
		
		Unlike greedy-spanners, WSPD-spanners tend to have long edges. In our case, this is beneficial since such long WSPD edges also help reduce the diameter of the final output spanner. For faraway point pairs, shortest paths use the long WSPD edges placed in this step. As a result, path-finding algorithms (described next) terminate fast in most cases.
		Note that after this step, $H$ becomes a connected graph on $P$ since the intra-leaf point pairs are connected via the local greedy-spanners and the inter-leaf point pairs are connected via local greedy spanners and the WSPD-spanner edges. 
		
		\item \textit{Step 4.}  We create a dual graph $G_T$ on the  leaves of $T$ where the edge $\{\sigma_i,\sigma_j\} \in G_T$ if and only if the leaves $\sigma_i,\sigma_j$ are neighbors in $T$ (their bounding-boxes intersect). See Fig.~\ref{fig:qt} (right). To find the neighboring leaves of a leaf $\sigma$,  we dilate the box corresponding to  $\sigma$ by a small quantity, and find the intersecting leaves by spawning a search at the root of $T$. 
		It is easy to check that $G_T$ is connected.
Using \textsc{Greedy-Merge} (Algorithm~\ref{alg:greedymerge}), we ``stitch'' (merge) the local spanners inside every pair of non-empty neighbors $\sigma_i,\sigma_j$ in $G_T$ and ensure $t$-paths between all point pairs in $P_i \times P_j$, where $P_i:=P\cap \sigma_i,P_j:=P\cap \sigma_j$. See Fig.~\ref{fig:merge} for an illustration. However, in our algorithm, to minimize the number of edges placed during the mergings, we consider the whole spanner constructed so far when we look for $t$-paths. From now on, by \textit{merging} two non-empty leaves in $G_T$, we mean merging the two local spanners inside them, as explained above. 
	
	\begin{figure}[ht]
		\centering
		\includegraphics[scale=0.6]{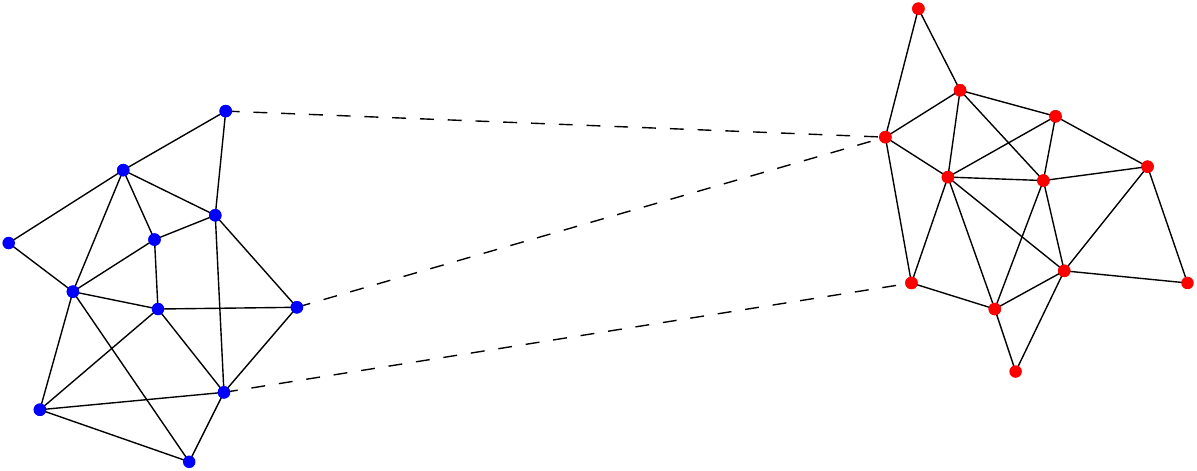}
		\caption{Two separate greedy $1.25$-spanners are constructed on the red and blue points. Then, the two spanners are merged by placing the dashed edges and ensuring $1.25$-paths for all point pairs $u,v$ where $u$ is red and $v$ is blue. The closer the two spanners are, the more edges may be needed to merge them.}
		\label{fig:merge}
	\end{figure}

		After completion of this step, the number of point pairs not having $t$-paths reduces drastically. Now we give a description of \textsc{Greedy-Merge}.
			Let $L := P_i \times P_j$.
		Check for the existence of the edge $e:=\{\ell_{\sigma_i}, \ell_{\sigma_j}\}$ in $H$. If $e$ exists, then it was placed in Step 3.
		If $e \in H$, prune pairs from $L$ in the following way.  For every pair $(u,v) \in L$,  find the lengths of the following three paths: a shortest path from $u$ to $\ell_{\sigma_i}$ in $H$, the length of the edge  $e$, and a  shortest path from $\ell_{\sigma_j}$ to $v$ in $H$. Let the lengths of these three paths be $a,b,c$, respectively. If $(a+b+c) / |uv| \leq t$ (a $t$-path between $u,v$ already exists in $H$),  remove $(u,v)$ from $L$.  Note that the two paths $u$ to $\ell_{\sigma_i}$ and $\ell_{\sigma_j}$ to $v$ always exist because of the local greedy $t$-spanners already created inside every non-empty leaf.  In our experiments, we found that owing to careful edge placements  by \textsc{FG-Greedy}, the lengths of the shortest paths between any two points $u,v$ from the same leaf is very close to and at most $t \cdot |uv|$. So, instead of storing the path lengths obtained from the \textsc{FG-Greedy} executions or recomputing them, we use the quantities $t\cdot |u\ell_{\sigma_i}|$ and $t\cdot |\ell_{\sigma_j}v|$ as a substitute of their shortest path lengths. This adjustment saves us time and memory without any visible increase in the average-degree of $H$. Next, we sort $L$ (after possible pruning) based on the
		Euclidean distance of the point pairs. Akin to
		\textsc{FG-Greedy}, the sorting step helps to reduce the new
		edges placed to merge two local greedy-spanners inside
		$\sigma_i,\sigma_j$.

			\begin{algorithm}[H]

			Let $L = P_i \times P_j$, where $P_i, P_j$ are the pointsets inside the leaves $\sigma_i,\sigma_j$, respectively\;

			\If{\upshape the edge $\{\ell_{\sigma_i},\ell_{\sigma_j}\} \in H$ }{
				
				\For{\upshape\textbf{each} $(u,v) \in L$}{
				
					\If{\upshape $\dfrac{t\cdot|u\ell_{\sigma_i}| + |\ell_{\sigma_i}\ell_{\sigma_j}| + t\cdot |\ell_{\sigma_j}v|}{ |uv|}\leq t$}
					{
						$L = L \setminus \{(u,v)\}$\;
					}	
				}
			}
			
			Sort $L$ according to the Euclidean distance of the point pairs\;
			$\texttt{bridges} = \emptyset$; 
			
			\For{\upshape\textbf{each} $\{u,v\} \in L$}{
				\If{\upshape $\nexists$ $x\curly y \in \texttt{bridges}$ such that $\dfrac{t\cdot |ux| + |x\curly y| + t \cdot |yv|  } {|uv| }\leq t$ } {
					\If{\upshape the path $u \curly v$ returned by \textsc{Greedy-Path}$(H,u,v)$ is a $t$-path between $u,v$}{
						$\texttt{bridges}  = \texttt{bridges} \cup \{u \curly v\}$\;
					}
					\ElseIf{\upshape the shortest path $\pi(u,v)$ between $u,v$ returned by $A^*$ is a $t$-path}{
							$\texttt{bridges}  = \texttt{bridges} \cup \{\pi(u,v)\}$\;
						}
						\Else{
							Place the edge $\{u,v\}$ in $H$\;
							$\texttt{bridges}  = \texttt{bridges} \cup \{\{{u,v}\}\}$\; 
						}			
					}	
				}
				\caption{\textsc{Greedy-Merge}$(\sigma_i,\sigma_j,t,H)$}
				\label{alg:greedymerge}
			\end{algorithm}

					A \emph{bridge} is a path that connects two
		vertices $u \in P_i,v\in P_j$. We maintain a set of \texttt{bridges} found so far. 
		For every pair $(u,v) \in L$, first we check if there is a bridge $x\curly y \in \texttt{bridges}$  such that the path from $u$ to $v$ via the bridge $x\curly y$ is a $t$-path for $u,v$. If there is none, we check if there is a $t$-path between them in $H$. The main objective behind caching all the bridges found so far is to reuse them for the future point pairs since we found that \textsc{Greedy-Merge} does not put many bridges while trying to ensure $t$-paths among all point pairs. Consequently, we are able reduce the number of expensive  $t$-path computations.

			\begin{algorithm}[H]
			$y=u$\;
			\texttt{path} $=[u]$\;
			Mark $u$ \texttt{visited} and  all vertices in $P\setminus \{u\}$ \texttt{unvisited}\;
			
			\While{\upshape  $y\neq v$ }{
				Let $X$ be the set of neighbors of $y$ in $H$ currently marked  \texttt{unvisited}\;
				\If{\upshape$X$ is empty}{Report path cannot be found\;}
				Find the neighbor $x \in X$ that minimizes $|yx| + |xv|$\;
				Append $x$ to \texttt{path}\;
				\If{$x == v$}{\textbf{break}\;}
				$y=x$\;
			}
			\Return \texttt{path}\;
			\caption{\textsc{Greedy-Path}$(H,u,v)$}
			\label{alg:greedypath}
		\end{algorithm}
		
		For finding a $t$-path between a vertex pair $u,v$, we first find a path using \textsc{Greedy-Path} (Algorithm~\ref{alg:greedypath}) and check if the path returned by it is a $t$-path between $u,v$ in $H$. It iteratively constructs a path (not necessarily a shortest one) starting at $u$ by including the next neighbor $x$ of the current vertex $y$ that minimizes the Euclidean distance between $y$ and $x$ plus the Euclidean distance between $x$ and $v$.   If the path returned is not a $t$-path for $u,v$, we run $A^*$,  a popular Dijkstra-based shortest path algorithm that runs fast on geometric networks~\cite{hart1968formal}, to find a shortest path $\pi(u,v)$ between $u,v$ in $H$. Then, we verify if $\pi(u,v)$  is a $t$-path between $u,v$.
		We note that \textsc{Greedy-Path} does not use a priority-queue like $A^*$, and as a result, it tends to be faster and uses less memory. For instance, for a $512K$-element pointset uniformly distributed inside a square, using \textsc{Greedy-Path} made our algorithm $\approx 23.5\%$ faster. As explained in Section~\ref{sec:exp}, \textsc{Greedy-Path} could find $t$-paths in most cases. As a result, the number of $A^*$ calls was much less than the number of \textsc{Greedy-Path} calls.

	If a $t$-path is found, we put the $t$-path between $u,v$ in \texttt{bridges}. Otherwise, we place the edge $\{u,v\}$ both in $H$ and \texttt{bridges}. 
	
	In our experiments, we found that after this step, the number of point pairs that do not have $t$-paths between them in $H$ is very low and sometimes is even zero.

			\item \textit{Step 5.} 
		For every non-empty leaf $\sigma \in T$, we merge $\sigma$  with the  leaves in $T$ that are at least $2$ and at most $h$ hops away from it in $G_T$ (constructed in Step 4). 		
		Note that in Step 4, we already merged $\sigma$ with the leaves that are exactly one hop away from it in $G_T$. However, we do this incrementally. First, we merge every leaf with its $2$-hop neighbors, then with its $3$-hop neighbors, and so on. In this step, we use a lighter version of the merge algorithm used in the previous step. The sorting step in \textsc{Greedy-Merge} takes a substantial amount of time since it sorts $|P_i| \cdot |P_j|$ point pairs and for merging the pointsets inside distant leaves ($2$ hops or more away from each other in $G_T$), we replace the sorting step with the following faster alternative.  Sort the points in $P_i$($P_j$) based on their distances from the leader point of the points inside the other leaf $\sigma_j$($\sigma_i$). In this case, we execute sorting twice; once on $|P_i|$ pairs and another time on $|P_j|$ pairs. Then we use a pair of nested loops to work with the point pairs, as shown in Algorithm~\ref{alg:greedymergelight}.
		We call this slightly modified merging algorithm
		\textsc{Greedy-Merge-Light}.
		
							\begin{algorithm}[H]
			Let $P_i, P_j$ be the two pointsets inside the leaves $\sigma_i,\sigma_j$, respectively\;
			
			Sort $P_i$ based on their distances from the leader point   $\ell_{\sigma_j}$\;
			
			Sort $P_j$ based on their distances from the leader point $\ell_{\sigma_i}$\;
			
			$\texttt{bridges} = \emptyset$\; 	
			\If{\upshape the edge $\{\ell_{\sigma_i},\ell_{\sigma_j}\} \in H$ }{
				$\texttt{bridges} = \texttt{bridges} \cup \{\{\ell_{\sigma_i},\ell_{\sigma_j}\}\}$\;
			}

			\For{\upshape\textbf{each} $u \in P_i$}{
				\For{\upshape\textbf{each} $v \in P_j$}{
					\If{\upshape $\nexists$ $x\curly y \in \texttt{bridges}$ such that $\dfrac{t\cdot |ux| + |x\curly y| + t \cdot |yv|  } {|uv| }\leq t$ } {
						\If{\upshape the path $u \curly v$ returned by \textsc{Greedy-Path}$(H,u,v)$ is a $t$-path between $u,v$}{
							$\texttt{bridges}  = \texttt{bridges} \cup \{u \curly v\}$\;
						}
						\ElseIf{\upshape the shortest path $\pi(u,v)$ between $u,v$ returned by $A^*$ is a $t$-path}{
								$\texttt{bridges}  = \texttt{bridges} \cup \{\pi(u,v)\}$\;
							}
							\Else{
								Place the edge $\{u,v\}$ in $H$\;
								$\texttt{bridges}  = \texttt{bridges} \cup \{\{{u,v}\}\}$\; 
							}			
						}	
								%
								%
					}	
				}
				
				\caption{\textsc{Greedy-Merge-Light}$(\sigma_i,\sigma_j,t,H)$}
				\label{alg:greedymergelight}
			\end{algorithm}
		\end{itemize}

			Since the merging of leaf pairs that are relatively
		far from each other 
		rarely need placement of new edges after Step 4, this modified step has almost the same effect as that of the sorting used in the previous step. Due to the higher number of mergings  required in this step (compared to Step 4), we use \textsc{Greedy-Merge-Light} to speed up our algorithm in practice. For instance, for a $1M$-element  pointset uniformly distributed inside a square, using \textsc{Greedy-Merge-Light} in Step 5 instead of \textsc{Greedy-Merge} helped the algorithm to speed up by a factor of $2$ but with a negligible increment of $0.1$ in the average-degree of the final spanner.
		
		In the following, we show that $H$ is sparse. We first present Lemma~\ref{lem:edges}, which is  subsequently used to prove Theorem~\ref{thm:sparse}.

		\begin{lemma}\label{lem:edges}
			Let $S$ be a  set of leaves of the quad-tree $T$ (constructed in Step 1). Then, the subgraph $G_T(S)$ of $G_T$, induced by $S$ has $O(|S|)$ edges.
		\end{lemma}
		
		\begin{proof}
			Two leaves in $S$ form a \emph{diagonal pair} if they intersect only at a point otherwise they form a \emph{non-diagonal pair}. An edge in $G_T(S)$ connecting the leaves of a diagonal pair is a \emph{diagonal edge}; otherwise, it is a \emph{non-diagonal edge}. Let the number of diagonal edges in $G_T(S)$ be $a$ and the non-diagonal edges be $b$. 
			
			Since every leaf in $S$ can form a diagonal pair with at most four other leaves in $S$, $a = O(|S|)$. Now we estimate the number of non-diagonal edges in $G_T(S)$.
			Construct a geometric graph $G$  using the leaves in $S$. For every leaf $\sigma \in S$, we consider its center to be a vertex in $G$. Place an edge (a  line-segment) between two vertices in $G$ if and only if the corresponding two leaves form a neighboring non-diagonal pair.  Every leaf $\sigma \in S$ forms a star centered at its center, with its adjacent non-diagonal leaves in $S$. We observe that these stars are plane (crossing-free). Refer to Fig.~\ref{fig:mergings}. 
			$G$ is a union of the stars centered at the  leaves in $S$.  So, $G$ is plane graph on $O(|S|)$ vertices. Consequently, $G$ has $O(|S|)$ edges. This implies, $b=O(|S|)$
			
			Thus, we conclude that $G_T(S)$ has $a+b=O(|S|)+ O(|S|) = O(|S|)$ edges.
		\end{proof}

\begin{figure}
	\centering
	\includegraphics[scale=0.6]{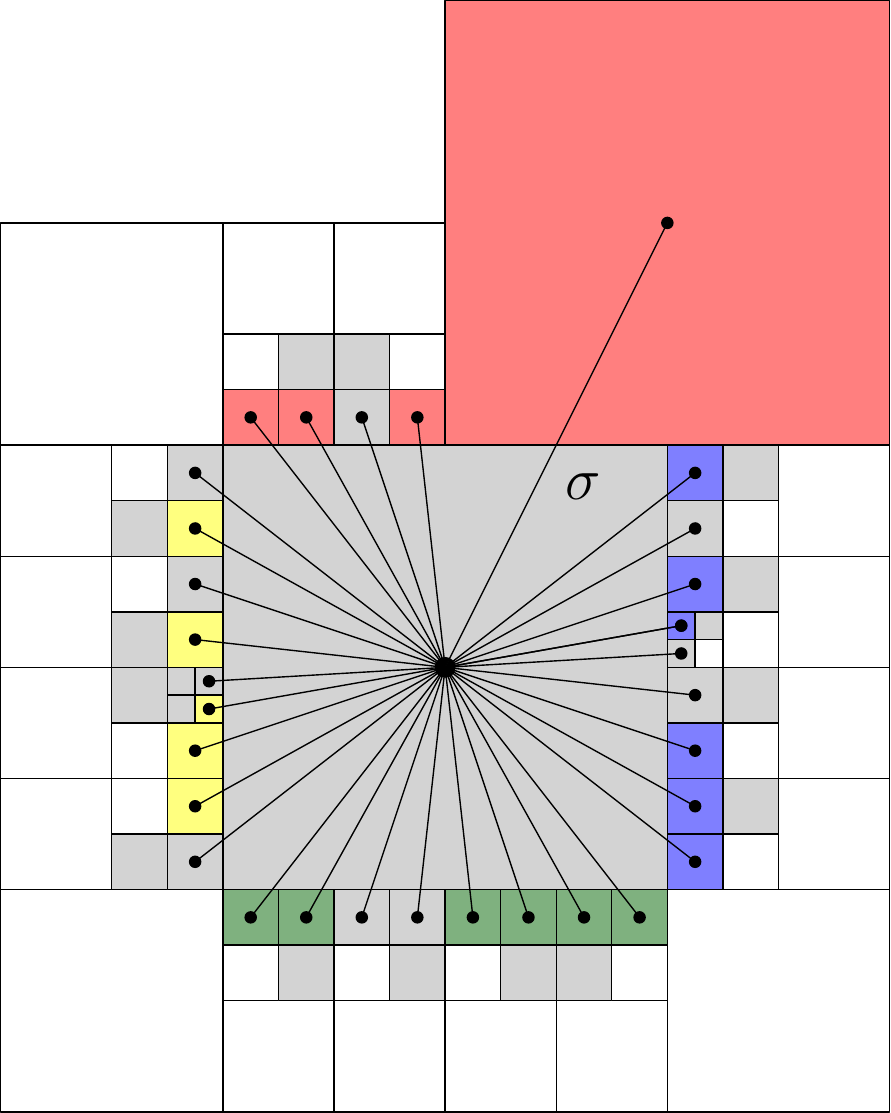}
	\caption{The non-empty leaves are shown in gray. The red, blue, green, and yellow leaves are the empty north, east, south, and west  neighbors of the non-empty leaf $\sigma$, respectively. }
	\label{fig:mergings}
	\end{figure}

	\begin{theorem} \label{thm:sparse}
		For fixed values of $h$ and $k$, $H$ is a  geometric graph on $P$ with $O(n)$ edges.
	\end{theorem}

\begin{proof}
Since $k=O(1)$, the  size of every greedy-spanner constructed by \textsc{FG-Greedy} inside every leaf is $O(1)$. The quad-tree $T$ has $O(n)$ non-empty leaves. So, the total number of edges put by \textsc{FG-Greedy} in Step 2 is $O(n)$. 
The WSPD-spanner $W$ constructed in step 3 on $O(n)$ leaders contains $O(n)$ edges~\cite{narasimhan2007geometric}. Thus, the total number of edges put so far in $H$ is $O(n)$.

Next, we show that the algorithm places $O(n)$ edges in $H$ in steps 4 and 5. To show this, we estimate the number of mergings $m$, executed in the Steps 4 and 5. Since the size of every non-empty leaf of the quad-tree $T$ is no more than $k$, every merging places at most $O(k^2)=O(1)$ edges in $H$. It remains to show that $m=O(n)$. 

As every merging requires at least one non-empty leaf (the other non-empty leaf is at most $h$ hops away in $G_T$),    $m$  equals $h$ times the total number of edges $q$ incident on the non-empty leaves of $T$ in $G_T$. Now we shall estimate $q$ and show that $q = O(n)$.

Let $A$ be the set of non-empty leaves in $T$.
Consider any leaf $\sigma \in A$. For every non-empty leaf $\sigma'$ (possibly same as $\sigma$), there can be at most two of its empty sibling leaves which are the east neighbors of $\sigma$. Refer to Fig.~\ref{fig:mergings}. 
Since there are $O(n)$ non-empty leaves, the total number of empty east neighbors for all the non-empty leaves in $T$ is also $O(n)$. Similarly, the total number of north, south, and west empty neighbors is also $O(n)$ each. 
Thus, there are a total of $O(n)$ empty neighbors of the non-empty leaves. We denote the set of such empty neighbors by $B$. Further, let $C$ be the set of diagonal neighbors of the non-empty leaves. Clearly, $|C|=O(n)$.
Now, let $G_T(A \cup B \cup C)$ be the subgraph of $G_T$, induced by the vertices in $A\cup B \cup C$. By Lemma~\ref{lem:edges}, $G_T(A\cup B \cup C)$ has $O(|A \cup B \cup C|) = O(|A| + |B| + |C|)= O(n + n + n) = O(n)$ edges.

Since $h$ is a constant, $m = h \cdot q = O(1) \cdot O(n) = O(n)$. Hence, $H$ has $O(n)$ edges.
\end{proof}
	
	\begin{corollary}\label{cor:mergings}
		For fixed values of $h$ and $k$, the total number of mergings executed by \textsc{Fast-Sparse-Spanner} is $O(n)$.
	\end{corollary}
	
	Now, we focus on our algorithm's time and space complexities.

		\begin{theorem} 
			\label{thm:time}
			For fixed values of $h$ and $k$, 
		\textsc{Fast-Sparse-Spanner} runs in $O(n^2\log n + d^2n)$ time and uses $O((d+1)n)$ extra space, where $d$ is the depth of the quad-tree $T$ used in Step 1.

	\end{theorem}
	
\begin{proof}
The construction of the quad-tree $T$ in Step 1 takes $O((d+1)n)$ time~\cite[Chapter 14]{de2008computational}. Since there are $O(n)$ non-empty leaves of $T$, our algorithm makes $O(n)$ \textsc{FG-Greedy} calls, each on a pointset of size at most $k=O(1)$. Thus, Step 2 takes $O(n)$ time. The construction of the WSPD-spanner $W$ in Step 3 takes $O(n\log n)$ time.  

Now we estimate the time taken to construct $G_T$. Let $n_v$ denote the number of neighbors of a vertex $v$ in $G_T$. The time taken to find $n_v$ neighbors of $v$ by doing a 4-way search starting at the root of the quad-tree $T$ is $O(d)\cdot n_v$. Since $G_T$ has $O((d+1)n)$ leaves~\cite[Chapter 14]{de2008computational}, by Lemma~\ref{lem:edges}$, G_T$ has $O((d+1)n)$ edges. So, the total time taken to compute $G_T$ is, $$\sum_{v \in G_T} ( O(d) \cdot n_v ) = O(d) \sum_{v \in G_T} n_v = O(d)\cdot O((d+1)n) = O(d^2 n).$$

We use $O(n)$ breadth-first traversals to find the leaves that are at most $h$ hops away. Since $G_T$ has $O((d+1)n)$ edges and $h=O(1)$, total time taken for the $n$ traversals amounts to $O(h(d+1)n) = O((d+1)n)$.

Since $H$ is sparse (Theorem~\ref{thm:sparse}), each execution of \textsc{Greedy-Path} and $A^*$ takes $O(n)$ and $O(n\log n)$ time each, respectively. 
The sorting calls in \textsc{Greedy-Merge} and \textsc{Greedy-Merge-Light} run in $O(1)$ time each since $k=O(1)$.  
By Corollary~\ref{cor:mergings}, our algorithm executes $O(n)$ mergings. Every merging considers at most $k^2 = O(1)$ point pairs (some leaves in $T$ may contain less than $k$ points). 
For every such point pair, $O(k^2 + O(n) + O(n\log n)) = O(k^2 + n\log n)$ time is spent to verify the existence of a $t$-path in $H$, since there are $O(k^2)$ bridges in the worst case. So, the total time spent in the Steps 4 and 5 is $O(n\cdot k^2 \cdot  (k^2 + n\log n)) = O(n^2\log n)$. 

Summing up all the above runtimes gives us a runtime of $O(n^2\log n + d^2n)$ for the \textsc{Fast-Sparse-Spanner}.

Now we analyze its space complexity. 
The quad-tree $T$  needs $O((d+1)n)$ storage space~\cite[Chapter 14]{de2008computational}. Step 2 uses $O(1)$ extra space since \textsc{FG-Greedy} is run on $k=O(1)$ points. The  construction of the WSPD-spanner $W$ on the leader points in Step 3 uses $O(n)$ space. $G_T$ is a dual graph on $O((d+1)n)$ leaves of $T$. Since $G_T$ has $O((d+1)n)$ edges, for storing $G_T$, $O((d+1)n)$ extra space is needed. For running $A^*$, we need $O(n)$ extra space for maintaining a priority-queue. \textsc{Greedy-Path} uses $O(n)$ extra space for bookkeeping. Thus, the total space requirement amounts to $O((d+1)n)$.
\end{proof}
	
 \textsc{Fast-Sparse-Spanner} is a $t$-spanner algorithm (can always produce $t$-spanners for any value of $t>1$) if $h$ is set to the diameter of $G_T$. Since $G_T$ is connected, in this case, all leaf pairs will be considered for merging in Steps 4 and 5, and consequently, $t$-paths will be ensured for every point pair in $P$. Hence, we state the following theorem without a proof.
 
	\begin{theorem}\label{thm:t-spanner}
	For any integer value of $k$, if $h$ is set to the diameter of $G_T$, then \textsc{Fast-Sparse-Spanner} will always produce $t$-spanners for any value of $t>1$.
\end{theorem}
 
 However, setting $h$ to the diameter of $G_T$ will make the algorithm slow for large pointsets. In Section~\ref{sec:exp}, we present precise $t$-dependent values for $h$ used in our experiments.  

While worst-case time complexities are high, our experiments suggest
that the average-case time complexities are exceptionally low. See the
discussion in Section \ref{sec:argue}. 
 
	\section {Computing the stretch-factor of $H$} \label{sec:sf}
	
	In this section, we present a simple algorithm, named \textsc{Fast-Stretch-Factor}, that can compute stretch-factor of any spanner $H$ generated by \textsc{Fast-Sparse-Spanner}. 
	Refer to Algorithm~\ref{alg:sf}. It is a minor modification of the
	\textsc{Greedy-Merge-Light} algorithm. It needs the quad-tree
	$T$ used for constructing $H$. 
	
	Let $\Xi$ be the set of non-empty	leaves of $T$ and $X$ be all pairs of distinct leaves in $\Xi$.  When
	$H$ was constructed by our algorithm  \textsc{Fast-Sparse-Spanner}, a
	subset $M$ of  $X$ were considered for merging in Steps 4 and 5. At
	that time, it was ensured 
	that for every leaf pair $\{\sigma_i,\sigma_j\}$ in $M$ and for every point pair $\{u \in \sigma_i,v\in \sigma_j\}$, there exists  a  $t$-path in $H$. Thus, it is enough to consider  the remaining leaf pairs in $X \setminus M$ and look for the set $\Gamma$ of point pairs (points from two different leaves) not having $t$-paths in $H$. Clearly, inside every leaf, it is impossible to have a pair that does not have a $t$-path since we have used  \textsc{FG-Greedy} to create local $t$-spanners inside every leaf.
	However, in doing so, we do not modify $H$ (since we are computing its stretch-factor). The algorithm returns $t_H := \max(t,\max_{\{u,v\}\in \Gamma} \pi(u,v) / |uv|)$. 
	
		\begin{algorithm}[H]
		
		$t_H=t$\;
		
		\For{\upshape \textbf{each} $\{\sigma_i,\sigma_j\} \in X \setminus M$}{
			
			Let $P_i, P_j$ be the two pointsets inside the leaves $\sigma_i,\sigma_j$, respectively\;
			
			Sort $P_i$ based on their distances from the leader point   $\ell_{\sigma_j}$\;
			
			Sort $P_j$ based on their distances from the leader point $\ell_{\sigma_i}$\;
			
			$\texttt{bridges} = \emptyset$\; 	
			\If{\upshape the edge $\{\ell_{\sigma_i},\ell_{\sigma_j}\} \in H$ }{
				$\texttt{bridges} = \texttt{bridges} \cup \{\{\ell_{\sigma_i},\ell_{\sigma_j}\}\}$\;
			}

			\For{\upshape\textbf{each} $u \in P_i$}{
				\For{\upshape\textbf{each} $v \in P_j$}{
					\If{\upshape $\nexists$ $x\curly y \in \texttt{bridges}$ such that $\dfrac{t\cdot |ux| + |x\curly y| + t \cdot |yv|  } {|uv| }\leq t$ } {
						\If{\upshape the path $u \curly v$ returned by \textsc{Greedy-Path}$(H,u,v)$ is a $t$-path between $u,v$}{
							$\texttt{bridges}  = \texttt{bridges} \cup \{u \curly v\}$\;
						}
						\ElseIf{\upshape the shortest path $\pi(u,v)$ between $u,v$ returned by $A^*$ is a $t$-path}{
								$\texttt{bridges}  = \texttt{bridges} \cup \{\pi(u,v)\}$\;
							}
							\Else{
								$t_H = \max\left(t_H, \dfrac{|\pi(u,v)| }{ |uv|}\right)$;\
							}			
						}	
					}	
				}
			}
			
			\textbf{return} $t_H$\;
			\caption{\textsc{Fast-Stretch-Factor}$(P,t,H,T,\Xi,M)$}
			\label{alg:sf}
		\end{algorithm}

	Next, we derive the asymptotic time and space complexities for \textsc{Fast-Stretch-Factor} in Theorem~\ref{thm:sf}.
	\begin{theorem} \label{thm:sf}
		For fixed values of $h$ and $k$, 
	\textsc{Fast-Stretch-Factor} runs in $O(n^3\log n)$ expected time and uses $O(n)$ extra space. 
\end{theorem}

\begin{proof}
There are $\Theta(n^2)$ non-empty leaf pairs in $X \setminus M$. If $M$ is maintained using a hash-table, it takes $O(1)$ expected time to whether a pair $\{\sigma_i,\sigma_j\} \in M$. Every merging maintains no more than $k^2$ bridges. Since \textsc{Greedy-Path} runs in $O(n)$ time and $A^*$ in $O(n\log n)$ time, for every pair in $X \setminus M$, we spend $O(k^4 \cdot n\log n ) =  O(n\log n)$ time to verify the existence of $t$-paths among $k^2$ pairs. The two sorting calls run in $O(k \log k) = O(1)$ time each. Hence, the total expected time taken by \textsc{Fast-Stretch-Factor} is $O(n^2 \cdot n\log n) = O(n^3\log n)$. 

For storing the bridges found so far, $O(1)$ space is needed. Both \textsc{Greedy-Path} and $A^*$ calls need $O(n)$ space each.  So, the total space complexity amounts to $O(n)$.
\end{proof}

	\textbf{Remarks.} In this algorithm, we initialize $t_H$ with $t$. So, if the actual stretch-factor is less than $t$, it will still return $t$ as the stretch-factor. 
	
	Note that the pairs in $X \setminus M$ are faraway leaf pairs since the ones which are at most $h$ hops away from each other have already been merged while constructing $H$. So, in most cases, \texttt{bridges} contain a very few bridges which are enough to verify the existence of $t$-paths between all points pairs in $P_i \times  P_j$. Further, the $t$-paths between  faraway point pairs have a low number of edges because of the existence of the long WSPD edges in $H$. Consequently, the two path-finding algorithms, \textsc{Greedy-Path} and $A^*$, terminate fast. This helps our algorithm to be faster than the naive Dijkstra-based algorithm in practice that runs in $O(n^2\log n)$ time on sparse graphs. Refer to Section~\ref{sec:exp} for experimental evidence.
	
	\section{Experiments}\label{sec:exp}
	
	We have implemented our algorithms in \textsf{C}\texttt{++}$17$ using the \textsf{CGAL} library. Two machines with same configuration were deployed for conducting the experiments. Both of them are equipped with  \textsf{Intel i9-12900K} processors and $32$ GB of main memory, and run \textsf{Ubuntu Linux 22.04 LTS}.  Our code was compiled using the \texttt{g++} compiler, invoked with \texttt{-O3} flag to achieve fast real-world speed. 
	
	Algorithms \textsc{Fast-Sparse-Spanner} and \textsc{Fast-Stretch-Factor} have
	been engineered to be as fast as possible. 
	From \textsf{CGAL}, the \texttt{Exact\_predicates\_inexact\_constructions\_kernel} is used for accuracy and speed.  For a quad-tree implementation, we have used \texttt{CGAL::Quadtree}. To engineer \textsc{FG-Greedy}, we have used  \texttt{boost::dijkstra\_shortest\_paths\_no\_color\_map} from the  \texttt{boost} library for a robust implementation of the Dijkstra's shortest path algorithm.
	For maintaining sets (where no ordering is required), \texttt{std::unordered\_set} have been used. 
	
	\medskip
	
	\noindent
	\textit{Synthetic pointsets.} 
	Following the strategies used in previous experimental work 
	\cite{bentley-gtsp-92, nz-geomst-01, friederich2023experiments, anderson2022bounded}, we 
	used the following distributions for drawing random
	pointsets. Refer to Fig.~\ref{fig:dist}, for a visualization of the point distributions. 
	
		\begin{figure}[h]
		\centering
		\begin{subfigure}[b]{0.22\linewidth}
			\centering
			\includegraphics[scale=0.15]{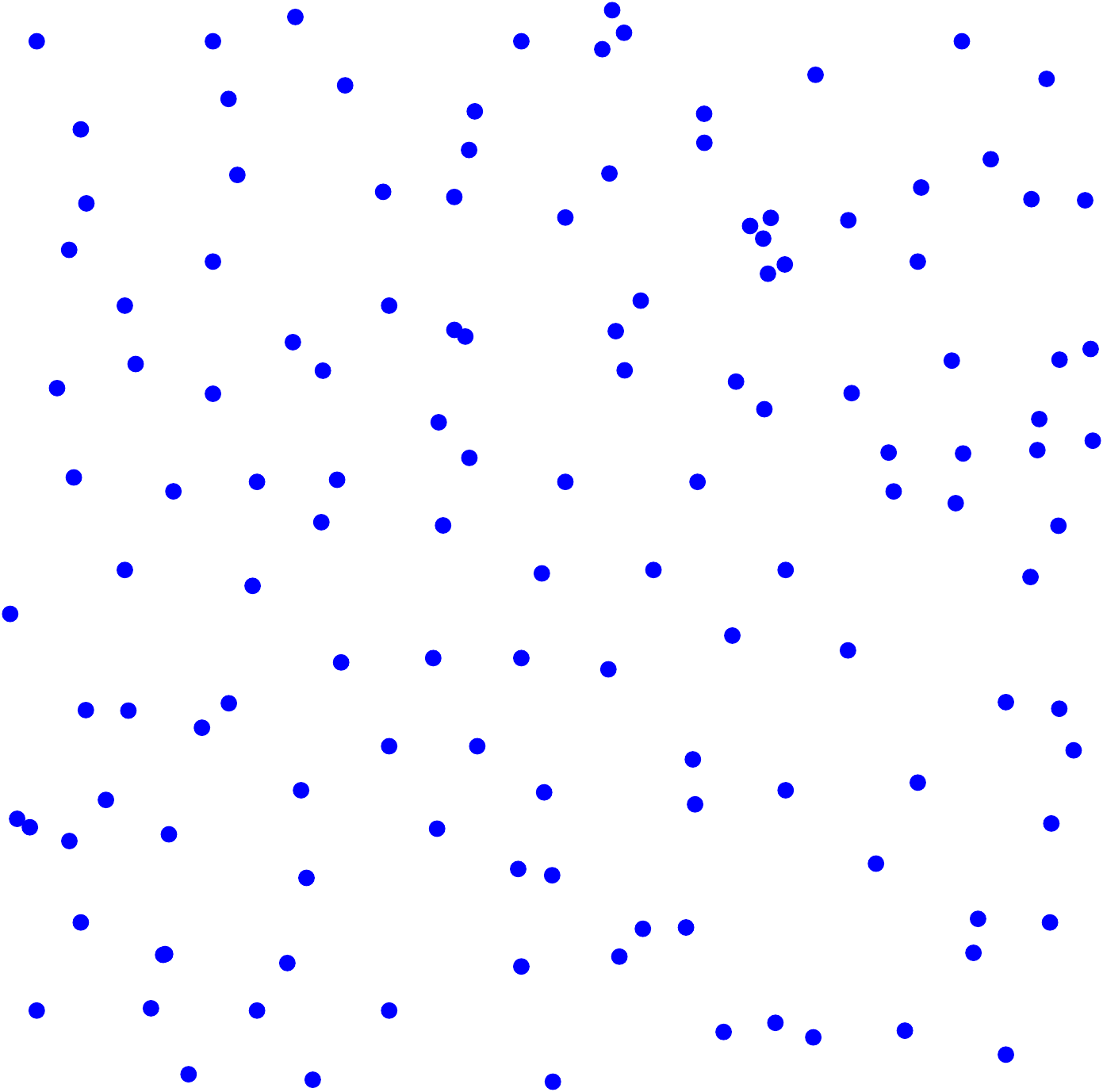}
			\subcaption{\texttt{uni-square}}
		\end{subfigure}
		\begin{subfigure}[b]{0.25\linewidth}
			\centering
			\resizebox{0.8\linewidth}{!}{\includegraphics{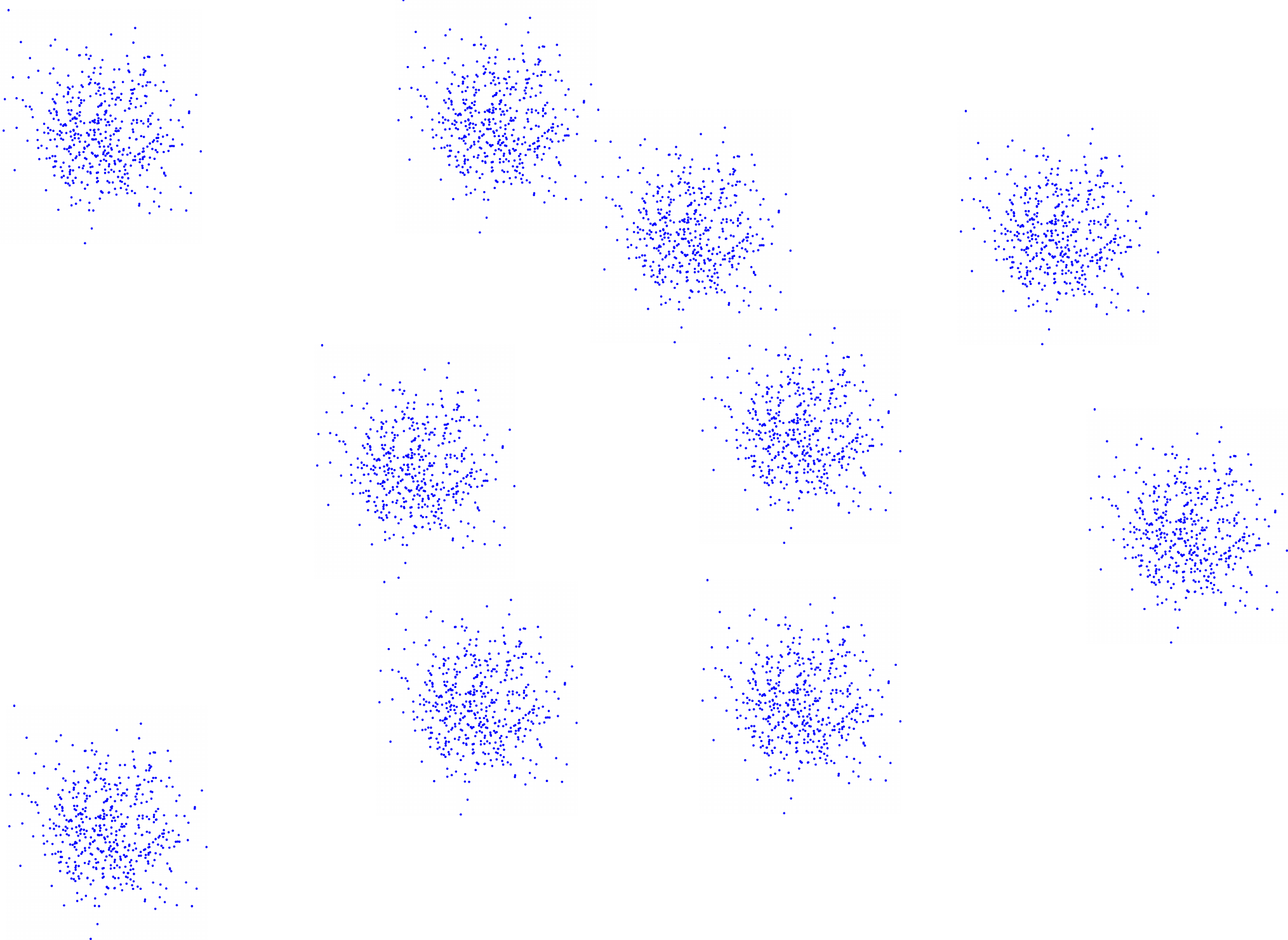}}
			\subcaption{\texttt{normal-clustered}}
		\end{subfigure}
		\begin{subfigure}[b]{0.22\linewidth}
			\centering
			\resizebox{0.6\linewidth}{!}{\includegraphics{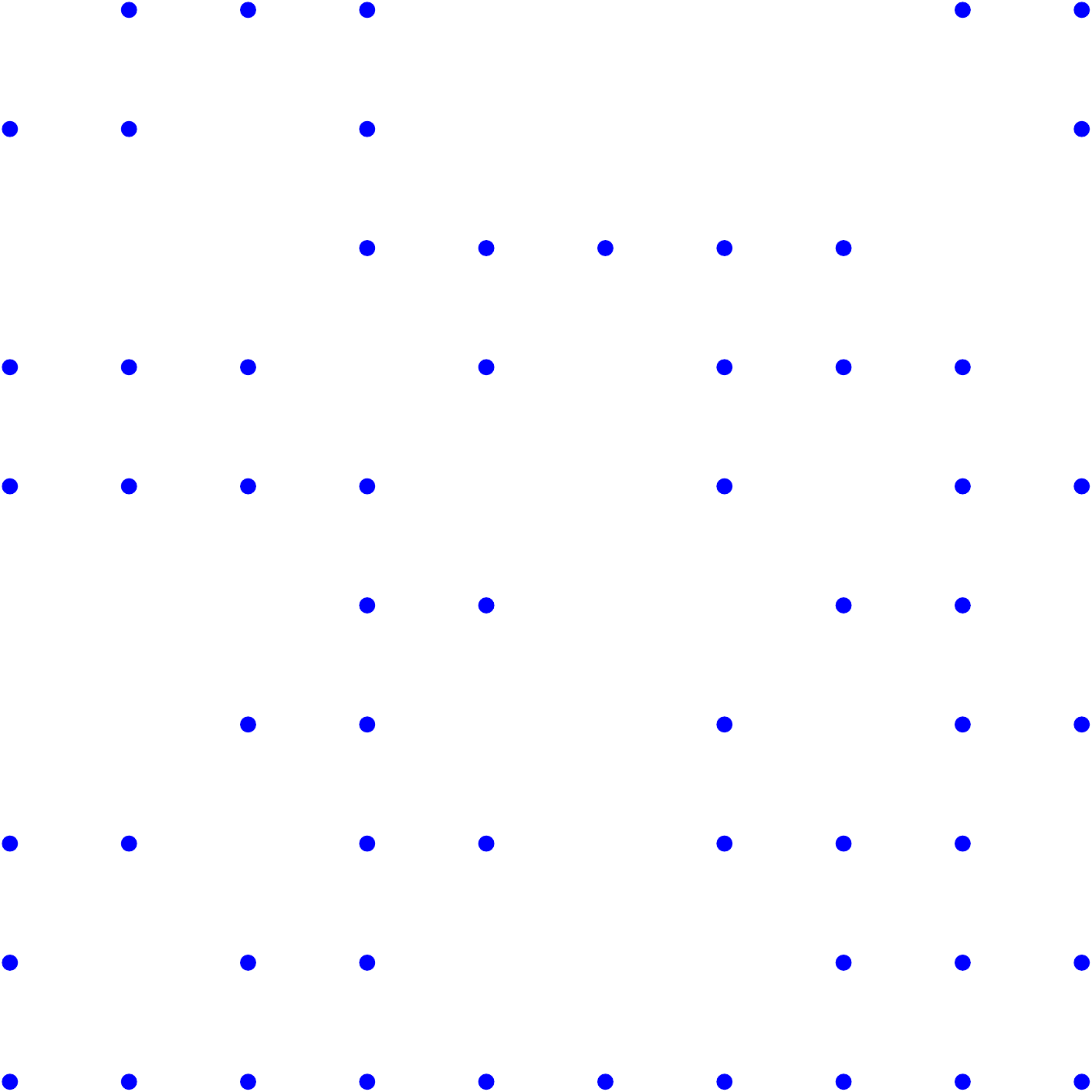}}
			\subcaption{\texttt{grid-random}}
		\end{subfigure}
	\vspace{2pt}
		\begin{subfigure}[b]{0.22\linewidth}
			\centering
			\resizebox{0.6\linewidth}{!}{\includegraphics{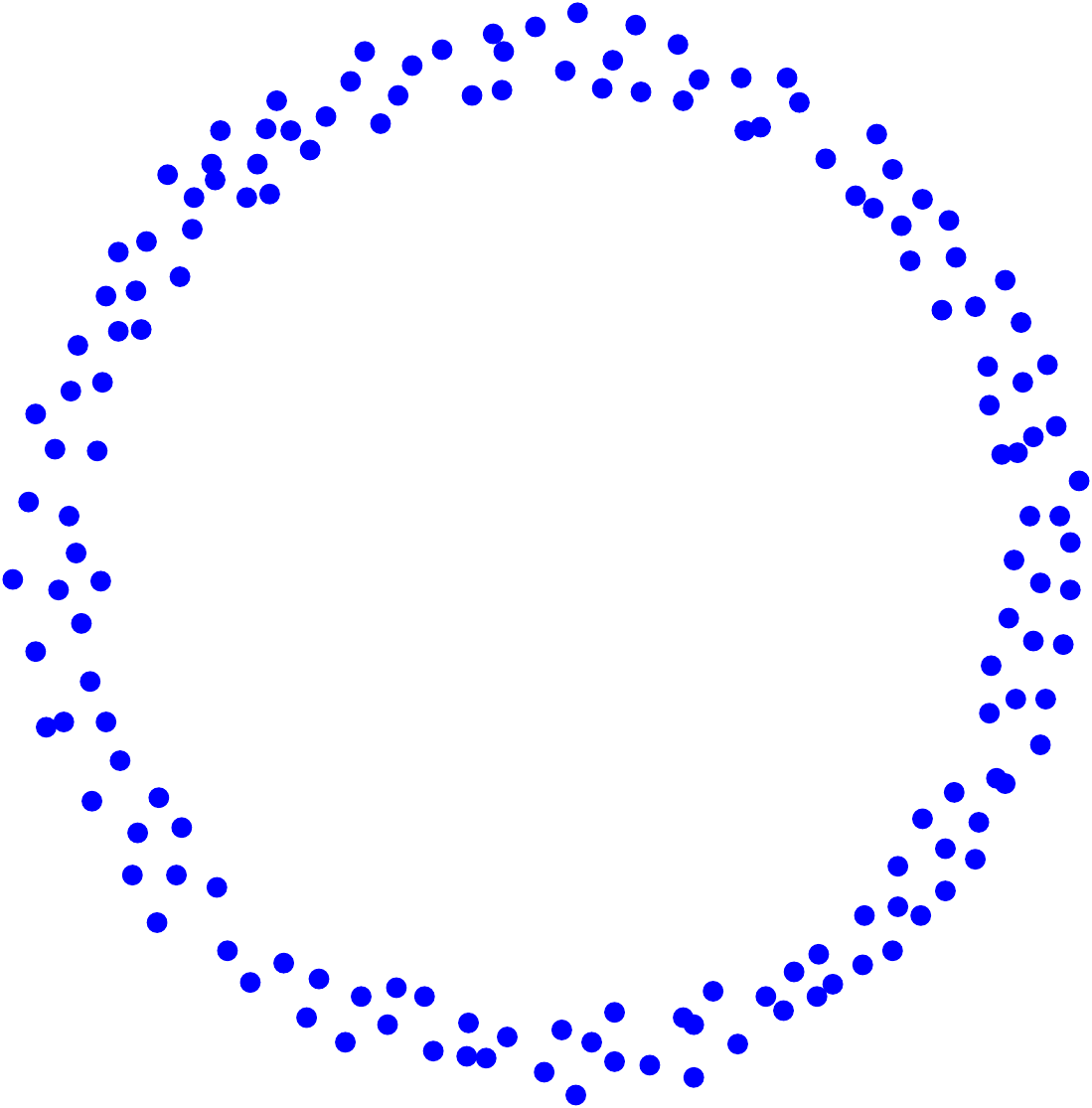}}
			\subcaption{\texttt{annulus}}
		\end{subfigure}
		\begin{subfigure}[b]{0.22\linewidth}
			\centering
			\resizebox{0.8\linewidth}{!}{\includegraphics{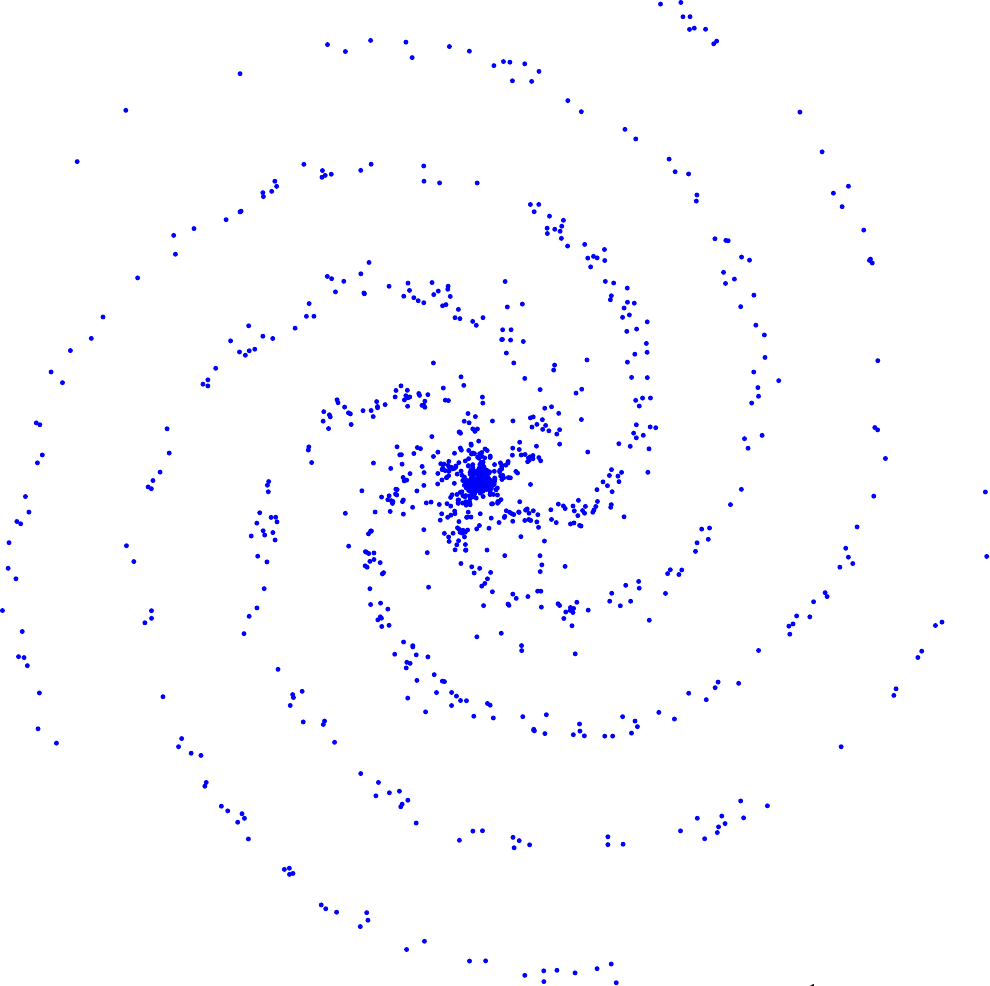}}
			\subcaption{\texttt{galaxy}}
		\end{subfigure}
		\begin{subfigure}[b]{0.22\linewidth}
			\centering
			\resizebox{0.6\linewidth}{!}{\includegraphics{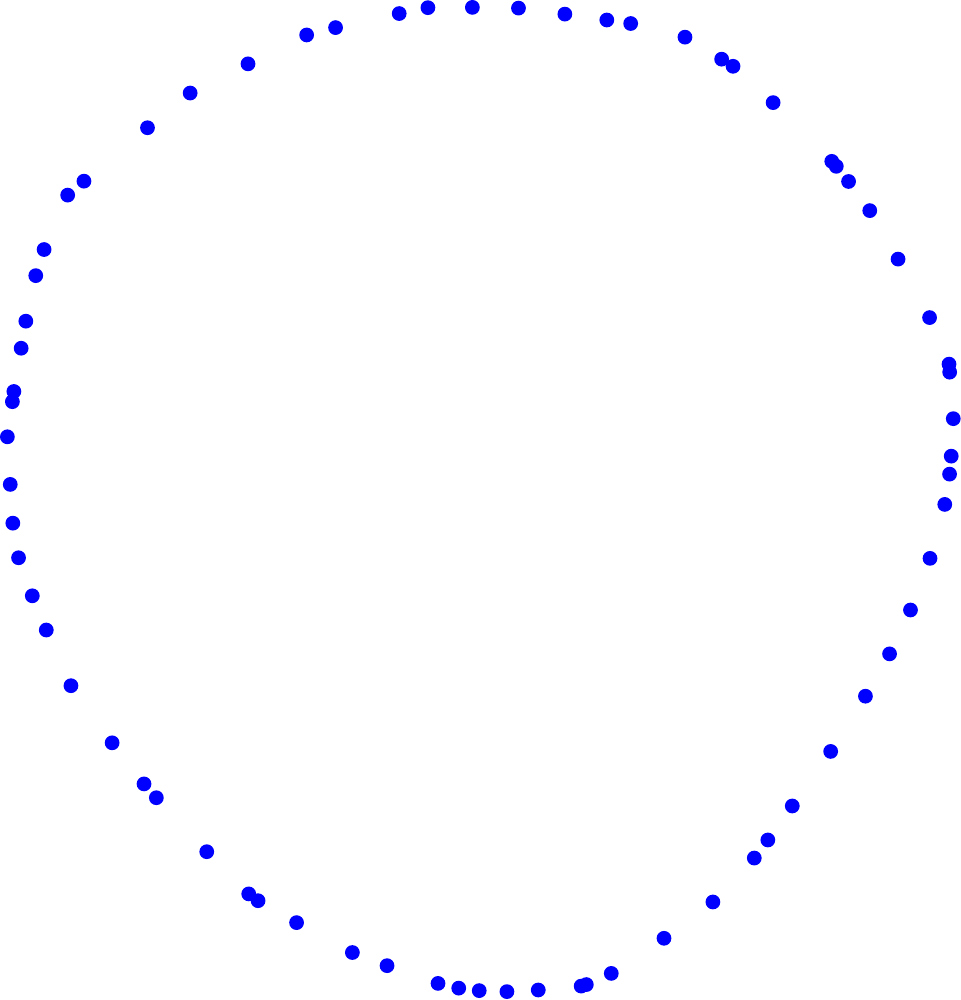}}
			\subcaption{\texttt{convex}}
		\end{subfigure}
		\begin{subfigure}[b]{0.22\linewidth}
			\centering
			\vspace{4pt}
			\resizebox{0.6\linewidth}{!}{\includegraphics{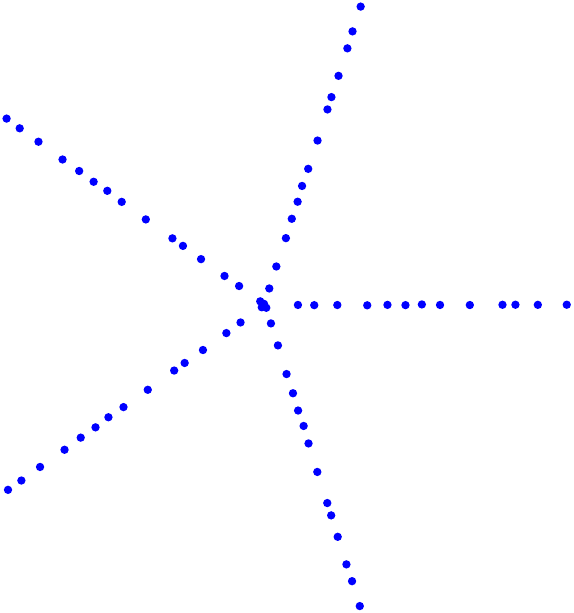}}
			\subcaption{\texttt{spokes}}
		\end{subfigure}
		\caption{The point distributions used in our experiments.}
		\label{fig:dist}
	\end{figure}

	\begin{enumerate}[label=(\alph*)]\itemsep0pt
		\item \texttt{uni-square}. The
		\texttt{CGAL::Random\_points\_in\_square\_2} generator was 
		used to generate points uniformly inside a square.
		\item \texttt{normal-clustered}. A set of $\sqrt{n}$ clusters placed
		randomly in the plane. Each cluster contains $\sqrt{n}$ normally
		distributed points (mean and standard-deviation are set to 2.0). The
		coordinates were 
		generated using \texttt{std::normal\_distribution<double>}.
		
		
		\item \texttt{grid-random}. Points were 
		generated on a $\lceil 0.7n \rceil \times \lceil 0.7n\rceil $ unit
		square grid. The value $0.7$ was 
		chosen arbitrarily to obtain well-spaced non-contiguous grid
		points. The coordinates are integers and were 
		generated independently using \texttt{std::uniform\_int\_distribution}.
		\item \texttt{annulus}. Points were 
		generated inside an annulus whose outer radius is set to 500 and the inner 
		radius was
		set to 400. We 
		used  \texttt{std::uniform\_real\_distribution} to generate the coordinates.
		
		\item \texttt{galaxy}. Points are generated in the shape of a spiral galaxy having outer five arms~\cite{itinerantgames_2014}.
		
		\item \texttt{convex}. Points were 
		generated using \texttt{CGAL::random\_convex\_set\_2}. 
		
		\item \texttt{spokes}{\tiny }. Points were 
		generated in the shape of five spokes. 
	\end{enumerate}

	\noindent
	\textit{Real-world pointsets.} The following real-world pointsets were
	used in our experiments. 
	Note that they are not known to follow any specific distribution. We removed the duplicate points from them before running the experiments.
	
	\begin{enumerate}[label=(\alph*)]\itemsep0pt
		\item \texttt{burma}~\cite{tsp,anderson2022bounded}. An $33,708$-element pointset representing cities in Burma.
		
		\item \texttt{birch3}~\cite{bus2018practical,ghosh2019unit,anderson2022bounded,friederich2023experiments}. An $99,999$-element pointset representing random clusters at random locations.
		
		\item \texttt{monalisa}~\cite{tsp,ghosh2019unit,anderson2022bounded,friederich2023experiments}: A $100,000$-city TSP instance representing a continuous-line drawing of the Mona Lisa. 
		
		\item  \texttt{KDDCU2D}~\cite{bus2018practical,anderson2022bounded,ghosh2019unit,friederich2023experiments}.  An $104,297$-element pointset representing the first two dimensions of a protein data-set.
		
		\item \texttt{usa}~\cite{tsp,ghosh2019unit,anderson2022bounded,friederich2023experiments}.  A $115,475$-city TSP instance representing (nearly) all towns, villages, and cities in the United States.
		
		\item \texttt{europe}~\cite{bus2018practical,ghosh2019unit,anderson2022bounded,friederich2023experiments}.  An $168,435$-element pointset representing differential coordinates of the map of Europe.
		
	\end{enumerate}

	\noindent
	\textit{Values for $t',k,h$.} Refer to Algorithm~\ref{alg:fss}. We
	used 
	$t'$ for WSPD-spanner construction on the leader points of the
	non-empty leaves. In our experiments, we found that $t'=1.25$ suffices
	when $1.1\leq t \leq 1.25$. However, when $t < 1.1$ or $t>1.25$, we
	used $t'=t$. 
	
		\begin{figure}[h]
		\centering
    \begin{minipage}{0.58\textwidth}
			\centering
		\includegraphics[scale=0.75]{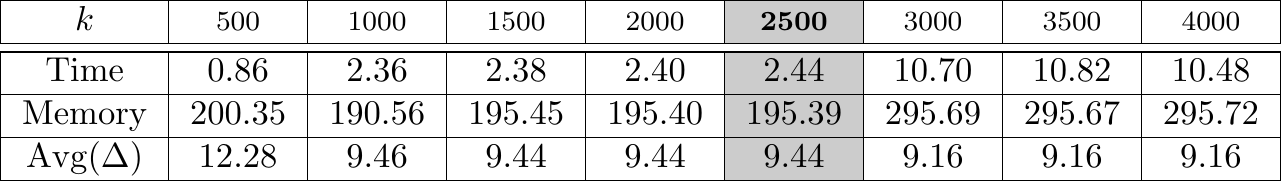}	
		\end{minipage}
    \begin{minipage}{0.4\textwidth}
			\centering
\includegraphics[scale=0.7]{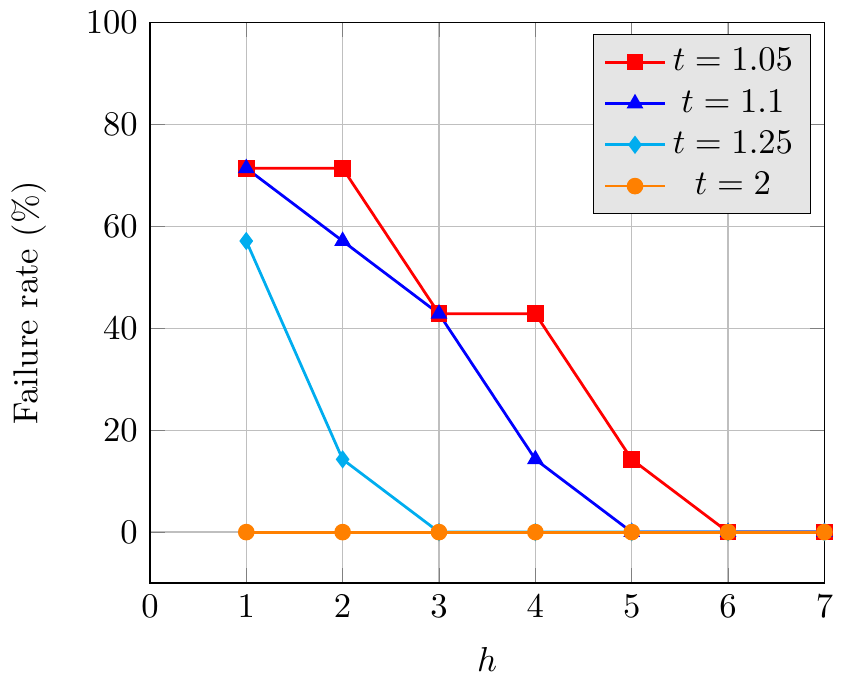} 
		\end{minipage}

		\caption{Left: Determining an appropriate value for $k$; Right: Determining appropriate values for $h$ for different values of $t$ used in our experiments.}
		\label{fig:kh-values}
	\end{figure}
	
	To find a suitable value for $k$, we ran an experiment with a $128K$-element pointset, drawn from the \texttt{uni-square} distribution. We varied $k$ from $500$ to $4000$, using increments of $500$. Refer to Fig.~\ref{fig:kh-values} (left). We observed that after $k=2500$, runtime increases substantially along with memory usage, and quite surprisingly, average-degree did not decrease, as one would expect. This motivated us to fix $k$ to $2500$ in our algorithm. We note that as $k$ increases, every execution of \textsc{FG-Greedy} gets more resource-intensive. Further, every execution of \textsc{Greedy-Merge} (Algorithm~\ref{alg:greedymerge}) and \textsc{Greedy-Merge-Light} (Algorithm~\ref{alg:greedymergelight}) used in Steps 4 and 5 also gets more expensive in terms of time and memory usage.  As $k$ approaches $n$, the output spanner converges to an actual greedy-spanner. 
	
	For choosing appropriate values for $h$, we  generated seven $1M$-element pointsets, one from each of the seven  synthetic distributions. Then, we ran  \textsc{Fast-Sparse-Spanner} on them by varying $h$ from $1$ to $7$. See Fig.~\ref{fig:kh-values} (right). In this case, the \emph{failure rate} is defined as the percentage of samples (out of $7$) where the desired stretch-factors were not achieved. For every $t \in \{1.05,1.1,1.25,2\}$ and $h \in \{1,2,\ldots,7\}$, we measured the  failure rates. The objective is to find the $h$-values that give us $0\%$ failure rates. Using the experimental data, the optimum $h$-values for the four stretch-factors $1.05,1.1, 1.25,2$ were set to $6,5,3,1$, respectively.
	
	\medspace
	
	\noindent
	\textit{Stretch-factors of the  \textsc{Fast-Sparse-Spanner} graphs.}  Surprisingly, 
	even though \textsc{Fast-Sparse-Spanner} only strives to get as close
	as possible to the desired stretch-factor, it never failed to achieve
	the	stretch-factor in our experiments.  To gain confidence, we ran our algorithm using pointsets drawn from the seven distributions. In this experiment, nine values of $n$ were chosen from  $\{1K,2K,\ldots,128K,256K\}$ and $t$  from  $\{1.05,1.1,1.25,2\}$. For every value of $n$, we have used $100$ samples. Thus in total, we conducted $7 \cdot 9 \cdot 100 \cdot 4 = 25,200$ trials and measured the stretch-factors of the output spanners using \textsc{Fast-Stretch-Factor}. Our algorithm  never missed the target stretch-factor. We have also performed a similar experiment on the aforementioned real-world pointsets using the four stretch-factors and observed that \textsc{Fast-Sparse-Spanner} could construct graphs with the desired stretch-factors. 
	
	\medspace

	\noindent
\textit{Comparison with the popular algorithms.} We have compared \textsc{Fast-Sparse-Spanner} with seven other popular algorithms mentioned in Section~\ref{sec:intro}. See Fig.~\ref{fig:8algs}. For every value of $n \in \{1K,2K,\ldots,128K\}$, we have drawn $10$ samples from the \texttt{uni-square} distribution. The stretch-factor was fixed at $1.1$ throughout the experiment.

	\begin{figure}[h]
	\centering
	\includegraphics{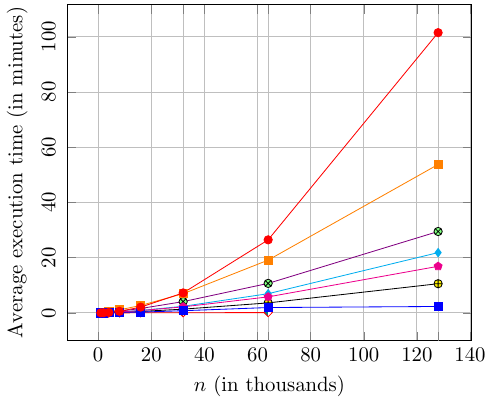} 
	\includegraphics{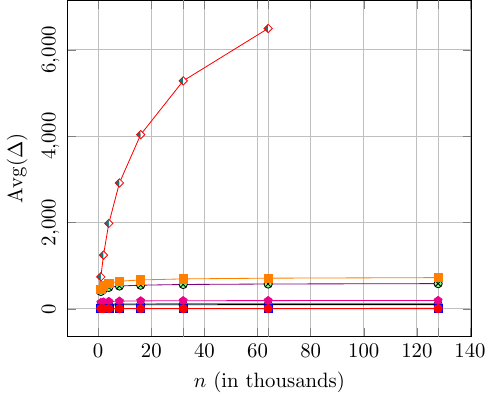}
	\includegraphics{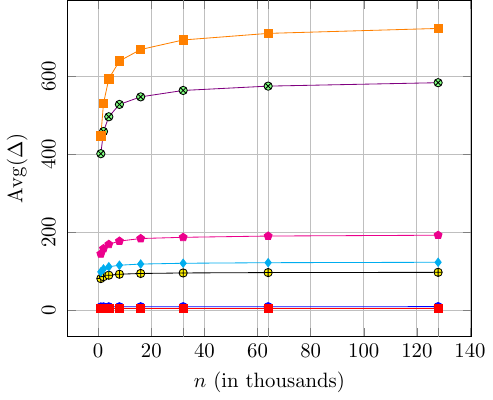}
	\includegraphics[scale=0.75]{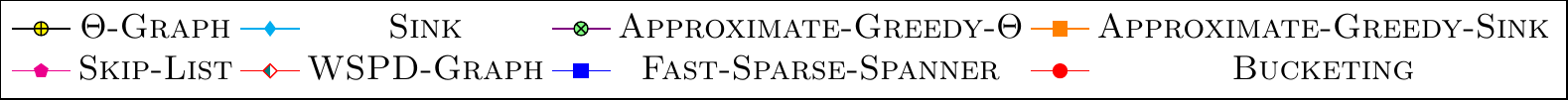}
	\caption{In this experiment, $n \in \{1K,2K,\ldots,128K\}$, $t=1.1$. The pointsets were drawn from the \texttt{uni-square} distribution. Left: Average execution times for the eight algorithms; Middle: Average-degrees of the spanners generated by the eight algorithms; Right: Average-degrees of the spanners generated by the eight algorithms without the WSPD-spanners. The WSPD-spanner algorithm crashed after $n=64K$.}
	\label{fig:8algs}
\end{figure}

 The \textsc{Approximate-Greedy} algorithm needs a bounded-degree spanner in its Step 1~\cite{das1997fast}. In the experimental paper~\cite{farshi2010experimental} by Farshi and Gudmundsson, the authors have used \textsc{Sink}-spanner in Step 1 since such spanners are bounded-degree. In our experiments, we have also used $\Theta$-\textsc{Graph} (although not guaranteed to be bounded-degree) in its place to observe how it performs. Thus, we have used two versions of the same algorithm, \textsc{Approximate-Greedy-$\Theta$} and \textsc{Approximate-Greedy-Sink}.

In terms of speed,  we found that WSPD-spanner algorithm was the fastest, but it ran out of memory after $n=64K$ since it tends to generate near-complete graphs for low values of $t$. Our algorithm was the second fastest, and \textsc{Bucketing} turned out to be the slowest. 

Average-degrees of the WSPD-graphs skyrocketed even for values of $n \leq 64K$. In this regard, \textsc{Bucketing} came out as the winner, and \textsc{Fast-Sparse-Spanner} was very close to \textsc{Bucketing}. The other six algorithms generated graphs that have substantially more average-degrees across all values of $n$. 
Both in terms of speed and average-degree, \textsc{Approximate-Greedy-$\Theta$} always beat \textsc{Approximate-Greedy-Sink} everywhere.
Hence, it can be concluded that \textsc{Bucketing} is our closest competitor when average-degree is the primary concern.
 
\medspace

	\noindent
	\textit{Experimental comparison with \textsc{Bucketing}.} Since the goal is to construct
	sparse spanners with low stretch-factors, we have used the following
	four values of $t$: $1.05,1.1,1.25,2$ for comparing 
	 \textsc{Fast-Sparse-Spanner} with \textsc{Bucketing}. The comparison is done based on runtime, memory usage,  average-degree, and diameter. Since \textsc{Bucketing} runs very slow, we have used $3$ trials for every value of $n$ and computed average runtimes, memory usages, average-degrees, and diameters. But for \textsc{Fast-Sparse-Spanner}, we have used $10$ trials since it runs fast to completion. 
	
	\begin{itemize}\itemsep0pt
		
		\item\text{\textit{Runtime}} and \text{\textit{memory usage}}. Fast
		execution speed and low memory usage are two desired
		characteristics of spanner construction algorithms when one tries to
		construct low-stretch-factor spanners on large pointsets. Refer to
		the top rows in the Figs.~\ref{fssvsbucketing:us},\ref{fssvsbucketing:nc},\ref{fssvsbucketing:gr},\ref{fssvsbucketing:an},\ref{fssvsbucketing:ga},\ref{fssvsbucketing:co},,\ref{fssvsbucketing:sp} for runtime
		and the bottom rows for memory usage comparisons for pointsets drawn from the seven pointset distributions (Fig.~\ref{fig:dist}). The legends can be found in Fig.~\ref{fig:8algs}. 
		In the synthetic pointset trials, we varied $n$ from $1K$ to $128K$, with increments by a factor of $2$.

		In our
		experiments with both synthetic and real-world pointsets, we found
		that \textsc{Fast-Sparse-Spanner} is remarkably faster than our
		closest competitor \textsc{Bucketing} as  $n$ gets larger. For $128K$-element  \texttt{uni-square} pointsets, the observed speedups are approximately $ 85, 41, 12 , 2$, for the four stretch-factors $1.05,1.1,1.25,2$, respectively. 
		Among the seven distributions used in our experiments, the highest speedups are observed for the
		\texttt{convex} and \texttt{galaxy} pointsets when $t=1.05$. For $n=128K$, the observed speedup is
		more than $1000$, in both the cases. For low values of $t$ such as $1.05,1.1,1.25$, our
		algorithm is substantially faster than \textsc{Bucketing}.  
		However, with the increase in $t$, \textsc{Bucketing} caught up with \textsc{Fast-Sparse-Spanner} (while still being noticeably slower). Overall, we observed that our algorithm behaved almost like a linear-time algorithm in practice.

					On the memory usage front, \textsc{Fast-Sparse-Spanner} beat \textsc{Bucketing} almost everywhere for large values of $n$. However, in some cases such as the \texttt{normal-clustered} distribution and $t=2$,  \textsc{Bucketing} used less memory but not by much.

				\begin{figure}[h]
			
			\centering
			
			\includegraphics[scale=0.8]{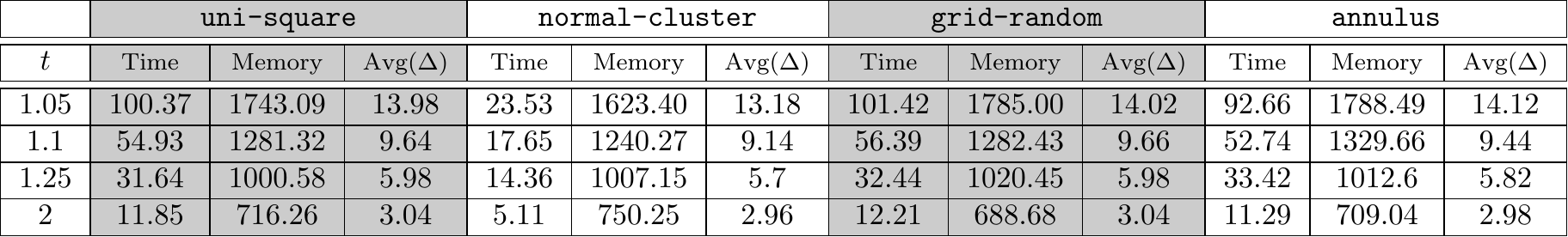}	
			\vspace{2pt}
			
			\includegraphics[scale=0.8]{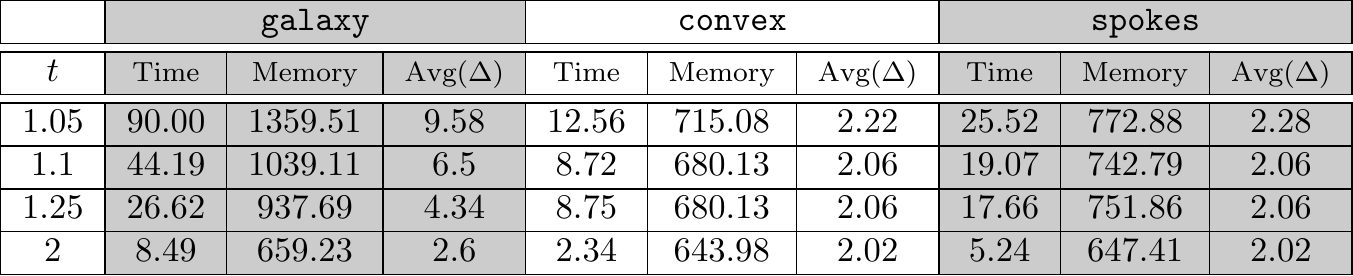}	
			\vspace{2pt}
			
			\caption{Single-threaded runtimes (in minutes), memory usages (in MB), and average-degrees for \textsc{Fast-Sparse-Spanner}, where $n=10^6$. }
			\label{fig:1million}
		\end{figure}

                Next we investigate the performance (time, memory and
                average degree) of \textsc{Fast-Sparse-Spanner} for
                different input point distributions and for different
                values of stretch factors.
		Refer to Fig.~\ref{fig:1million} 
                for the case when $n=10^6$. We found that even for such large pointsets and stretch-factor as low as $1.05$, our algorithm finished within two hours for all the distributions and used a reasonable amount of extra memory (at most $1.8$ GB). When $n=10^6$, \textsc{Bucketing} took around five days (not shown in the table) to construct a $1.1$-spanner on a pointset drawn from \texttt{uni-square} and used $5.8$ GB of main memory. In contrast, our algorithm took just $55$ minutes and used $\approx 1.3$ GB of main memory, making it roughly $130$ times faster. For \texttt{convex} pointsets, for $n=10^6$, our algorithm took a mere $13$ minutes to construct a $1.05$ spanner and used around $715$ MB of extra memory. Note that \textsc{Bucketing} took more than $20$ hours even for a $128K$-element \texttt{convex} pointset. Among the seven distributions, when $n=10^6$, our algorithm ran fastest on \texttt{convex} pointsets and slowest on \texttt{uni-square} and \texttt{grid-random} pointsets.

		Our algorithm beat \textsc{Bucketing} for all the real-world pointsets used in our experiments in terms of speed. See Fig.~\ref{fig:realWorld}. Once again, we observed that \textsc{Bucketing} struggled to complete with \textsc{Fast-Sparse-Spanner} for low values of $t$. For the used real-world pointsets, the best speed-up of $\approx 136$ was observed in the case of \texttt{europe} ($n=168,435$), when $t=1.1$. In most cases, our algorithm used less memory than \textsc{Bucketing}. We note that for some of the real-world pointsets viz., \texttt{burma}, \texttt{birch3}, and \texttt{europe} and certain values of $t$, \textsc{Bucketing} used less memory than \textsc{Fast-Sparse-Spanner} but the differences are quite tolerable in practice. The higher memory usages are marked in bold in Fig~\ref{fig:realWorld}. 
		
		Therefore, we conclude  that irrespective of the geometry of the input pointset, \textsc{Fast-Sparse-Spanner} is substantially faster than \textsc{Bucketing} in practice and in most cases tends to use much less amount of runtime memory.
		
			\begin{figure}
			\centering
			\includegraphics[scale=0.9]{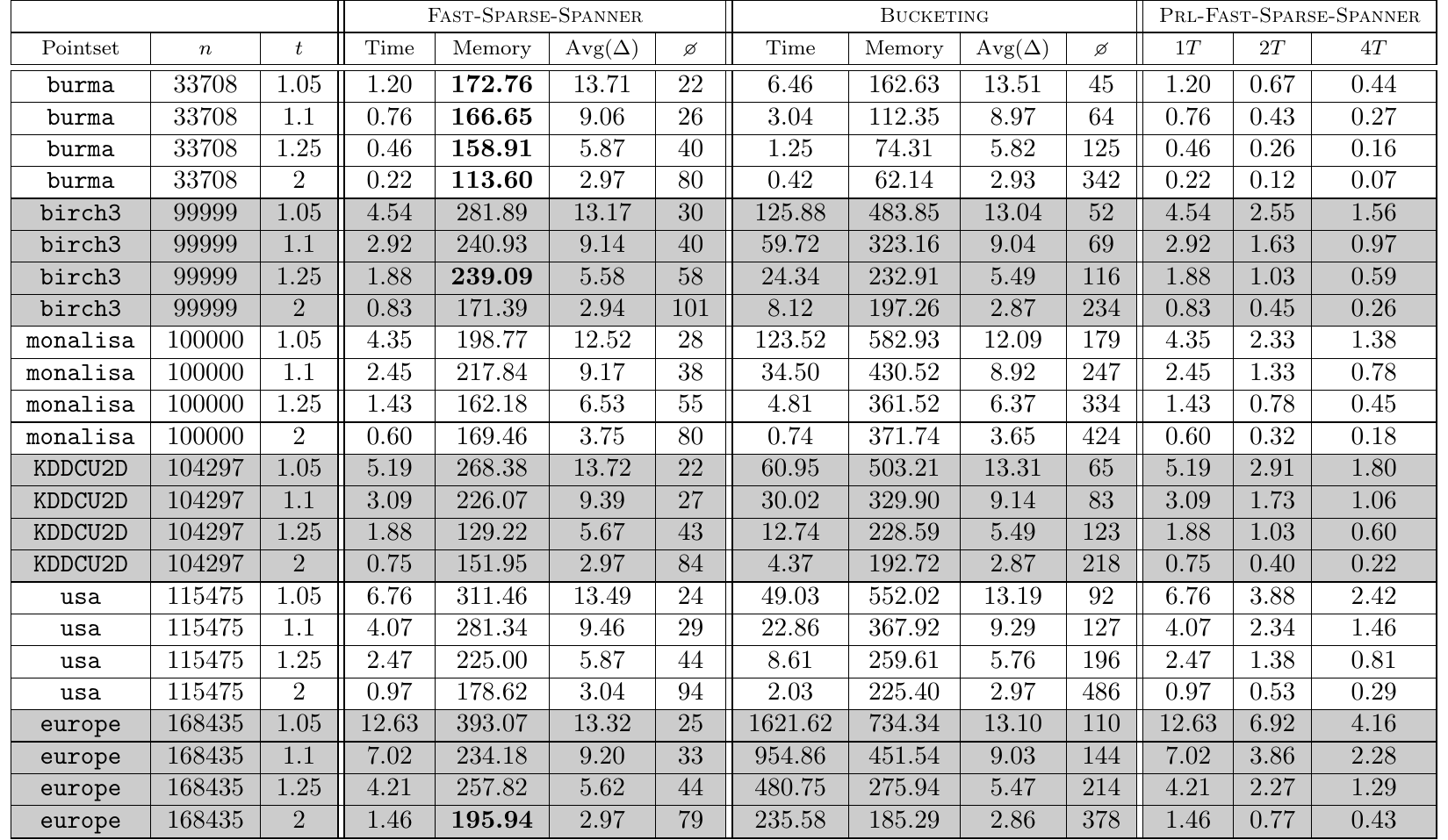}
			\caption{Runtime (in minutes), memory usage (in MB), average-degree (Avg($\Delta$)), and diameter ($\diameter$) in the case of real-world pointsets  are shown.  In this table, \textsc{Prl-Fast-Sparse-Spanner} stands for  \textsc{Parallel-Fast-Sparse-Spanner}; the columns $1T,2T,4T$ show the runtimes for $1$ thread (same as the fourth column in this table), $2$ threads, and $4$ threads, respectively. }
			\label{fig:realWorld}
		\end{figure}

		\item{\textit{Average-degree} (Avg($\Delta$)).} 
		For low stretch-factor spanners, low average-degree is desired since it helps to optimize per-node communication cost. With the decrease in $t$, average-degree tends to increase since more edges must be placed to ensure $t$-paths between all vertex pairs in a spanner. Greedy-spanners are unbeatable in this aspect and there is no other practical algorithm (fast and memory-efficient) that can produce spanners having near-greedy average-degree. In our experiments, we found that average-degrees of the spanners produced by \textsc{Fast-Sparse-Spanner} are not far from that of the greedy-spanners. For average-degree comparisons for synthetic pointsets, refer to the Figs.~\ref{avd:us},\ref{avd:nc},\ref{avd:gr},\ref{avd:an},\ref{avd:ga},\ref{avd:co},\ref{avd:sp} and Fig.~\ref{fig:realWorld} for real-world pointsets. 
		
		For $t=1.05,1.1,1.25,2$, we observed that the average-degrees of the spanners produced by \textsc{Fast-Sparse-Spanner} never exceeded $14.2, 9.7, 6, 3.75$, respectively. The observed differences in average-degrees of the spanners produced by the two algorithms never exceeded a unit. For smaller pointsets, the difference is much less. 
		It shows that our algorithm produced near-greedy size spanners everywhere. In the case of synthetic pointsets, the average-degrees of the spanners did not increase by much even when $n$ is set to $1M$ (see Fig.~\ref{fig:1million}), corroborating the fact that our algorithm is efficient when it comes to placements of edges. For example, when $t$ was set to $1.1$, the average-degree of the spanner produced by \textsc{Bucketing} was observed to be $9.08$ and that of the spanner produced by our algorithm was $9.64$.  We believe that such small increments in average-degrees should be acceptable in most practical applications given that \textsc{Bucketing} runs for days and our algorithm runs to completion in a couple of hours (much faster when multiple threads are used) on $1M$-element pointsets.  
		Further, \textsc{Fast-Sparse-Spanner}   tends to place fewer edges with the increase of $t$, a desired property of spanner construction algorithms. 
		
		Step 2 in our algorithm, where greedy-spanners are constructed inside every non-empty leaf, contributed the most in average-degree of the final spanner.  
		We found that the number of edges placed in steps 3,4 was very low and contributed marginally to average-degree. Step 5 did not put any edges for some pointsets, such as those drawn from the \texttt{uni-square} distribution. 
		To illustrate, for $128K$-element \texttt{uni-square} pointsets, step 2 added $\approx 8.46$ to the average-degree of the final spanner and steps 3, 4, 5 added $\approx 0.23,0.65,0$, respectively. In contrast, for $128K$-element \texttt{normal-clustered} pointsets, the contributions are approximately $8.5,0.1,0.01,0.002$, for steps 2, 3, 4, 5, respectively. Empty spaces (holes) inside the convex hull of the pointsets force our algorithm to place a few edges (compared to $n$) in step 5 for ensuring $t$-paths between distant point pairs whose parent leaves are within $h$ hops of each other in $G_T$, separated by a void (a sequence of empty leaves).

		\item{\textit{Diameter} ($\diameter$).} As mentioned earlier,  the \textit{diameter} of a graph $H$ is defined as the length (number of edges) of the longest shortest path among all vertex pairs in $H$. In our case, the presence of long WSPD edges (placed in Step 3) gives us the benefit of low diameter in most cases. We found that the diameter of the spanners produced by our algorithm is substantially less than the greedy-spanners in most cases, without any intolerable increase in average-degrees. See Fig.~\ref{fig:realWorld} for the diameters obtained for  real-world points and Figs.~\ref{dia:us},\ref{dia:nc},\ref{dia:gr},\ref{dia:an},\ref{dia:ga},\ref{dia:co},\ref{dia:sp} for the synthetic distributions. 
		
		In the case of real-world pointsets, \textsc{Fast-Sparse-Spanner} could always beat \textsc{Bucketing}. In this regard, we observed the best performance for the \texttt{monalisa} pointset ($n=100K$),  $t=1.1$; the diameter of the spanner produced by \textsc{Bucketing} was $6.5$ times more than that of the one constructed by \textsc{Fast-Sparse-Spanner}. Further, our algorithm placed just $\approx 2.8\%$ extra edges, was $\approx14$ times faster, and used $\approx 50\%$ less memory.  
		
		For the \texttt{uni-square}, \texttt{normal-clustered}, \texttt{grid-random}, \texttt{annulus}, and \texttt{galaxy} distributions, \textsc{Fast-Sparse-Spanner} always produced spanners having considerably lower diameter for all the four values of stretch-factors. However, we noticed that for the \texttt{convex} and \texttt{spokes} distributions when $t$ was set to $1.05$, \textsc{Bucketing} beat our algorithm by a considerable difference. In those two cases, the WSPD edges could not help much in reducing the diameter of the spanners. However, with the increase in $t$, our algorithm took the lead. For all the seven distributions, the differences in diameters of the spanners generated by the two algorithms increased with the increase in $t$ ($\geq 1.1$); \textsc{Fast-Sparse-Spanner} generated spanners with far less diameter than the ones constructed by \textsc{Bucketing}. This is expected since greedy-spanners tend to have shorter edges when $t$ is increased. In contrast, \textsc{Fast-Sparse-Spanner} still places long WSPD edges in the spanners, no matter how large $t$ is,  thereby vastly reducing the diameters of the spanners. 
	\end{itemize}
	
	\noindent
	\textit{Efficacy of \textsc{Fast-Stretch-Factor}.}  We compare the Dijkstra-based stretch-factor measurement algorithm (run from every vertex) with our algorithm \textsc{Fast-Stretch-Factor}.  Refer to Fig.~\ref{fig:fs-sf}. We have used two values of $n$, $256K$ and $1M$.  Further, we fixed $t$ to $1.1$ in our experiment since the same trend was observed for all values of $t$. 
	 
	\begin{figure}[h]
		\centering
		\includegraphics[scale=0.8]{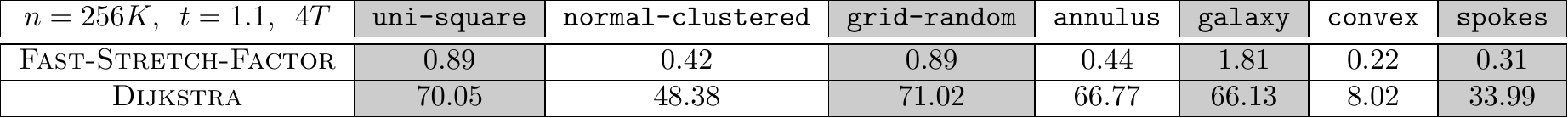}\\
		\vspace{4pt}
		\includegraphics[scale=0.8]{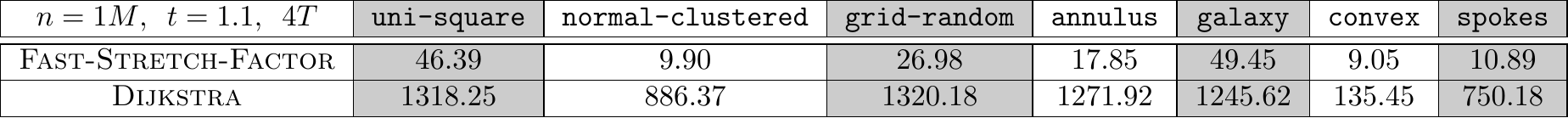}
		\caption{\textsc{Fast-Stretch-Factor} vs \textsc{Dijkstra}: $t$ was fixed to $1.1$ and $4$ threads were used.  Top: $n=256K$; Bottom: $n=10^6$. The reported times are in minutes.}
		\label{fig:fs-sf}
	\end{figure}

	 \textsc{Fast-Stretch-Factor} is easily parallelizable since the unmerged leaf pairs can be considered independently in the main for-loop. Similarly, the Dijkstra-based algorithm can be parallelized by executing the shortest path computations from every vertex in parallel.
	As both are easily parallelizable, we have used $4$ threads in our comparison experiment. It is clearly evident that in both cases ($n=256K,1M$), \textsc{Fast-Stretch-Factor} outperformed the Dijkstra-based algorithm. For instance, for a $1.1$-spanner, constructed on a $1M$-element \texttt{uni-square} pointset by \textsc{Fast-Sparse-Spanner}, the Dijkstra-based algorithm took $1318.25$ minutes (around $22$ hours), whereas \textsc{Fast-Stretch-Factor} took just $46.39$ minutes making our algorithm $\approx 28.42$ times faster.  The best speedup of $\approx 89.5$ was achieved for the \texttt{normal-clustered} distribution, $n=1M$, where \textsc{Fast-Stretch-Factor} could finish up within $10$ minutes but \textsc{Dijkstra} took around $14.8$ hours. 
	
	The main reason behind the speed of our algorithm is the avoidance of unnecessary graph explorations done from every vertex of the input spanner by the Dijkstra-based algorithm. \textsc{Fast-Sparse-Spanner} explores the spanner in a restrictive way using \textsc{Greedy-Path} and $A^*$. Further, it works only with the leaf pairs which were not considered during the spanner construction in \textsc{Fast-Sparse-Spanner} and tries to reuse the \textsc{Greedy-Path} and $A^*$ paths to check for the existence of $t$-paths between point pairs from two different leaves. Further, due to the long edges of $W$ in $H$, we found that the \textsc{Greedy-Path} and $A^*$ executions ran to completion fast. For instance, for a $1.1$-spanner generated by \textsc{Fast-Sparse-Spanner} on a $1M$-element \texttt{uni-square} pointset, \textsc{Greedy-Path} and $A^*$ explored approximately $1690$ and $10700$ vertices on average. Note that these numbers are much less than $1M$. Further, \textsc{Greedy-Path} was successful $\approx 92.6\%$ of the times in finding $t$-paths. Consequently, in practice, for every pair, the time taken to compute $t$-paths is much less than linearithmic, as assumed in Theorem~\ref{thm:sf}. As a result, \textsc{Fast-Stretch-Factor} behaves like a quadratic-time algorithm in practice and could easily beat \textsc{Dijkstra} everywhere. 

	\medskip
	
	\noindent
	\textit{Parallel \textsc{Fast-Sparse-Spanner}.} Parallelization of \textsc{Fast-Sparse-Spanner} is straightforward. In Step 1, we construct local greedy-spanners in parallel. We always found that the runtimes of steps 1 and 2 are a minuscule of the total runtime. So, we did not parallelize it. The mergings in steps 4 and 5 can be easily executed in parallel. In our experimental results, we set the number of threads to $1,2$ and $4$.  
	We have leveraged \textsf{OpenMP} for parallelization. To see
        thread-dependent runtimes, refer to Fig.~\ref{fig:realWorld}
        for real-world pointsets and
        Figs.~\ref{multithread:us},\ref{multithread:nc},\ref{multithread:gr},\ref{multithread:ann},\ref{multithread:ga},\ref{multithread:co},\ref{multithread:sp}
        for the synthetic distributions.

	Speed-ups are remarkably close to the number of threads used. 
        However, we point out that with the increase in threads, the memory usage increases because of the per-thread space requirements for the mergings in Steps 4 and 5. For instance, on a $128K$-element \texttt{uni-square} pointset, our implementation consumed approximately $ 189, 285,  473$ MB of main memory for $1,2,4$ threads, respectively.
	
	\clearpage
	
	\begin{figure}
		\centering
\includegraphics[scale=0.85]{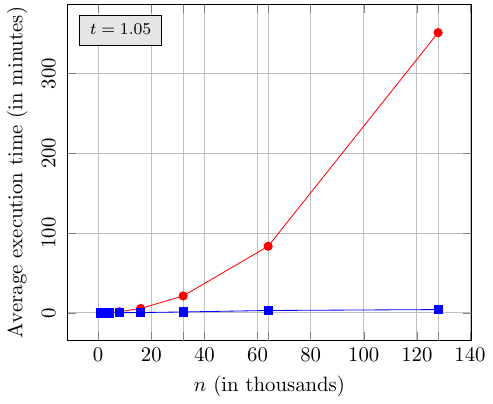}
\includegraphics[scale=0.85]{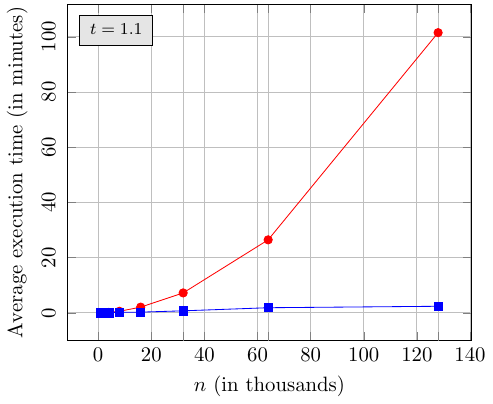}
\includegraphics[scale=0.85]{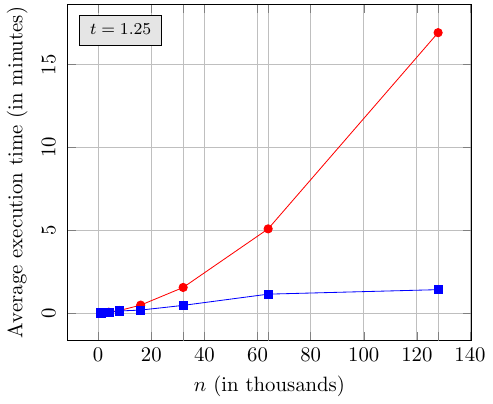}
\includegraphics[scale=0.85]{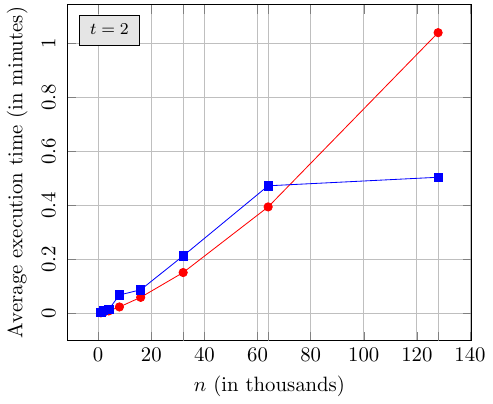}
		
\includegraphics[scale=0.85]{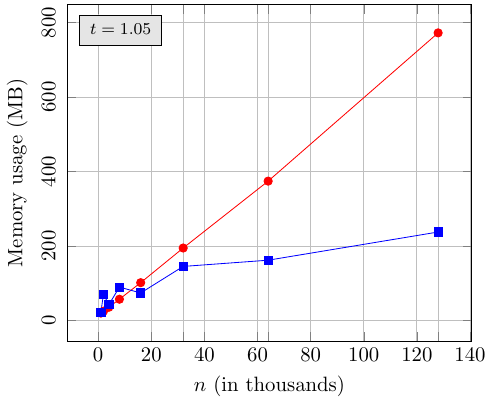}
\includegraphics[scale=0.85]{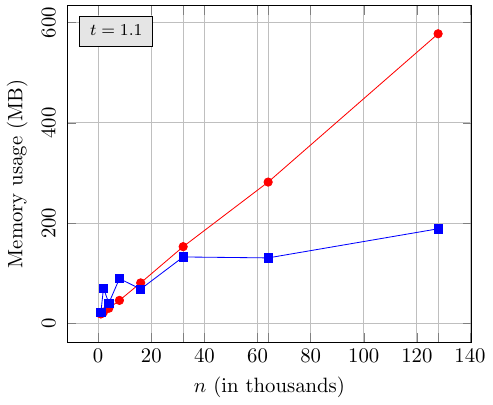}
\includegraphics[scale=0.85]{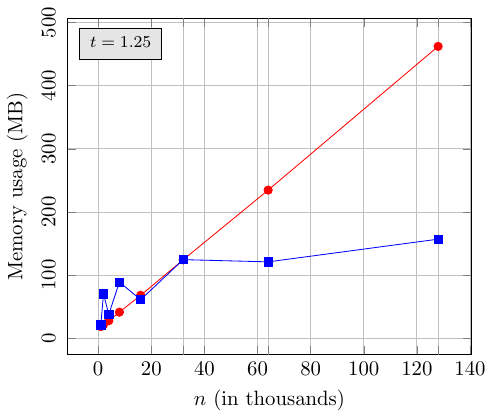}
\includegraphics[scale=0.85]{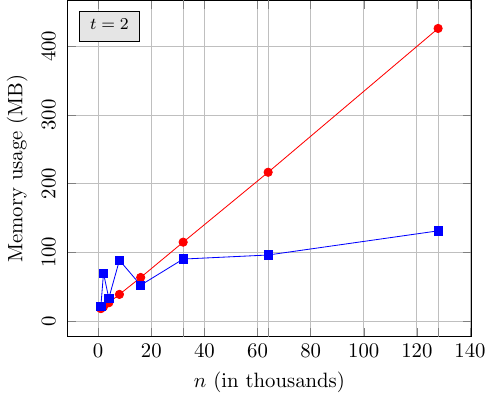}

\vspace{-10pt}
\caption{Time and memory usage comparisons for the \texttt{uni-square} distribution.}
		\label{fssvsbucketing:us}
	\end{figure}

	\begin{figure}
		\centering

		\includegraphics[scale=0.85]{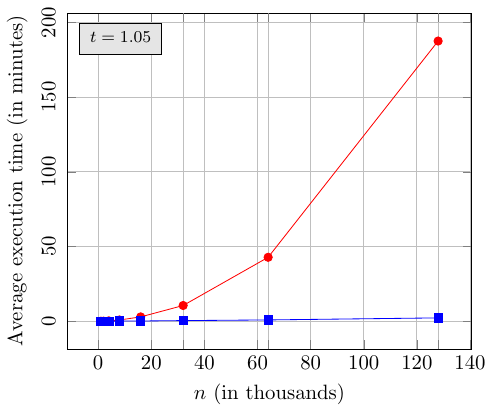}
		\includegraphics[scale=0.85]{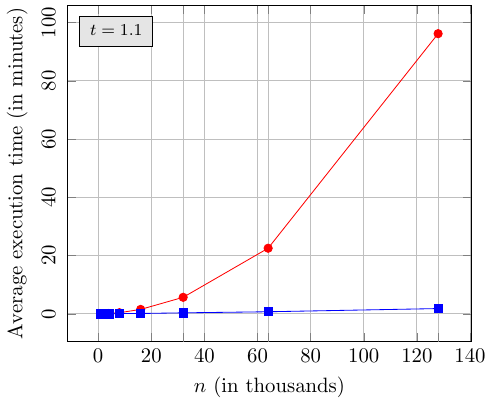}
		\includegraphics[scale=0.85]{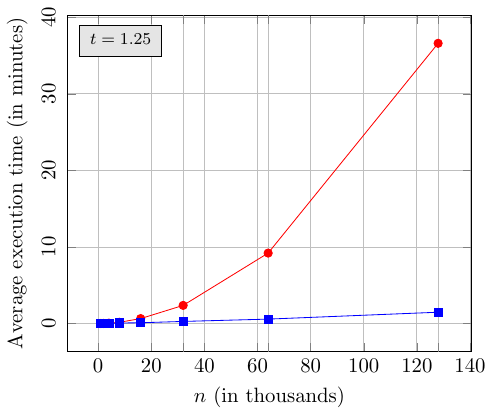}
		\includegraphics[scale=0.85]{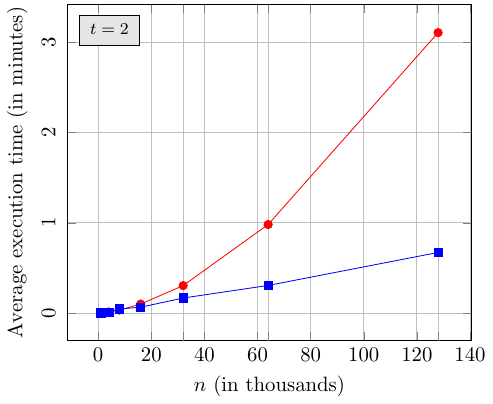}
		
		\includegraphics[scale=0.85]{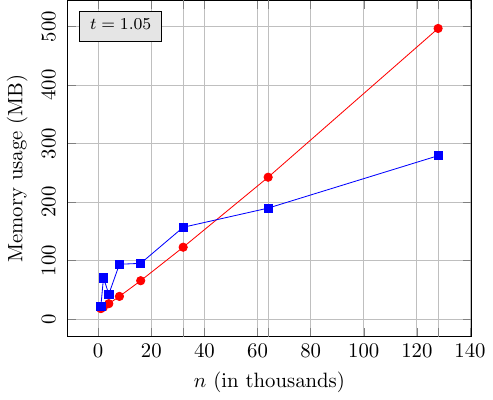}
		\includegraphics[scale=0.85]{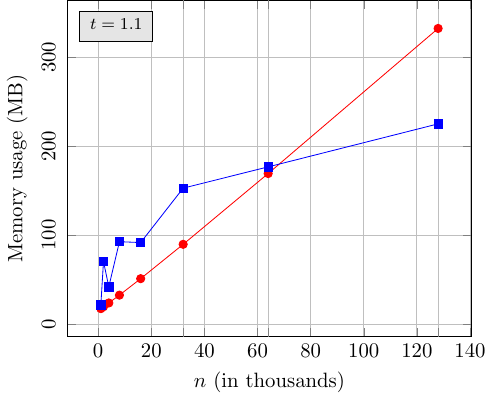}
		\includegraphics[scale=0.85]{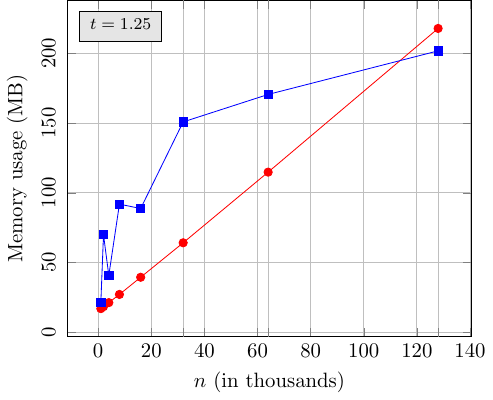}
		\includegraphics[scale=0.85]{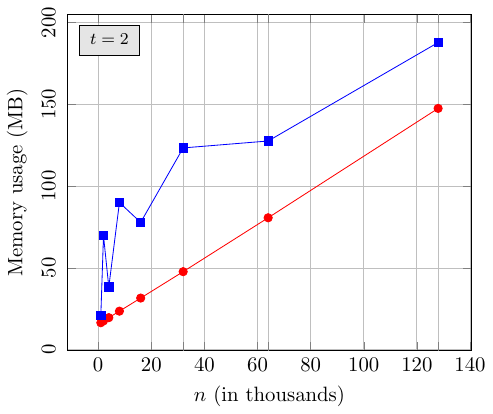}
		\vspace{-10pt}
\caption{Time and memory usage comparisons for the \texttt{normal-clustered} distribution.}
		\label{fssvsbucketing:nc}
	\end{figure}

	\begin{figure}
	\centering
	\includegraphics[scale=0.85]{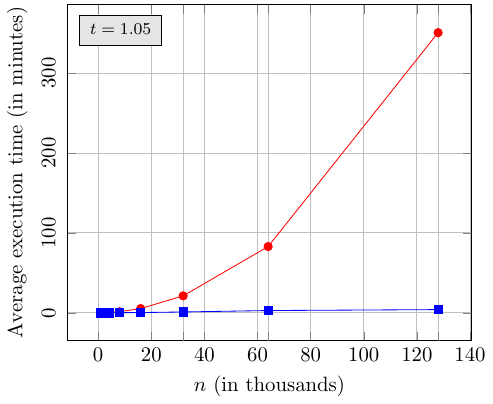}
	\includegraphics[scale=0.85]{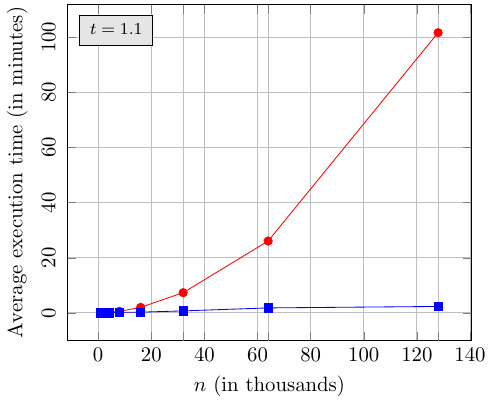}
	\includegraphics[scale=0.85]{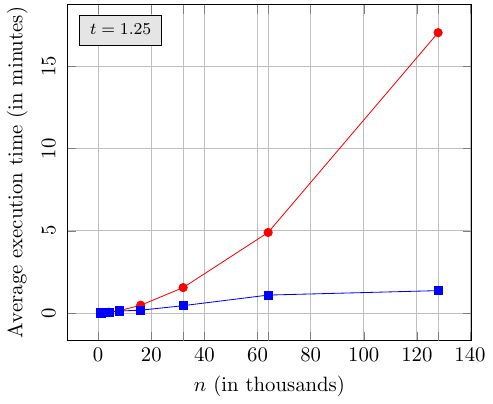}
	\includegraphics[scale=0.85]{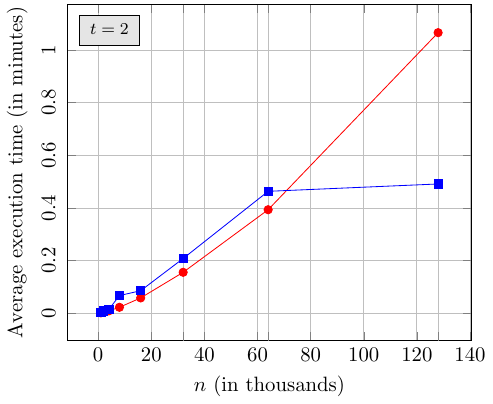}
	\includegraphics[scale=0.85]{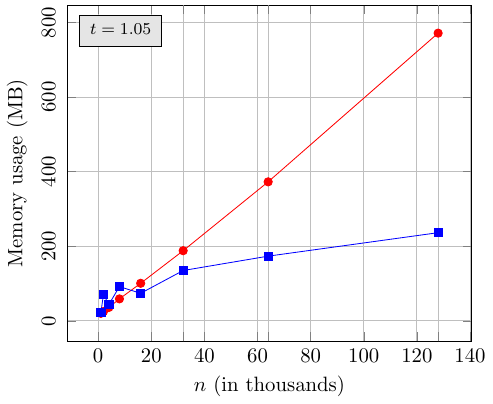}
	\includegraphics[scale=0.85]{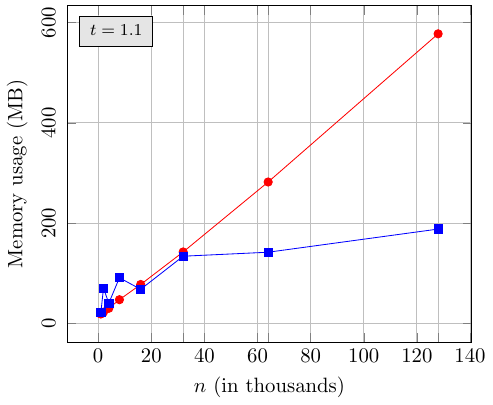}
	\includegraphics[scale=0.85]{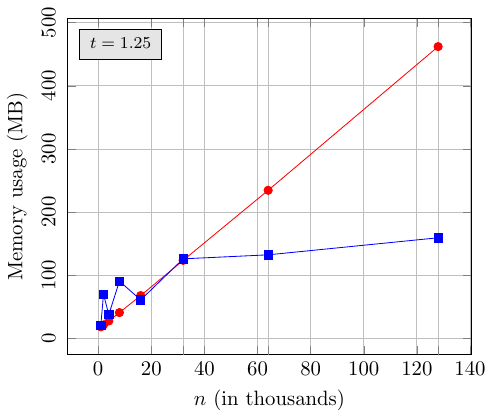}
	\includegraphics[scale=0.85]{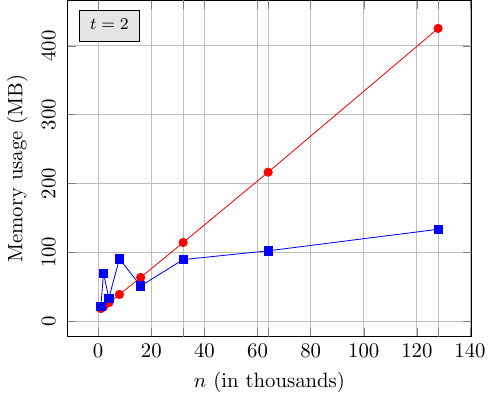}
	\vspace{-10pt}
	\caption{Time and memory usage comparisons for the \texttt{grid-random} distribution.}
	\label{fssvsbucketing:gr}
\end{figure}
	
		\begin{figure}
		\centering
		\includegraphics[scale=0.85]{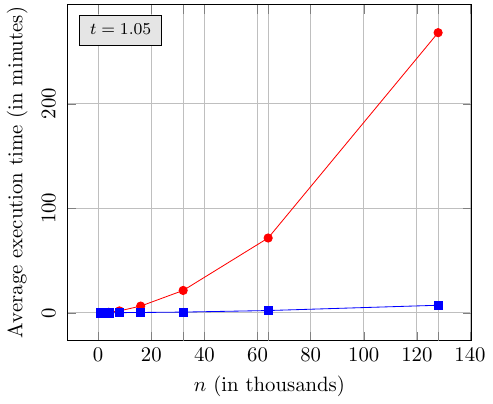}
		\includegraphics[scale=0.85]{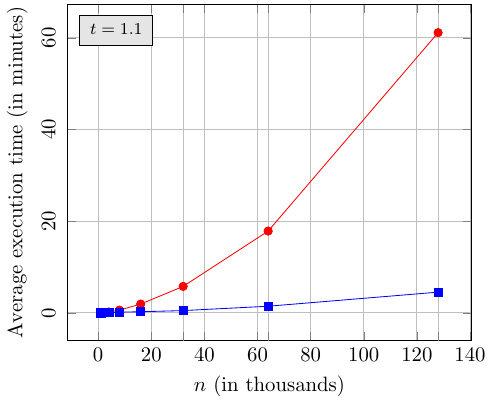}
		\includegraphics[scale=0.85]{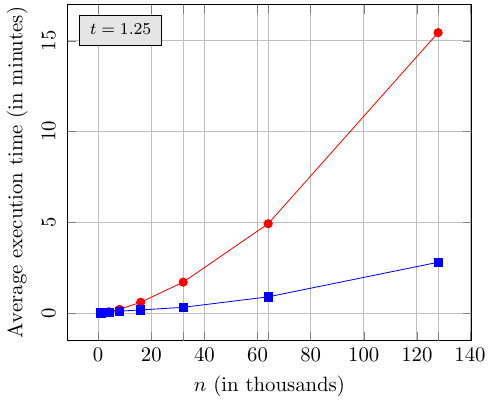}
		\includegraphics[scale=0.85]{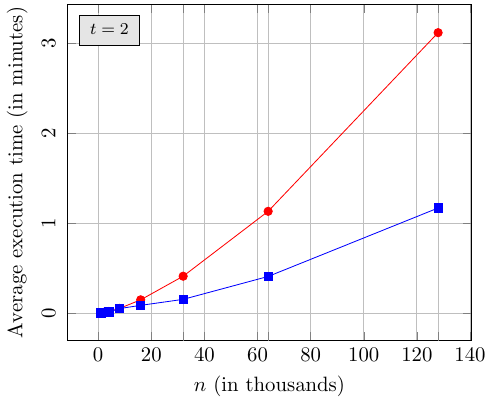}
		\includegraphics[scale=0.85]{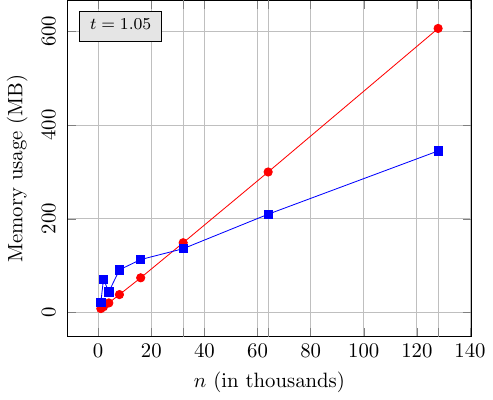}
		\includegraphics[scale=0.85]{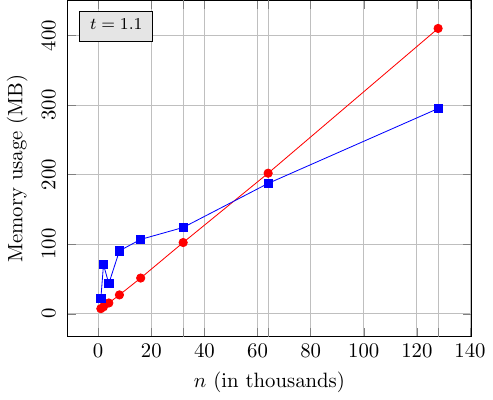}
		\includegraphics[scale=0.85]{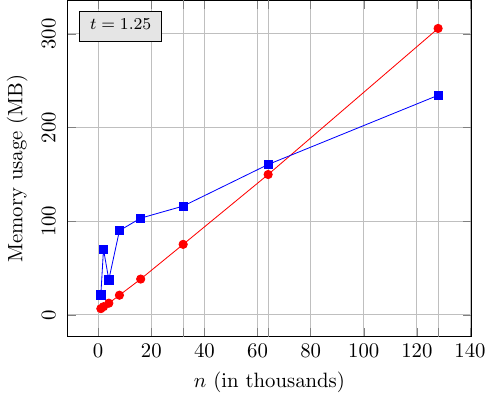}
		\includegraphics[scale=0.85]{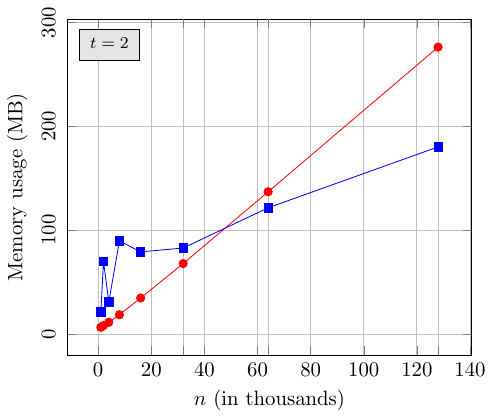}
		\vspace{-10pt}
		\caption{Time and memory usage comparisons for the \texttt{annulus} distribution.}
		\label{fssvsbucketing:an}
	\end{figure}

			\begin{figure}
		\centering
		\includegraphics[scale=0.85]{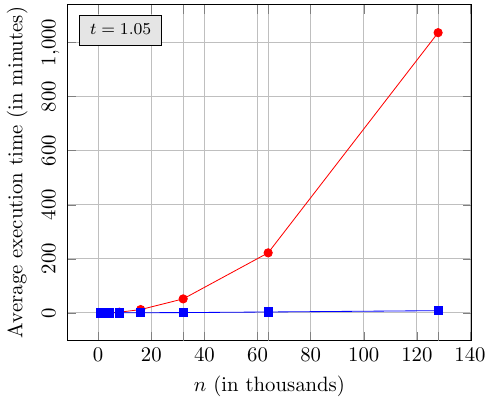}
		\includegraphics[scale=0.85]{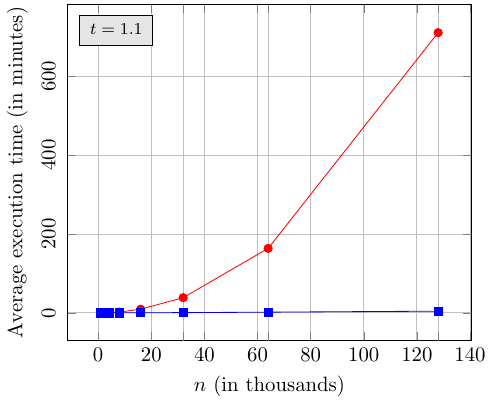}
		\includegraphics[scale=0.85]{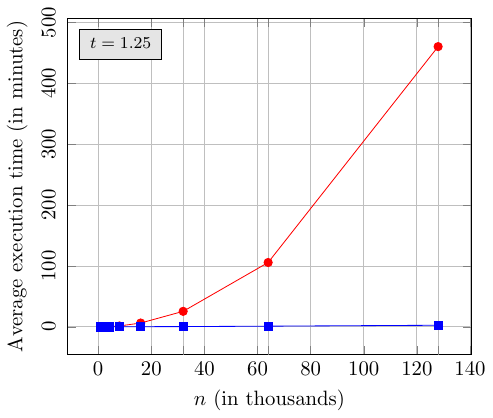}
		\includegraphics[scale=0.85]{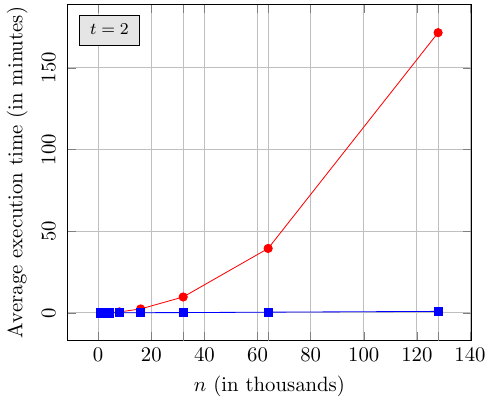}
		\includegraphics[scale=0.85]{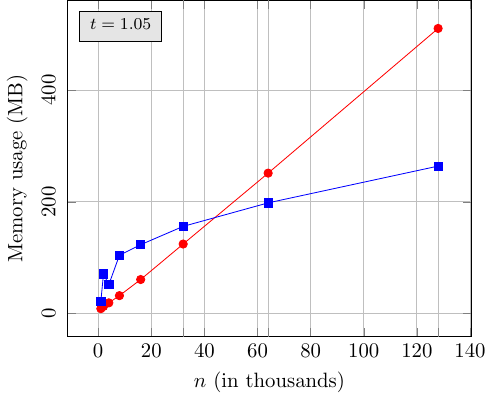}
		\includegraphics[scale=0.85]{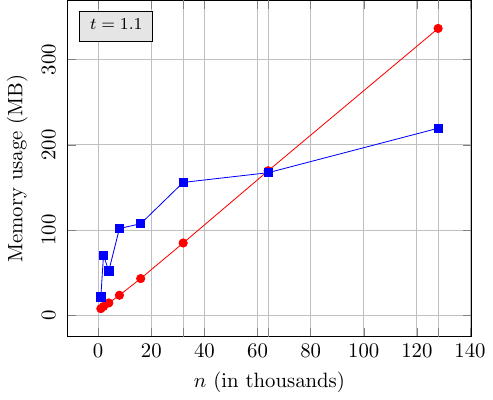}
		\includegraphics[scale=0.85]{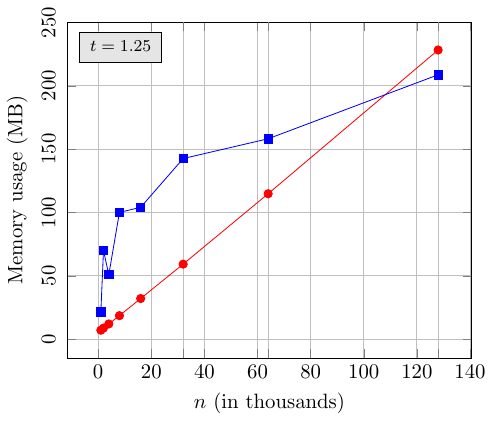}
		\includegraphics[scale=0.85]{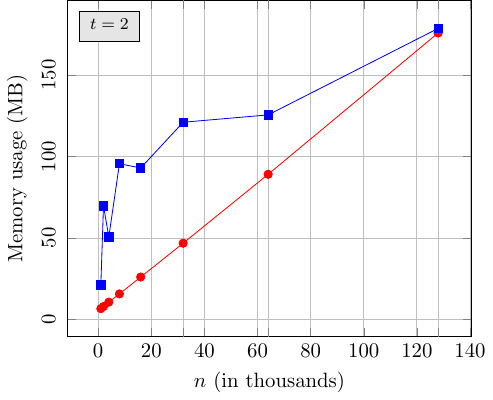}
		\vspace{-10pt}
		\caption{Time and memory usage comparisons for the \texttt{galaxy} distribution.}
		\label{fssvsbucketing:ga}
	\end{figure}

		\begin{figure}
			\centering
			\includegraphics[scale=0.85]{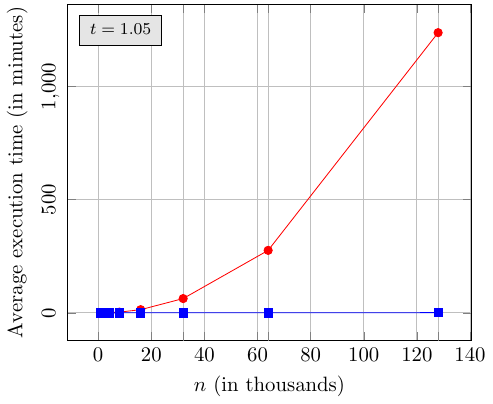}
			\includegraphics[scale=0.85]{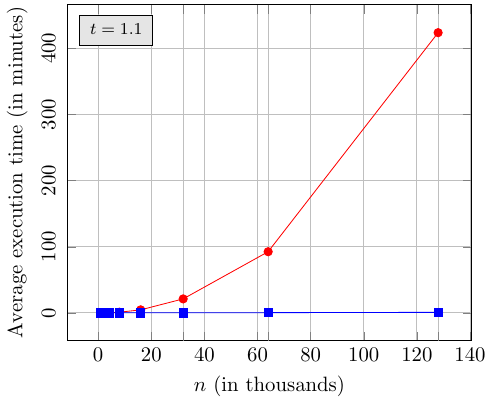}
			\includegraphics[scale=0.85]{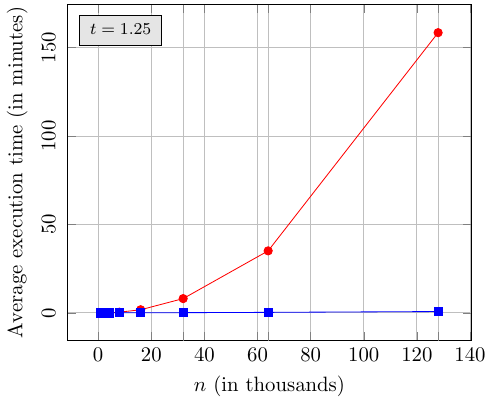}
			\includegraphics[scale=0.85]{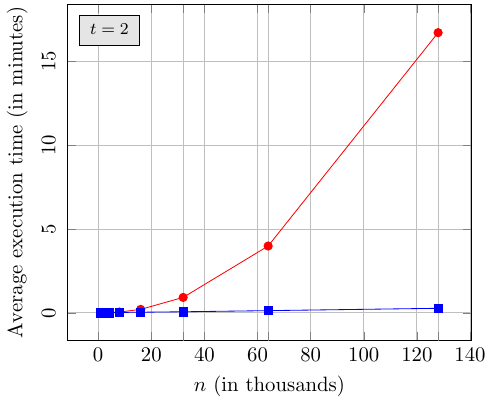}
			\includegraphics[scale=0.85]{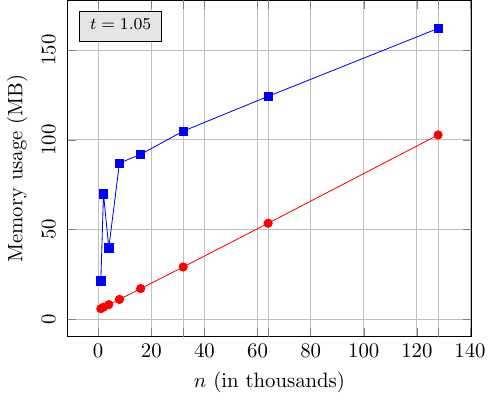}
			\includegraphics[scale=0.85]{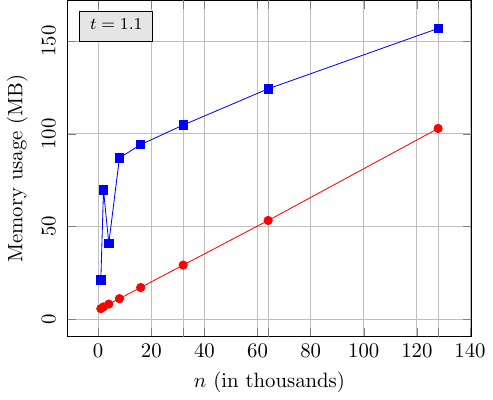}
			\includegraphics[scale=0.85]{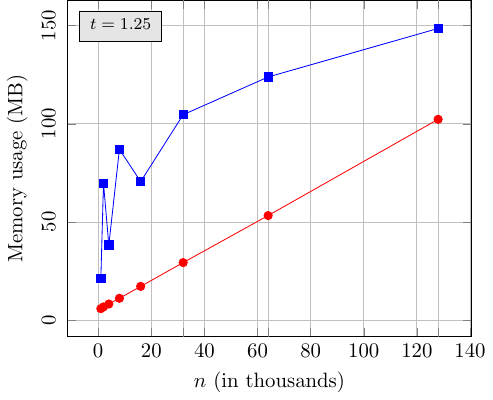}
			\includegraphics[scale=0.85]{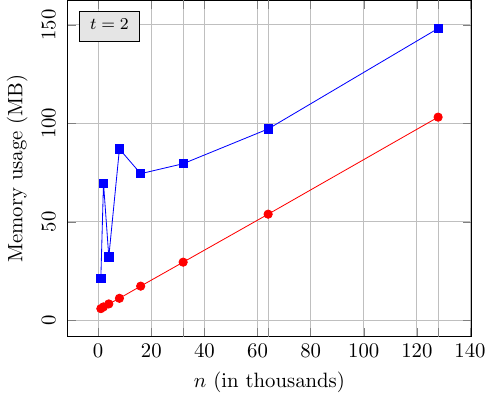}
			\vspace{-10pt}
			\caption{Time and memory usage comparisons for the \texttt{convex} distribution.}
			\label{fssvsbucketing:co}
			\end{figure}

		\begin{figure}
	\centering
	\includegraphics[scale=0.85]{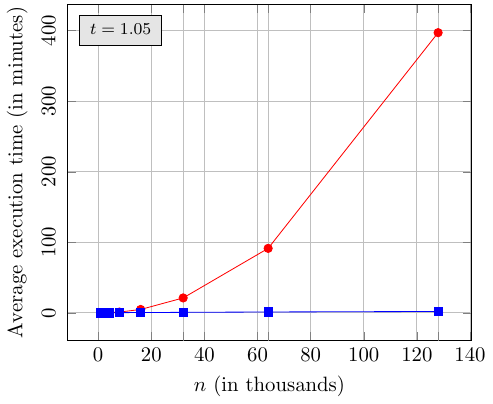}
	\includegraphics[scale=0.85]{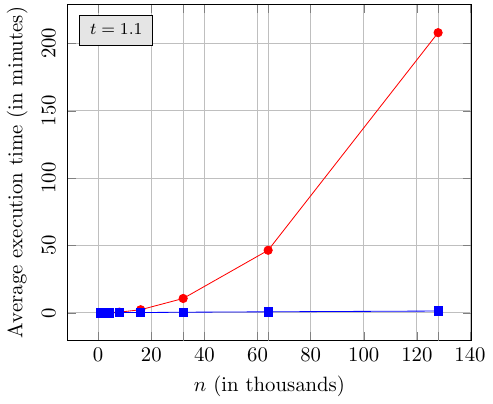}
	\includegraphics[scale=0.85]{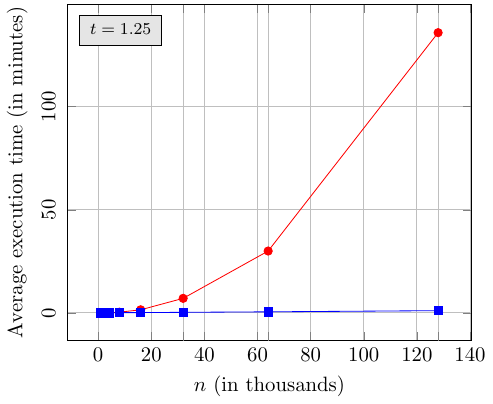}
	\includegraphics[scale=0.85]{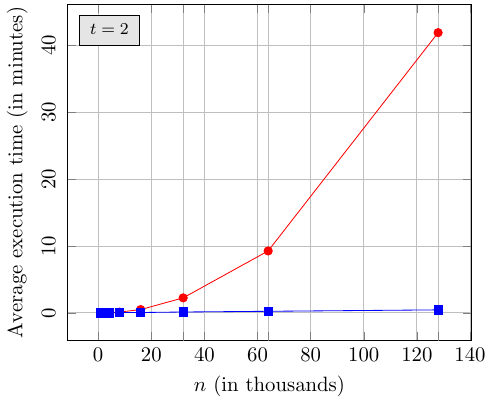}
	\includegraphics[scale=0.85]{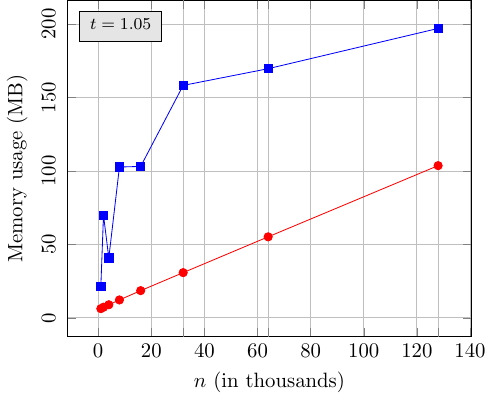}
	\includegraphics[scale=0.85]{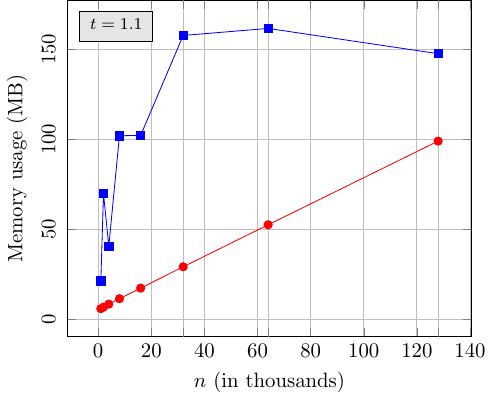}
	\includegraphics[scale=0.85]{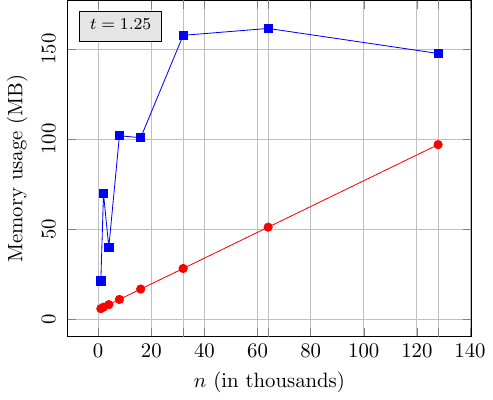}
	\includegraphics[scale=0.85]{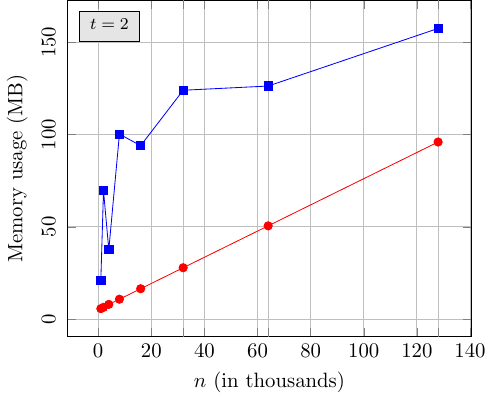}
	\vspace{-10pt}
	\caption{Time and memory usage comparisons for the \texttt{spokes} distribution.}
	\label{fssvsbucketing:sp}
	
\end{figure}
		\begin{figure}
		\centering
		\includegraphics[scale=0.85]{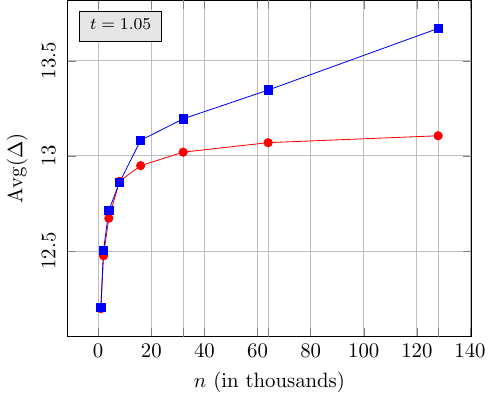}
		\includegraphics[scale=0.85]{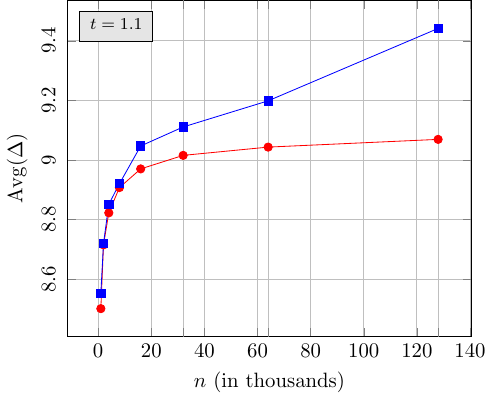}
		\includegraphics[scale=0.85]{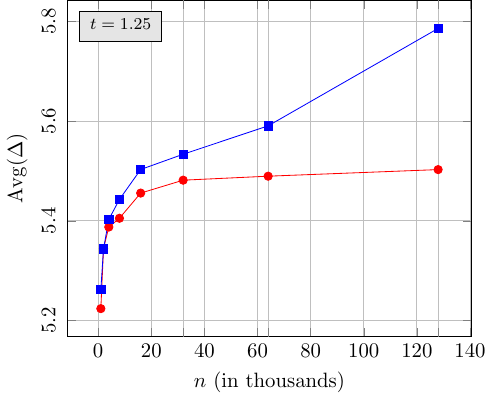}
		\includegraphics[scale=0.85]{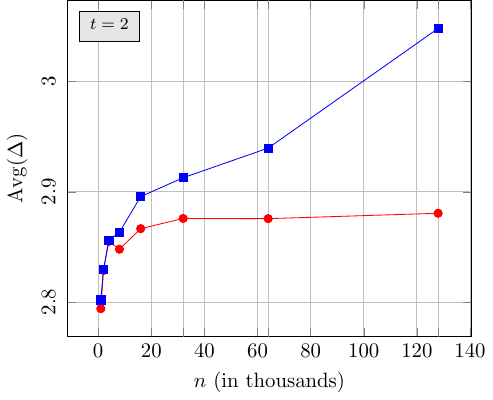}
		\vspace{-10pt}
		\caption{Average-degree comparisons for the \texttt{uni-square} distribution.}
		\label{avd:us}
	\end{figure}

	\begin{figure}
		\centering
		\includegraphics[scale=0.85]{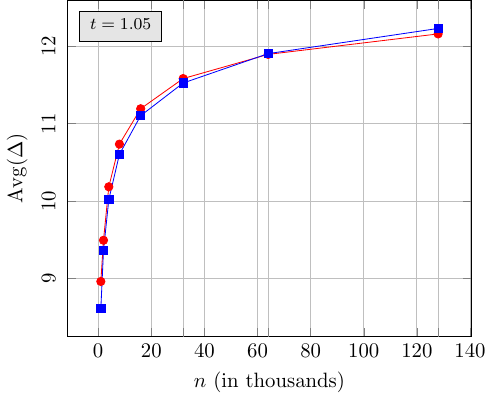}
		\includegraphics[scale=0.85]{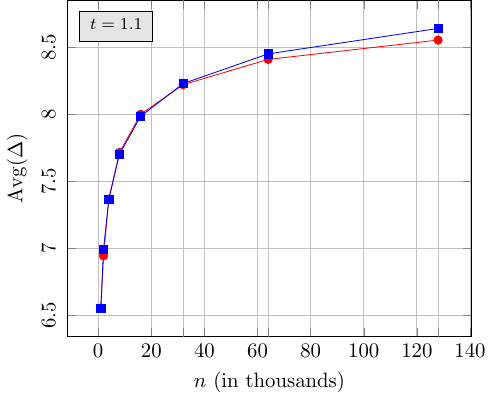}
		\includegraphics[scale=0.85]{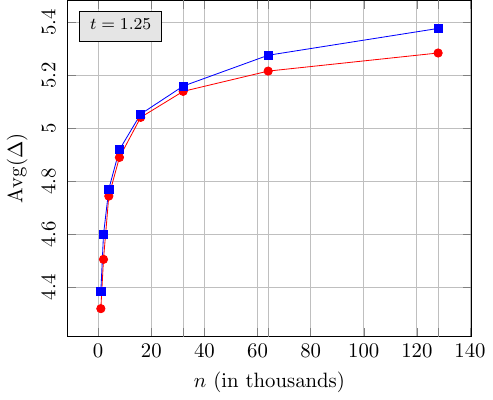}
		\includegraphics[scale=0.85]{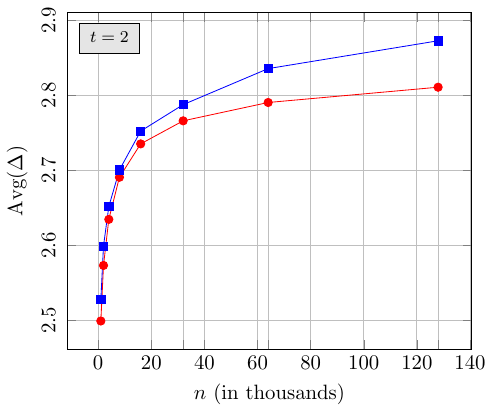}
		\vspace{-10pt}
		\caption{Average-degree comparisons for the \texttt{normal-clustered} distribution.}
		\label{avd:nc}
	\end{figure}
	
	\begin{figure}
		\centering
		\includegraphics[scale=0.85]{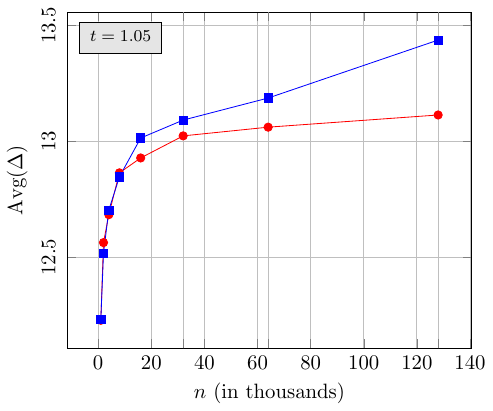}
		\includegraphics[scale=0.85]{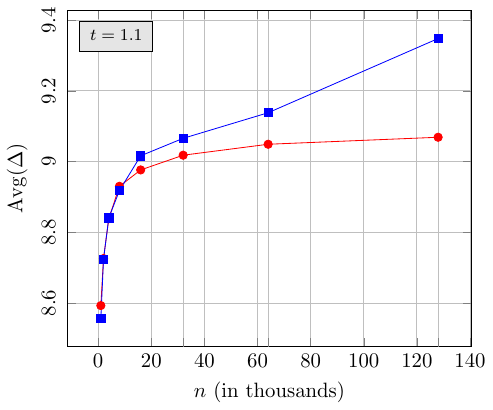}
		\includegraphics[scale=0.85]{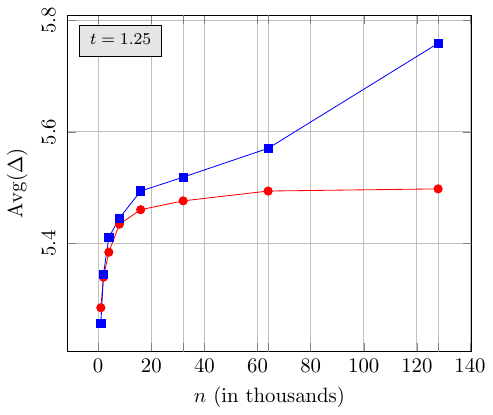}
		\includegraphics[scale=0.85]{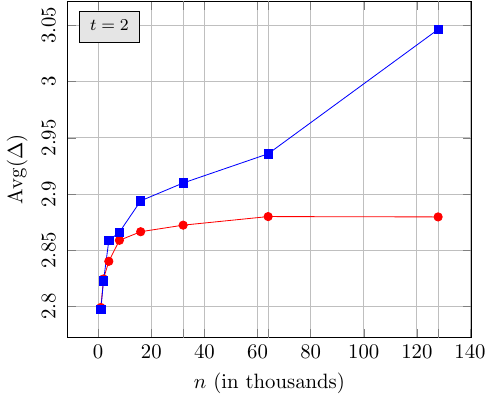}
		\vspace{-10pt}
		\caption{Average-degree comparisons for the \texttt{grid-random} distribution.}
		\label{avd:gr}
	\end{figure}
	
	\begin{figure}
		\centering
		\includegraphics[scale=0.85]{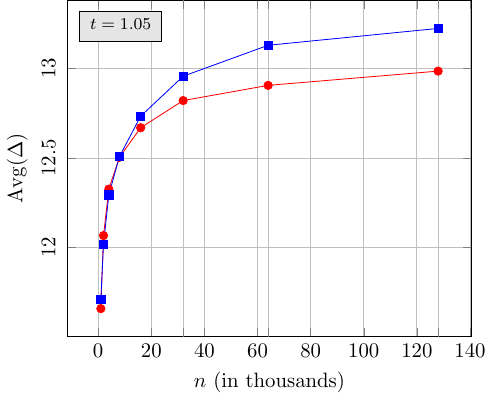}
		\includegraphics[scale=0.85]{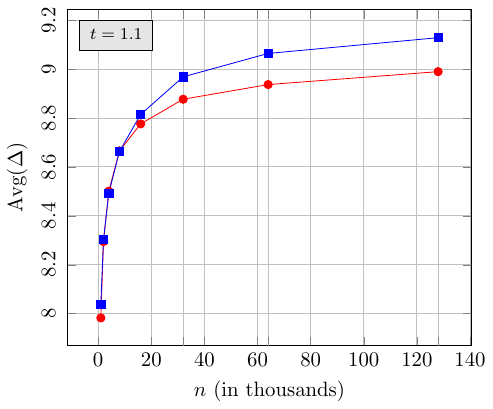}
		\includegraphics[scale=0.85]{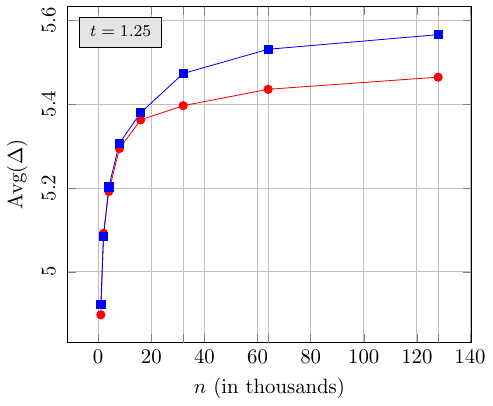}
		\includegraphics[scale=0.85]{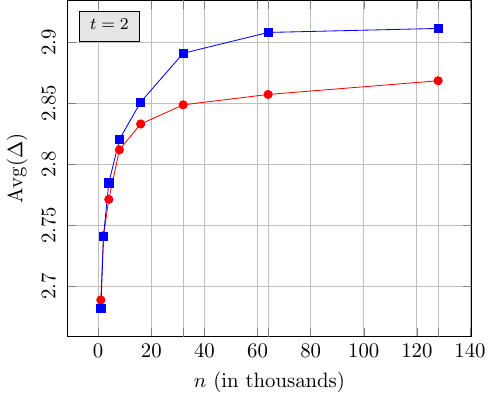}
		\vspace{-10pt}
		\caption{Average-degree comparisons for the \texttt{annulus} distribution.}
		\label{avd:an}
	\end{figure}
	
	\begin{figure}
		\centering
		\includegraphics[scale=0.85]{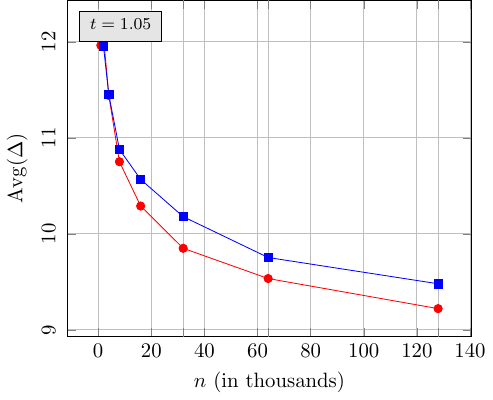}
		\includegraphics[scale=0.85]{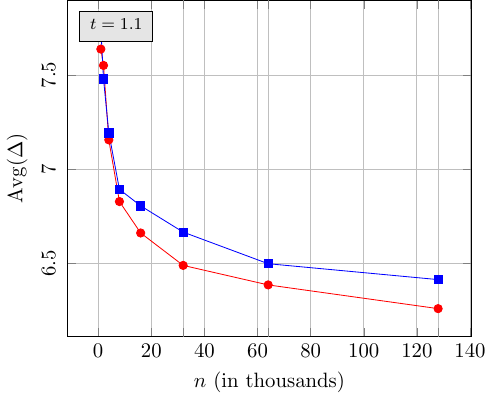}
		\includegraphics[scale=0.85]{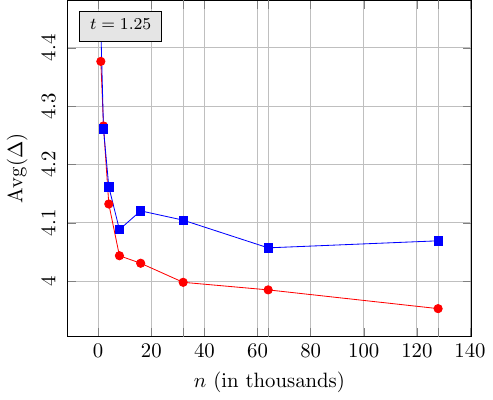}
		\includegraphics[scale=0.85]{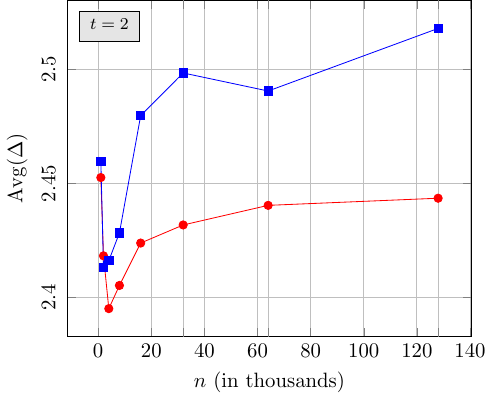}
		\vspace{-10pt}
		\caption{Average-degree comparisons for the \texttt{galaxy} distribution.}
		\label{avd:ga}
	\end{figure}
	
	\begin{figure}
		\centering
		\includegraphics[scale=0.85]{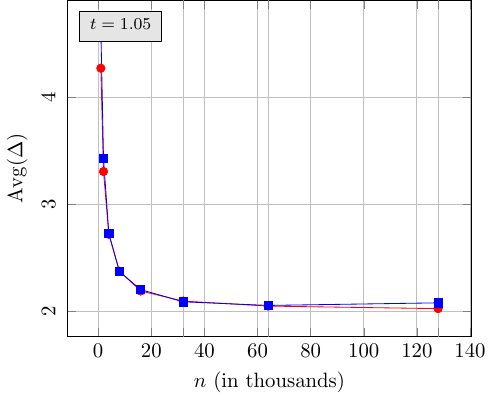}
		\includegraphics[scale=0.85]{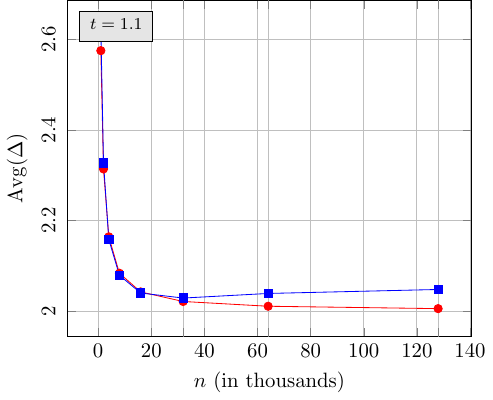}
		\includegraphics[scale=0.85]{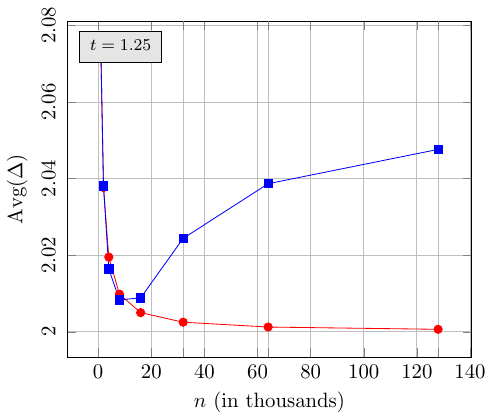}
		\includegraphics[scale=0.85]{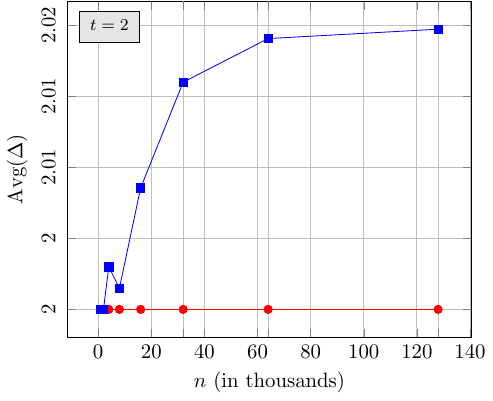}
		\vspace{-10pt}
		\caption{Average-degree comparisons for the \texttt{convex} distribution.}
		\label{avd:co}
	\end{figure}
	
	\begin{figure}
		\centering
		\includegraphics[scale=0.85]{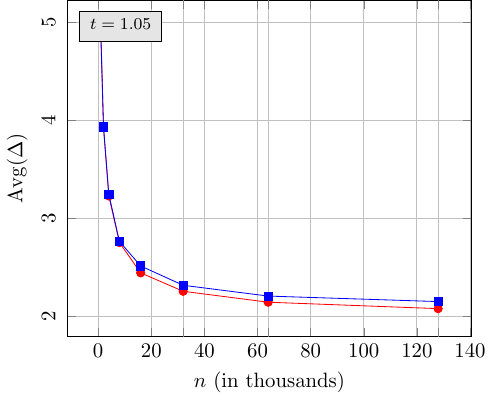}
		\includegraphics[scale=0.85]{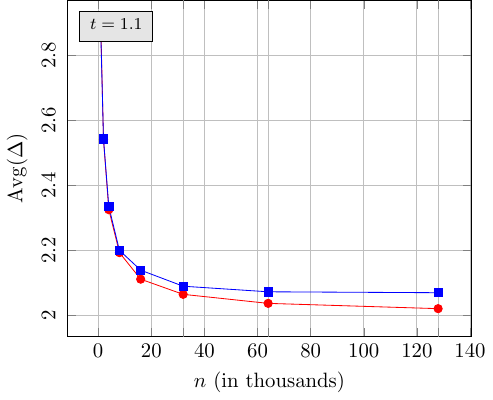}
		\includegraphics[scale=0.85]{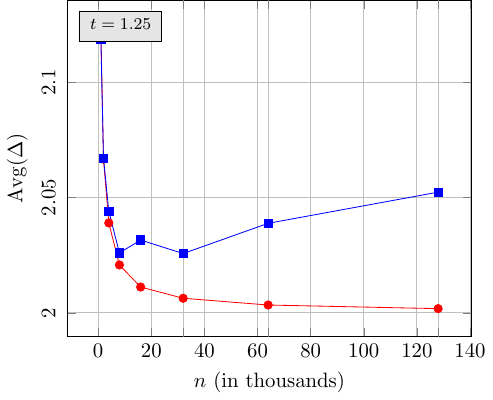}
		\includegraphics[scale=0.85]{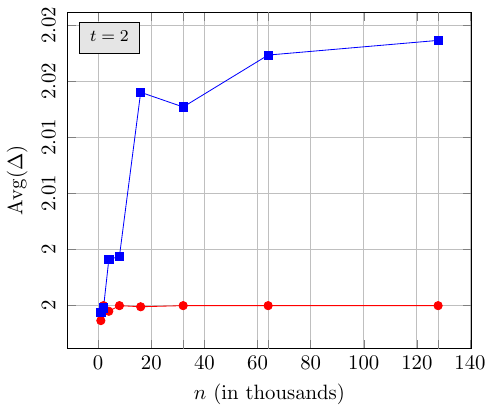}
		\vspace{-10pt}
		\caption{Average-degree comparisons for the \texttt{spokes} distribution.}
		\label{avd:sp}
	\end{figure}
	
		\begin{figure}
		\centering
		\includegraphics[scale=0.85]{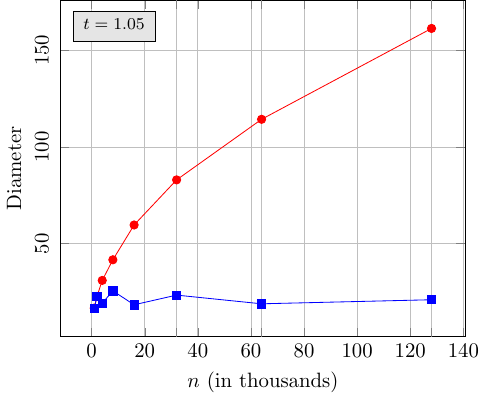}
		\includegraphics[scale=0.85]{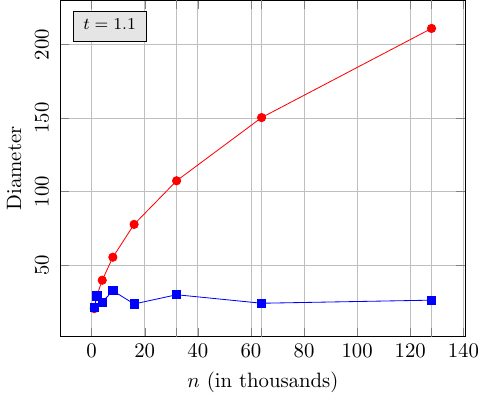}
		\includegraphics[scale=0.85]{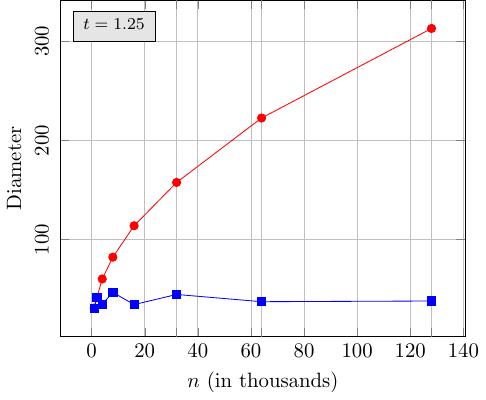}
		\includegraphics[scale=0.85]{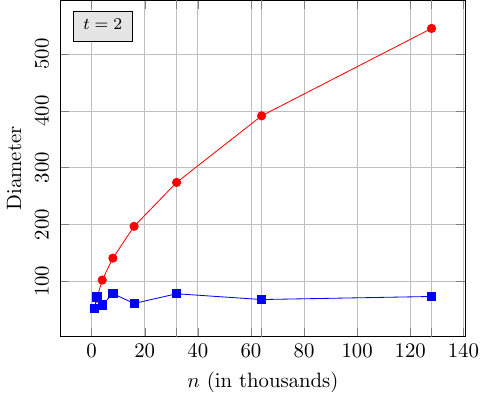}
		\vspace{-10pt}
		\caption{Diameter comparisons for the \texttt{uni-square} distribution.}
		\label{dia:us}
	\end{figure}

	\begin{figure}
		\centering
		\includegraphics[scale=0.85]{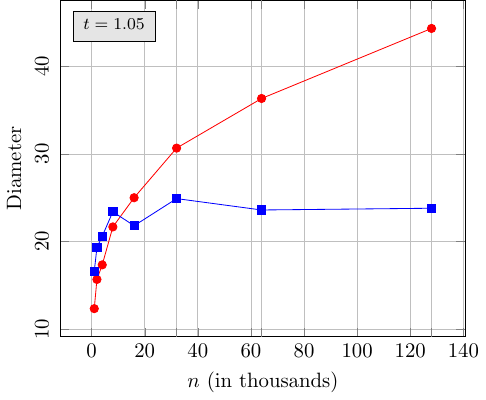}
		\includegraphics[scale=0.85]{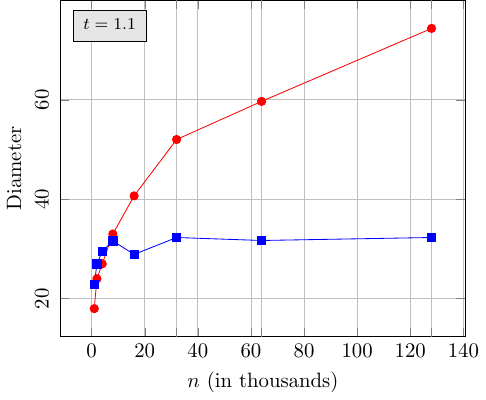}
		\includegraphics[scale=0.85]{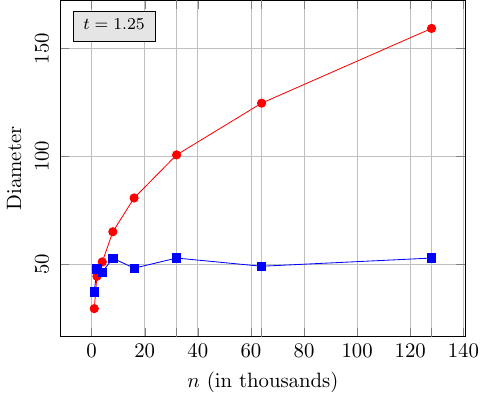}
		\includegraphics[scale=0.85]{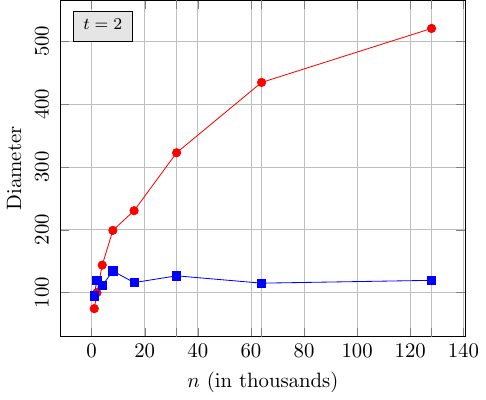}
		\vspace{-10pt}
		\caption{Diameter comparisons for the \texttt{normal-clustered} distribution.}
		\label{dia:nc}
	\end{figure}
	
		\begin{figure}
		\centering
		\includegraphics[scale=0.85]{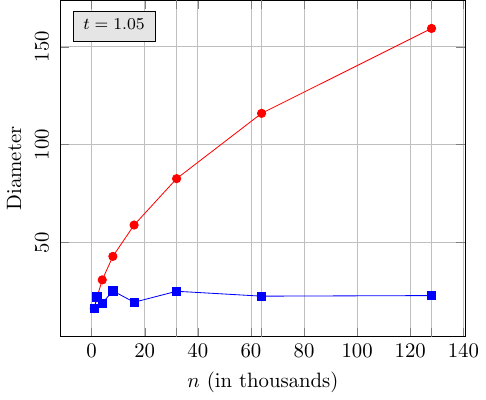}
		\includegraphics[scale=0.85]{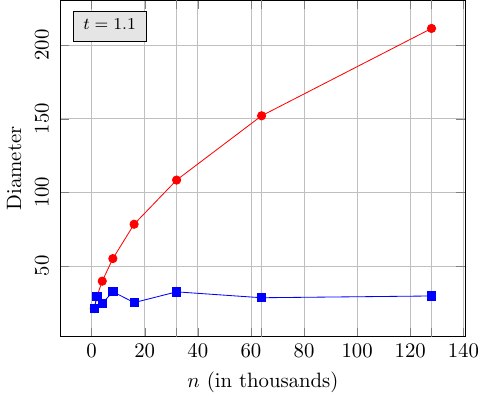}
		\includegraphics[scale=0.85]{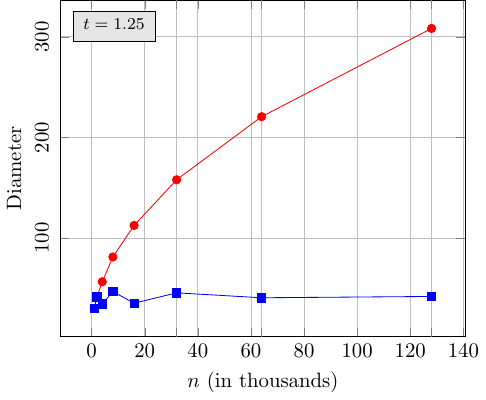}
		\includegraphics[scale=0.85]{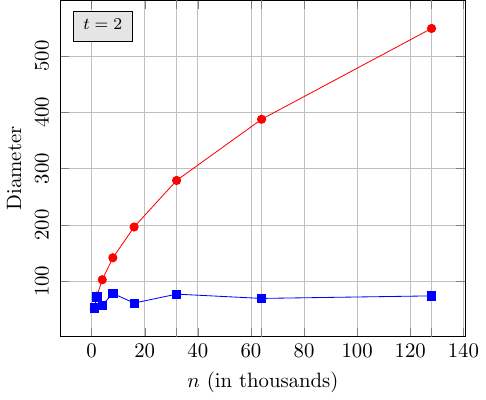}
		\vspace{-10pt}
		\caption{Diameter comparisons for the \texttt{grid-random} distribution.}
		\label{dia:gr}
	\end{figure}
	
			\begin{figure}
		\centering
		\includegraphics[scale=0.85]{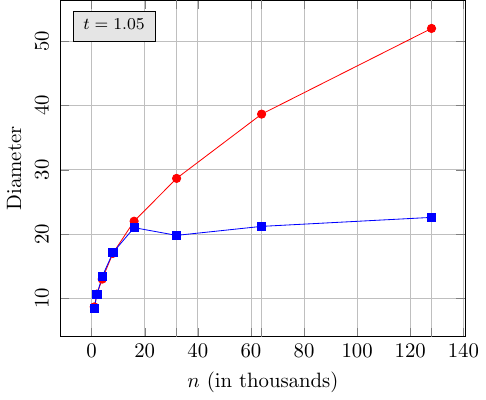}
		\includegraphics[scale=0.85]{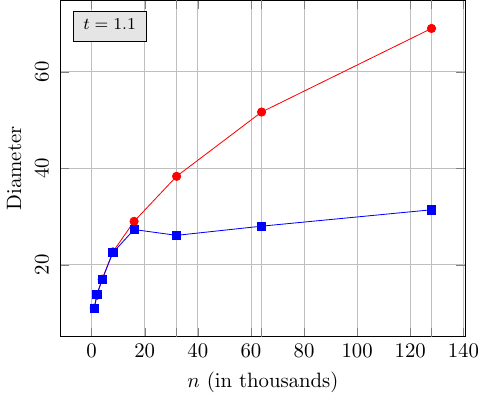}
		\includegraphics[scale=0.85]{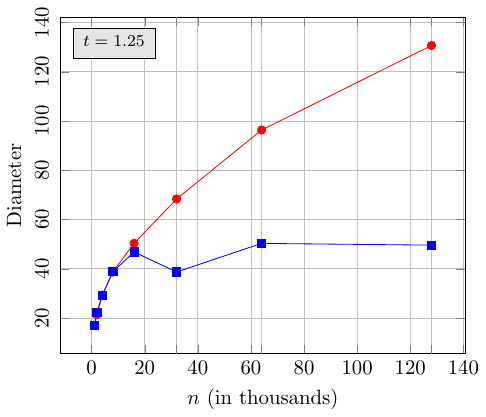}
		\includegraphics[scale=0.85]{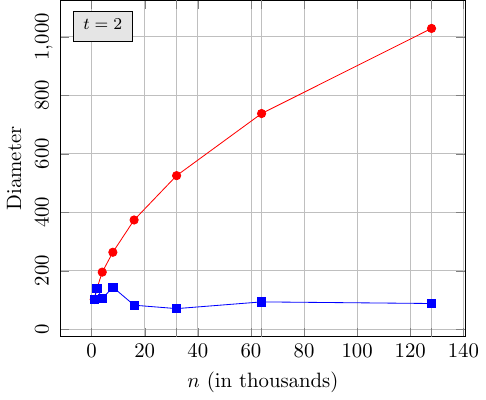}
		\vspace{-10pt}
		\caption{Diameter comparisons for the \texttt{annulus} distribution.}
		\label{dia:an}
	\end{figure}
	
		\begin{figure}
		\centering
		\includegraphics[scale=0.85]{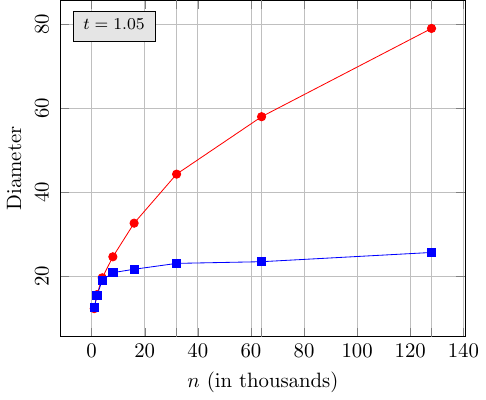}
		\includegraphics[scale=0.85]{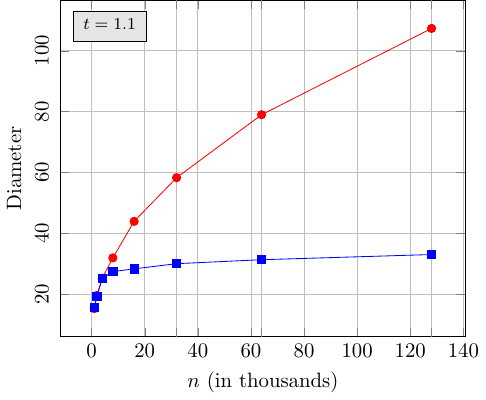}
		\includegraphics[scale=0.85]{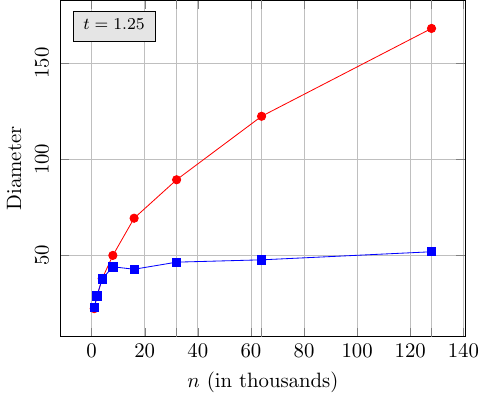}
		\includegraphics[scale=0.85]{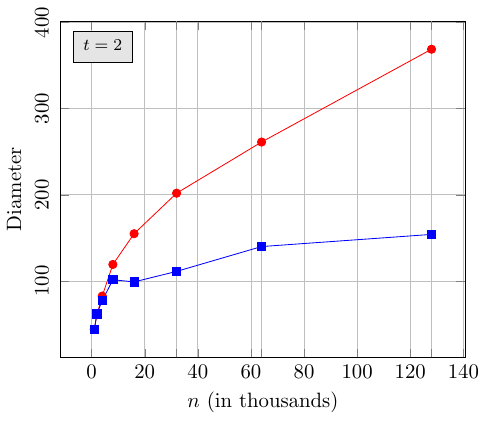}
		\vspace{-10pt}
		\caption{Diameter comparisons for the \texttt{galaxy} distribution.}
		\label{dia:ga}
	\end{figure}
	
			\begin{figure}
		\centering
		\includegraphics[scale=0.85]{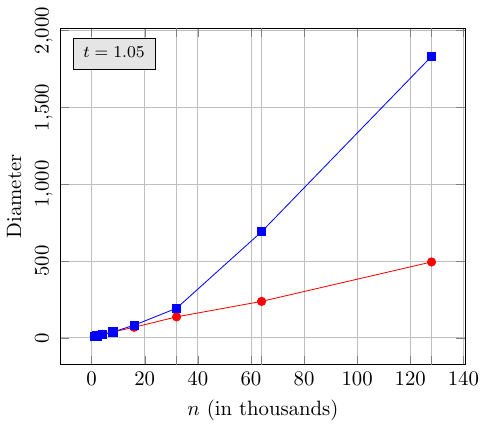}
		\includegraphics[scale=0.85]{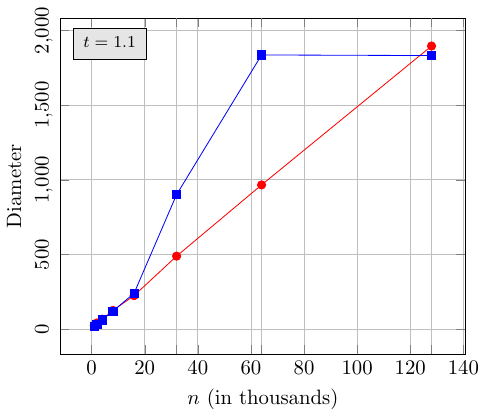}
		\includegraphics[scale=0.85]{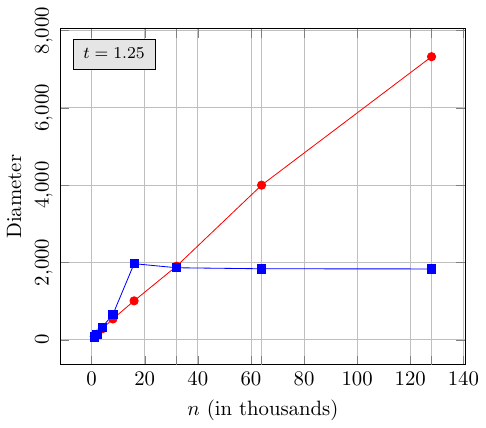}
		\includegraphics[scale=0.85]{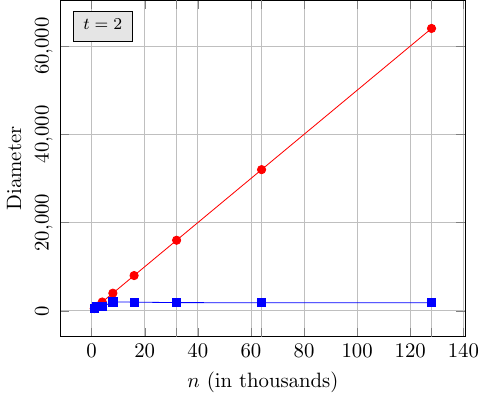}
		\vspace{-10pt}
		\caption{Diameter comparisons for the \texttt{convex} distribution.}
		\label{dia:co}
	\end{figure}
	
				\begin{figure}
		\centering
		\includegraphics[scale=0.85]{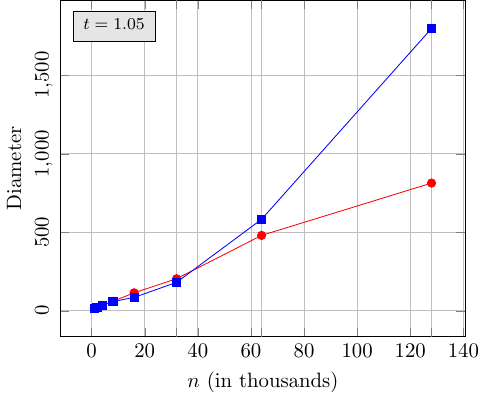}
		\includegraphics[scale=0.85]{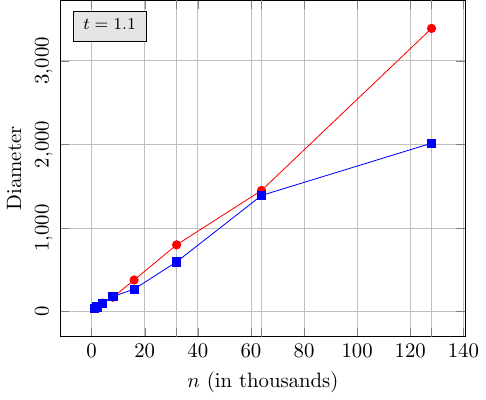}
		\includegraphics[scale=0.85]{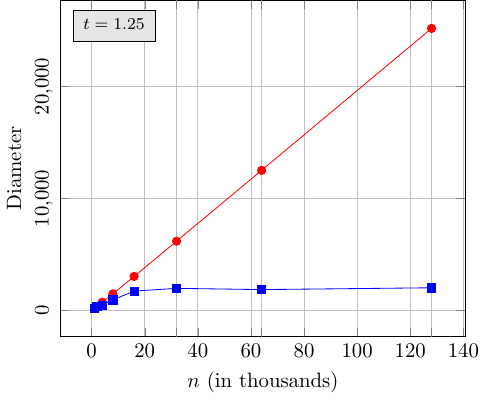}
		\includegraphics[scale=0.85]{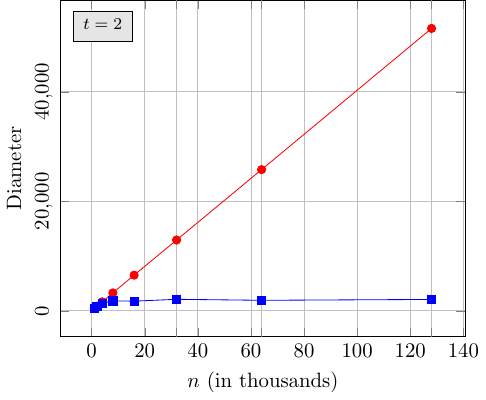}
		\vspace{-10pt}
		\caption{Diameter comparisons for the \texttt{spokes} distribution.}
		\label{dia:sp}
	\end{figure}
	\begin{figure}[ht]
		\centering
		\includegraphics[scale=0.85]{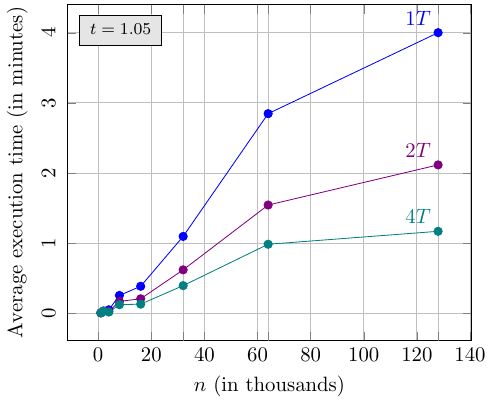}
		\includegraphics[scale=0.85]{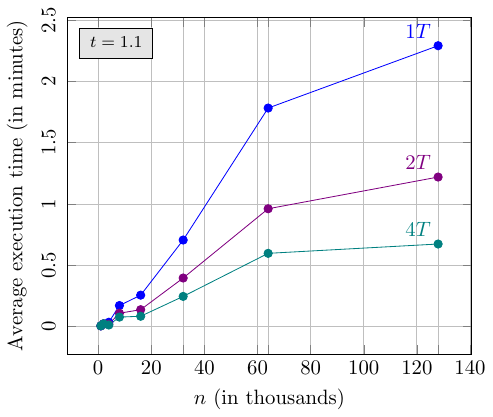}
		\includegraphics[scale=0.85]{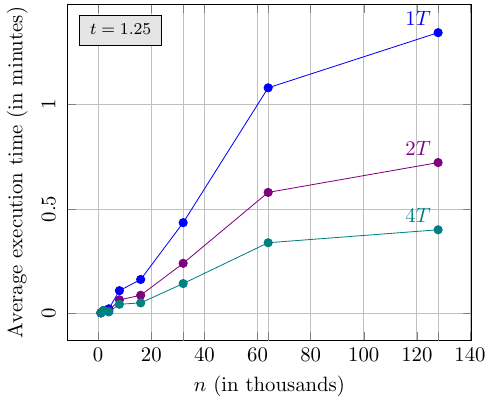}
		\includegraphics[scale=0.85]{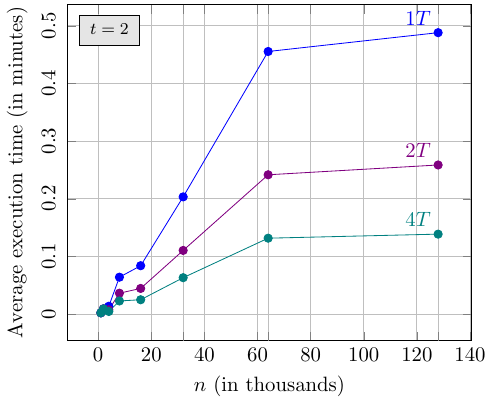}
		\vspace{-10pt}
		\caption{Multithreaded runtimes for the  \texttt{uni-square} distribution; $T$ stands for thread.}
		\label{multithread:us}
	\end{figure}
	
	\begin{figure}[ht]
	\centering
	\includegraphics[scale=0.85]{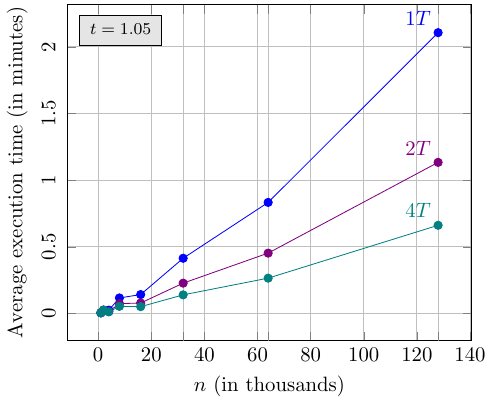}
	\includegraphics[scale=0.85]{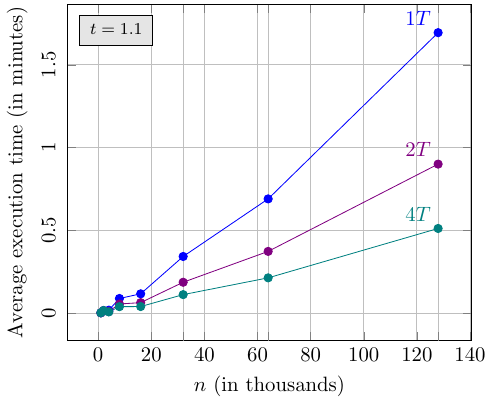}
	\includegraphics[scale=0.85]{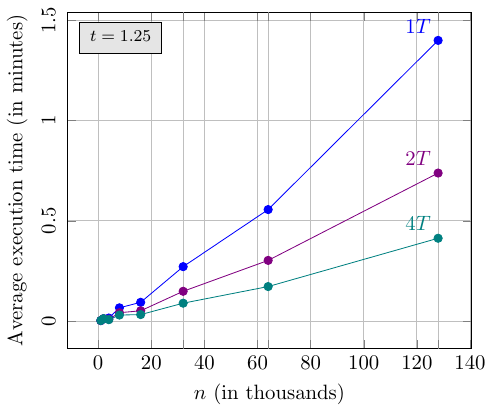}
	\includegraphics[scale=0.85]{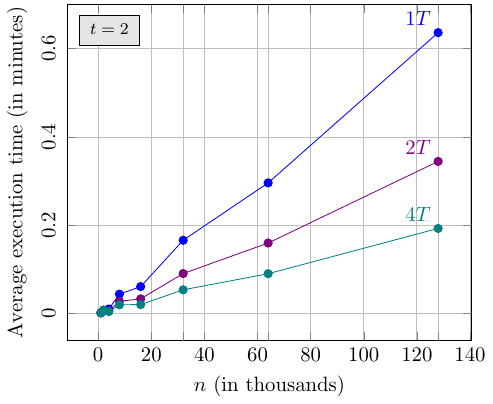}
	\vspace{-10pt}
	\caption{Multithreaded runtimes for the  \texttt{normal-clustered} distribution; $T$ stands for thread.}
	\label{multithread:nc}
\end{figure}

\begin{figure}[ht]
\centering
\includegraphics[scale=0.85]{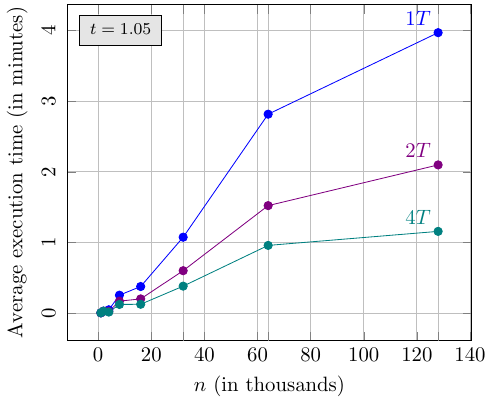}
\includegraphics[scale=0.85]{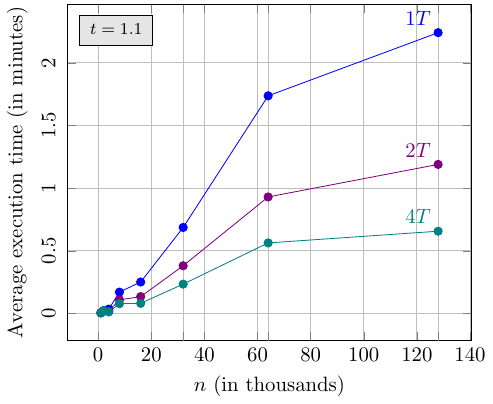}
\includegraphics[scale=0.85]{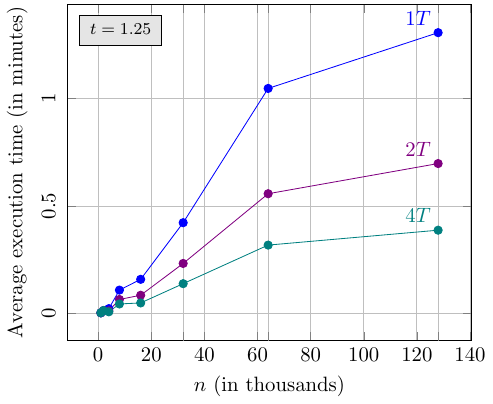}
\includegraphics[scale=0.85]{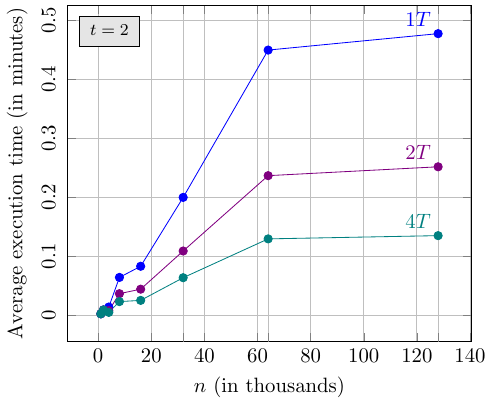}
\vspace{-10pt}
\caption{Multithreaded runtimes for the  \texttt{grid-random} distribution; $T$ stands for thread.}
\label{multithread:gr}
\end{figure}

\begin{figure}[ht]
	\centering
	\includegraphics[scale=0.85]{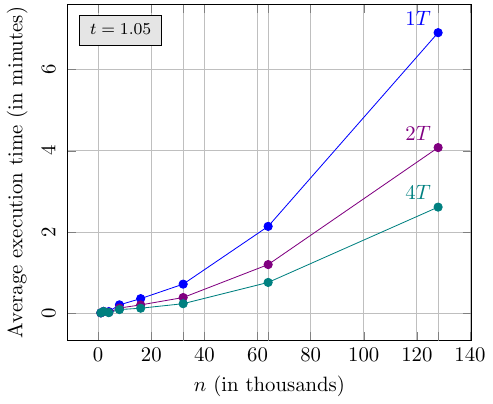}
	\includegraphics[scale=0.85]{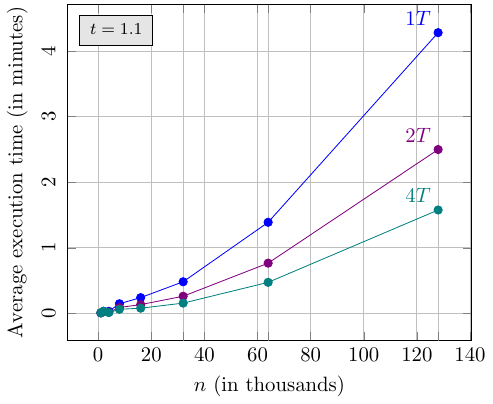}
	\includegraphics[scale=0.85]{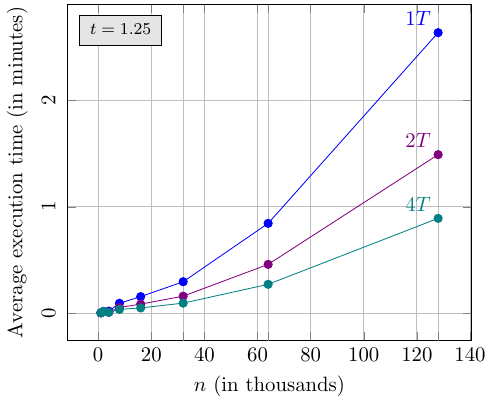}
	\includegraphics[scale=0.85]{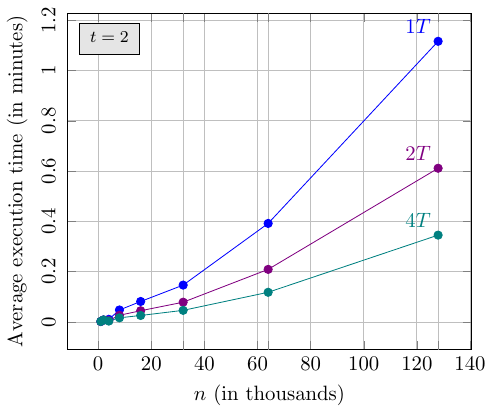}
	\vspace{-10pt}
	\caption{Multithreaded runtimes for the  \texttt{annulus} distribution; $T$ stands for thread.}
	\label{multithread:ann}
\end{figure}

\newpage
\clearpage

\begin{figure}[ht]

	\centering
	\includegraphics[scale=0.85]{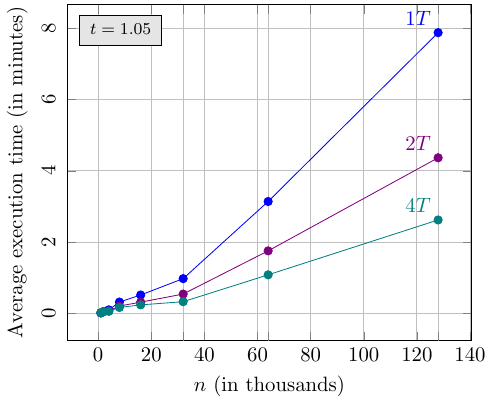}
	\includegraphics[scale=0.85]{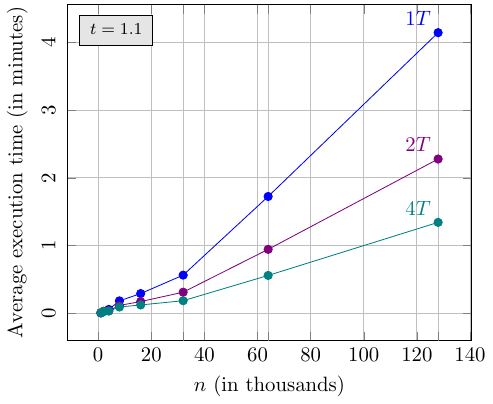}
	\includegraphics[scale=0.85]{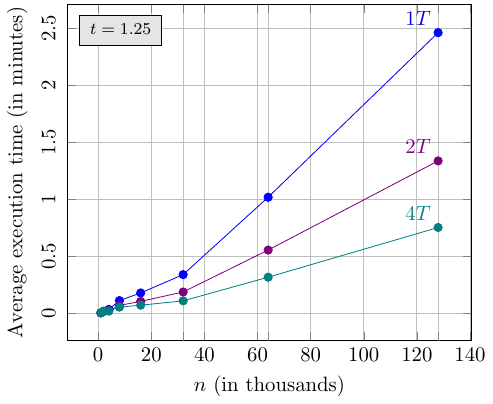}
	\includegraphics[scale=0.85]{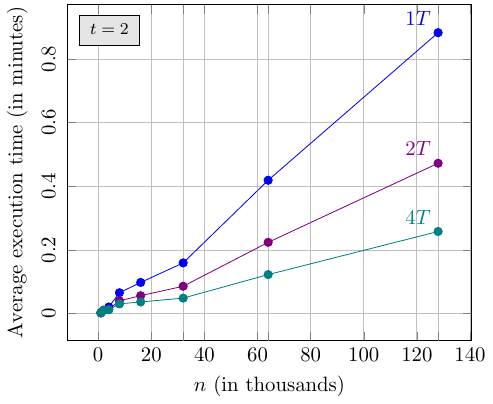}
	\vspace{-10pt}
	\caption{Multithreaded runtimes for the  \texttt{galaxy} distribution; $T$ stands for thread.}
	\label{multithread:ga}
\end{figure}

\begin{figure}[ht]

	\centering
	\includegraphics[scale=0.85]{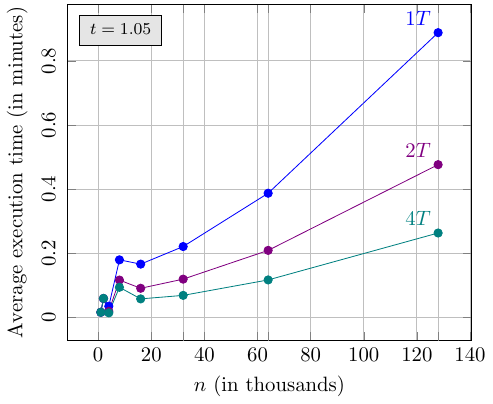}
	\includegraphics[scale=0.85]{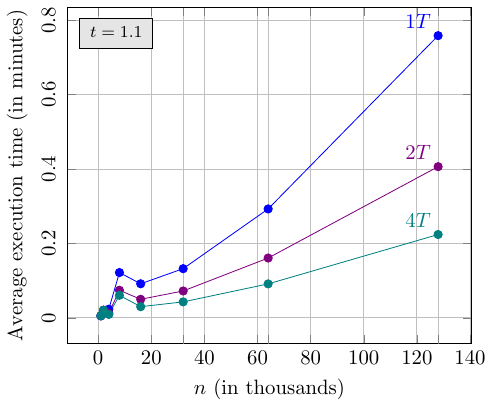}
	\includegraphics[scale=0.85]{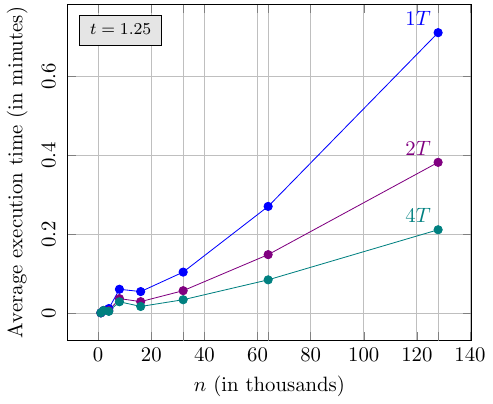}
	\includegraphics[scale=0.85]{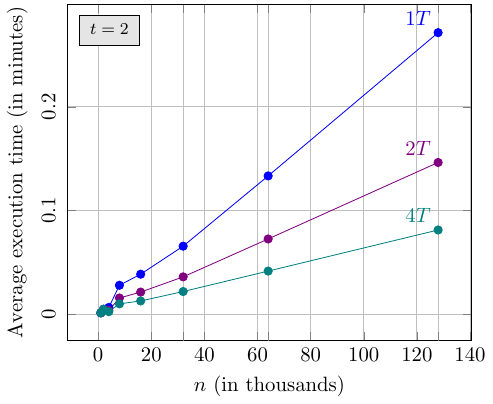}
	\vspace{-10pt}
	\caption{Multithreaded runtimes for the  \texttt{convex} distribution; $T$ stands for thread.}
	\label{multithread:co}
\end{figure}

\begin{figure}[ht]

	\centering
	\includegraphics[scale=0.85]{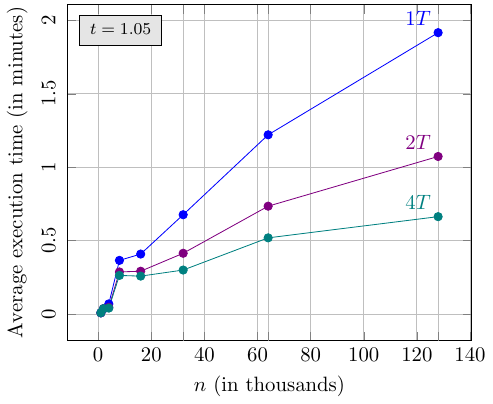}
	\includegraphics[scale=0.85]{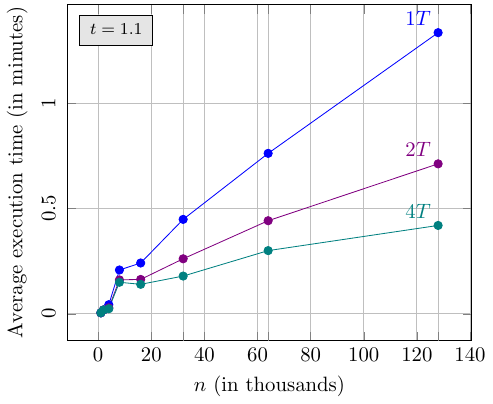}
	\includegraphics[scale=0.85]{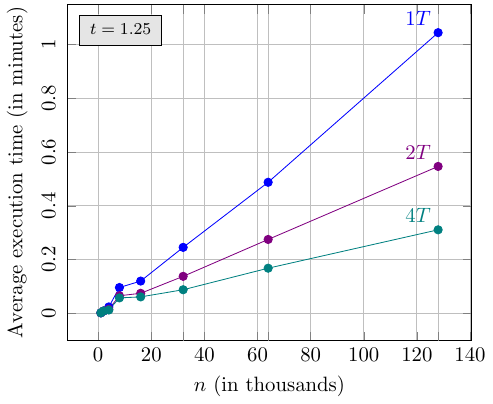}
	\includegraphics[scale=0.85]{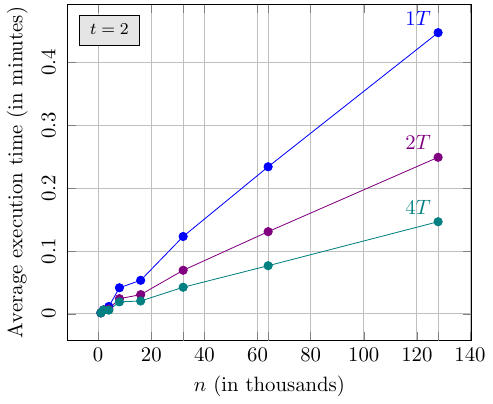}
	\vspace{-10pt}
	\caption{Multithreaded runtimes for the  \texttt{spokes} distribution; $T$ stands for thread.}
	\label{multithread:sp}
\end{figure}


\section{Why is our algorithm fast in practice?}
\label{sec:argue}


Section~\ref{sec:exp} provides evidence that
\textsc{Fast-Sparse-Spanner} behaves like a near-linear-time
algorithm in practice, suggesting that the average-case time
complexity is nearly linear for most distributions.
The running time of \textsc{Fast-Sparse-Spanner} (see
Theorem~\ref{thm:time}) is $O(n\log n + n (d^2+f(n) ))$, where
$O(n\log n)$ is the time taken to construct the WSPD-spanner $W$, $d$
is the depth of the quad tree $T$, and $f(n)$ is the time taken to
compute a $t$-path between two vertices in $H$. In the worst-case
scenario, $f(n) = O(n\log n)$ since $A^*$ is a modified version of
Dijkstra, thus making the overall worst-case time complexity to be
$O(n^2\log n)$ since there are $O(n)$ mergings
(Corollary~\ref{cor:mergings}).
We argue that in the case of uniform distributions (e.g.,
\texttt{uni-square}), $A^*$ is likely to be fast. Our intuition is
based on the observation that neighboring leaves of the quad-tree are
likely to be of almost the same size for the uniform distribution,
making the dual graph of small degree, thus making the number of
vertices visited by $A^*$ within $h$ hops to be small. Thus, on the
average, we expect $f(n) = O(1)$. We conjecture that the average-case
time complexity of \textsc{Fast-Sparse-Spanner} is $O(n\log n)$ with
linear additive terms. 

Consistent with out intuition above, 
in our experiments, we observed that the construction of $W$ takes a negligible fraction of the total runtime. For instance, on a $1M$-element \texttt{uni-square} pointset and $t=1.1$, the construction of $W$ took just $35$ milliseconds, and the whole spanner construction took around $55$ minutes (Fig~\ref{fig:1million}). 
Next, we always found that $d$ was never more than $20$ (quad-trees for \texttt{galaxy} pointsets had the highest depths in our experiments). As a result, the term $d^2n$ tends to be linear in practice. Further, we found that both $A^*$ and \textsc{Greedy-Path} always explored a low number of vertices in $H$ (much less than $n$). For the same pointset, on average, \textsc{Greedy-Path} explored approximately $35$ and $417$ vertices in steps $4$ and $5$, respectively. $A^*$ explored approximately $38$ and $2210$ vertices on average in the two steps, respectively. The early terminations were possible owing to the long WSPD edges placed by the algorithm. Further, we noticed that \textsc{Greedy-Path} was successful $\approx 83.75\%$ of the times in finding $t$-paths. As a result, the number of $A^*$ calls was much less compared to the number of \textsc{Greedy-path} calls. 
It resulted in fast $t$-path computations inside \textsc{Greedy-Merge} and \textsc{Greedy-Merge-Light}. Consequently, $f(n)$ was slow-growing in practice. The above observations account for the linear runtime behavior of \textsc{Fast-Sparse-Spanner}.

\section{Conclusions}

Our experiments show that \textsc{Fast-Sparse-Spanner} is remarkably
faster than \textsc{Bucketing} and uses substantially less memory in
most cases. The spanners produced by  \textsc{Fast-Sparse-Spanner} were found to have near-greedy average-degree, and low diameter. 
Surprisingly, in our rigorous
testing, we found that the produced spanners always had the desired stretch-factors. It would be interesting to extend the presented ideas to higher dimensions for future work. 

\textit{Acknowledgment.} We are grateful to Kevin Buchin, one of the authors of the \textsc{Bucketing} algorithm~\cite{alewijnse2017distribution}, for generously sharing their \textsc{Bucketing} code with us and communicating over email.

	\bibliographystyle{acm}
	\bibliography{main}

\end{document}